\documentclass[a4paper,twocolumn,11pt,accepted=2025-03-31]{quantumarticle}
\pdfoutput=1
\usepackage[utf8]{inputenc}
\usepackage[english]{babel}
\usepackage[T1]{fontenc}
\usepackage{amsmath}
\usepackage{amsthm}
\usepackage{amssymb}
\usepackage{enumitem}
\usepackage{graphicx}
\usepackage{subfigure}
\usepackage{csquotes}
\usepackage{hyperref}
\usepackage[capitalize,nameinlink,noabbrev]{cleveref}
\usepackage{booktabs}
\usepackage{array}
\usepackage{lipsum}
\usepackage{tikz}
\usepackage[ruled,vlined,linesnumbered,longend,algopart]{algorithm2e}
\usepackage[numbers,sort&compress]{natbib}
\usepackage{listings}
\usepackage[resetlabels]{multibib}

\newcites{article}{article references}
\newcites{book}{book references}
\newcites{misc}{misc references}
\newcites{repo}{repository references}
\newcites{web}{website references}
\newcites{other}{Other references}

\newcolumntype{C}{@{}>{\hfil$}p{0.7cm}<{$\hfil}@{}}
\usetikzlibrary{backgrounds}
\usetikzlibrary{arrows}
\usetikzlibrary{positioning,quotes}
\newcommand\tikznode[3][]%
   {\tikz[remember picture,baseline=(#2.base)]
      \node[minimum size=0pt,inner sep=0pt,outer sep=0pt,#1](#2){#3};%
   }
\tikzset{>=stealth}

\usepackage{colortbl}
\newcommand\cl{\cellcolor{clight1}}
\definecolor{clight1}{RGB}{212, 237, 244}%
\newcommand\cs{\cellcolor{clight2}}
\definecolor{clight2}{RGB}{255,221,225}
\newcommand\ce{\cellcolor{clight3}}
\definecolor{clight3}{HTML}{C5E1A5}
\definecolor{mygreen}{RGB}{116,183,33}

\newtheorem{theo}{Theorem}

\newtheorem{lemma}{Lemma}
\newtheorem{exam}{Example}
\newtheorem{defi}{Definition}
\newtheorem{rem}{Remark}
\newtheorem{cons}{Construction}

\Crefname{theo}{Theorem}{Theorems}
\Crefname{exam}{Example}{Examples}
\Crefname{cons}{Construction}{Construction}

\begin{document}

\title{Spatially-Coupled QLDPC Codes}
\author{Siyi Yang}
\affiliation{Duke Quantum Center, Duke University, Durham, NC 27708, USA}
\email{siyi.yang@duke.edu}
\orcid{0000-0001-7512-1913}
\author{Robert Calderbank}
\affiliation{Duke Quantum Center, Duke University, Durham, NC 27708, USA}
\email{robert.calderbank@duke.edu}
\orcid{0000-0003-2084-9717}

\maketitle

\begin{abstract}
Spatially-coupled (SC) codes is a class of convolutional LDPC codes that has been well investigated in classical coding theory thanks to their high performance and compatibility with low-latency decoders. We describe toric codes as quantum counterparts of classical two-dimensional spatially-coupled (2D-SC) codes, and introduce spatially-coupled quantum LDPC (SC-QLDPC) codes as a generalization. We use the convolutional structure to represent the parity check matrix of a 2D-SC code as a polynomial in two indeterminates, and derive an algebraic condition that is both necessary and sufficient for a 2D-SC code to be a stabilizer code. This algebraic framework facilitates the construction of new code families. While not the focus of this paper, we note that small memory facilitates physical connectivity of qubits, and it enables local encoding and low-latency windowed decoding. In this paper, we use the algebraic framework to optimize short cycles in the Tanner graph of 2D-SC hypergraph product (HGP) codes that arise from short cycles in either component code. While prior work focuses on QLDPC codes with rate less than 1/10, we construct 2D-SC HGP codes with small memories, higher rates (about 1/3), and superior thresholds.
\end{abstract}

\section{Introduction}
\label{sec: introduction}

Quantum computers promise to be more capable of solving certain problems than any classical computer, and the theory of fault-tolerant quantum computation describes how a quantum computer with an arbitrarily low error rate can be built from faulty parts. Quantum error correction is an essential building block, and the class of stabilizer codes has been studied extensively \cite{calderbank1997quantum,calderbank1998quantum,gottesman1997stabilizer}. A stabilizer code is the fixed subspace of a commutative subgroup of the Pauli group, and a CSS (Calderbank-Shor-Steane) code is a particular type of stabilizer code where the generators can be separated into strictly X-type and strictly Z-type operators \cite{calderbank1997quantum,calderbank1998quantum}. 

Geometric locality of stabilizers is important in superconducting qubit platforms, and this design principle has encouraged the development of topological codes. Kitaev \cite{kitaev1997quantum,kitaev2003fault} introduced toric codes, a class of topological codes defined on a two-dimensional spin lattice. Bravyi et al. \cite{bravyi1998quantum} modified these codes to allow physical layout of qubits in the Euclidean plane. Surface codes \cite{dennis2002topological,freedman2002z2,bravyi2010tradeoffs,fowler2012surface,horsman2012surface} are a broad class of topological codes defined by tilings over a closed surface, where holes in the surface correspond to logical operators, and the genus determines code rate. Delfosse and Breuckmann \cite{delfosse2010quantum,terhal2015quantum,delfosse2016generalized,breuckmann2017hyperbolic} introduced hyperbolic and semi-hyperbolic surface codes that use tilings of a closed hyperbolic surface to improve the overhead and logical error rate of toric codes.

Surface codes have very low rates, and quantum LDPC (QLDPC) codes enable higher rates while preserving locality of stabilizers. Tillich and Zemor \cite{tillich2014hypergraph} introduced hypergraph product codes (the hypergraph of two classical LDPC codes), after recognizing the Tanner graph of a surface code as the graph product of two repetition codes. The minimum distance of a hypergraph product code grows with the square root of the block length. Leverrier \cite{leverrier2022xyz} has developed XYZ codes as a three-dimensional extension of hypergraph product codes. Hypergraph product codes can be viewed as a special case of codes subsequently constructed from more general ``products'' of component codes.

There has been a great deal of research on constructing codes that achieve a target scaling of minimum distance $D$ and number of encoded bits $K$ with respect to the block length $N$. Constant rate QLDPC codes with $D\in O(N^{1/2}\log N)$ had been best possible for more than a decade \cite{freedman2002z2,evra2022decodable}, until CSS codes constructed from different homological products of classical codes are brought to our attention. The classical codes are either good random LDPC codes, or expander codes \cite{sipser1996expander} constructed from Cayley graphs of projective linear group or hyperbolic tessellations over closed surfaces \cite{zemor2009cayley,breuckmann2017hyperbolic}. Fibre bundle codes proposed by Hastings, Haah, and O’Donnell \cite{hastings2021fiber} are the first codes that improve the scaling exponents to into $K\in O\left(N^{3/5}\right)$ and $D\in O\left(N^{3/5}/ \textrm{polylog}(N)\right)$. Panteleev and Kalachev \cite{panteleev2022liftedproduct} demonstrated that a scaling of $K\in O(N^{\alpha}\log N)$, $D\in O(N^{1-\alpha/2}/\log N)$ is achievable for arbitrary $0\leq \alpha <1$ through the quasi-abelian lifted product codes. Breuckmann and Eberhardt proposed the balanced product codes and improved the scaling exponents to $K\in O(N^{\frac{4}{5}})$ and $D\in O(N^{\frac{3}{5}})$ \cite{breuckmann2021balanced}.

Whether constant rates and linear minimum distances could be simultaneously achieved through tensor products of two expander codes defined over non-commutative group rings have remained an open problem \cite{panteleev2022liftedproduct,breuckmann2021balanced}. Panteleev and Kalachev completely proved the existence of constant rate codes with linear distance scaling in \cite{panteleev2022asymptotically}, referred to as expander lifted product codes, obtained from tensor product of two Tanner codes with different local codes defined over the same graph. The local codes are chosen such that their classical product attains good local high-dimensional expansion property, motivated by the idea of high-dimensional expanders proposed in \cite{evra2022decodable}. Concurrently, Leverrier and Zemor \cite{zemor2022quantumtanner} introduced a simplification of the lifted product codes from \cite{panteleev2022asymptotically} called quantum Tanner codes using a left-right Cayley complex, which can be viewed as the balanced product of Cayley graphs or their double covers.

The recent blossoming of research on QLDPC codes revolves around constructions that guarantee minimum distance. However, classical LDPC codes are not designed to maximize minimum distance, and it does not matter if there is a modest number of low-weight codewords, as long as the noise is statistically unlikely to take the signal in a bad direction. Classical LDPC codes have come to dominate coding practice by focusing on typical errors in graph-based codes rather than on the worst-case errors (minimum distance) in algebraic codes. LDPC codes have succeeded by combining local encoding with iterative decoding. Gallager \cite{gallager1962low} first suggested using the adjacency matrix of a randomly chosen low-degree bipartite graph as a parity-check matrix. He also introduced belief propagation (BP) decoders, where the task of globally estimating the joint probability of errors (maximum likelihood decoding) is decomposed into parallel local approximations of marginal distributions at each bit (see \cite{richardson2001design,richardson2001capacity} for a comprehensive analysis of BP decoding). 

While BP decoders have been introduced for QLDPC codes \cite{poulin2008iterative,lidar2013quantum}, subgraphs of the Tanner graph that dominate decoding errors (objects) have not received much attention. Finite length optimization of QLDPC codes have also been overlooked, compared with the rich literature of removing detrimental objects in classical LDPC codes \cite{mitchell2015spatially,mitchell2017edge,beemer2017generalized,beemer2018design,hareedy2016general,hareedy2017high,hareedy2020channel,Yang2023breaking}. Motivated by the resemblance between toric codes and two dimensional spatially-coupled (SC) codes, we develop SC-QLDPC codes as the quantum analogue of the SC codes in classical world. The SC-QLDPC codes we consider here differ from the non-tail-biting codes introduced by Hagiwara \cite{hagiwara2011spatially}. We consider tail-biting codes to make it easier to increase minimum distance with short component codes. We consider time-invariant codes, since they admit a compact algebraic representation that facilitates systematic optimization. The structure of SC codes is consistent with the locality desired by quantum hardware and it enables systematic optimization with respect to the influence of short cycles. We now highlight our main contributions:

\textbf{Spatially-coupled QLDPC Codes}: We describe toric codes as quantum counterparts of classical two-dimensional spatially-coupled (2D-SC) codes, and introduce spatially-coupled QLDPC codes as a generalization. The parity check matrix of a SC-QLDPC code is assembled into component matrices, which are vertically concatenated into replicas, and these replicas are coupled in a convolutional framework (see \cite{lentmaier2010iterative}) The memory of the code is determined by the number of component matrices in a replica, and small memory simplifies physical connectivity of qubits, enables local encoding, and low-latency windowed decoding (see \cite{iyengar2011windowed,iyengar2012windowed,hassan2016non,wei2016design,zhu2017braided,klaiber2018avoiding,esfahanizadeh2020multi,tauz2020non,ram2022decoding}, for more information about windowed decoding in the classical world).

\textbf{Characteristic Functions}: We use the language of two-dimensional convolution to represent the parity check matrix of a 2D-SC code as polynomial $\mathbf{F}(U, V)$ in two indeterminates U, V, where the coefficients are matrices of Pauli matrices. We derive a simple algebraic condition that is necessary and sufficient for a 2D-SC code to be a stabilizer code, and our algebraic framework facilitates the construction of new code families. Characteristic functions are similar in spirit to the Laurent polynomial representation of translation-invariant codes and to the generator polynomials of quantum convolutional codes \cite{haah2013commuting,haah2016algebraic,ollivier2003description}. We observe that the quasi-abelian lifted product codes proposed in \cite{panteleev2022liftedproduct} can be viewed as high-dimensional spatially-coupled hypergraph product (HGP) codes. Our algebraic framework provides insight into short cycles in the Tanner graph that arise from short cycles in either component codes. 

\textbf{Optimization Framework}: Short cycles in Tanner graphs create problems for belief propagation (BP) decoders in both the waterfall and error floor regions (see \cite{mitchell2015spatially,mitchell2017edge,beemer2017generalized,beemer2018design,hareedy2016general,hareedy2017high,hareedy2020channel,Yang2023breaking} for cycle optimization in the classical world). The fact that quasi-abelian lifted product codes can be viewed as finite-dimensional SC-HGP codes implies that they can be optimized through methods borrowed from classical world with respect to the number of short cycles, and high performance lifted product codes can be constructed even when the lifting exponents are restricted to a small set (corresponding to low memories). As an example, we provide an optimization framework for 2D-SC-HGP codes to minimize the number of short cycles in the Tanner graph that arise from short cycles in either component code, in the same way that generator functions facilitate analysis of the weight distributions of classical convolutional codes. There are short cycles in the Tanner graph that do not involve short cycles in the component codes, and these rigid cycles lead to a fundamental limit on performance (the rigid cycle limit). Simulation results on the quantum depolarizing channel demonstrate that our optimization method provides codes with small memory that approach the rigid cycle limit.

\textbf{High Rate Quantum LDPC Codes}: Research attention has focused on constructing codes that achieve linear scaling of minimum distance and block length and on performance evaluation of low rate codes, such as surface codes. SC-QLDPC codes combine higher rates ($1/3$ rather than $<1/10$) with high performance. 

The rest of the paper is organized as follows. In \Cref{sec: preliminaries}, we first introduce classical SC codes, then, after reviewing stabilizer codes, we describe toric codes as 2D-SC codes. In \Cref{sec: quantum 2d-sc codes}, we generalize to SC-QLDPC codes, and provide a necessary and sufficient condition for stabilizers to commute. We present two families of SC-QLDPC codes constructed from hypergraph product codes and XYZ codes, and connected the first class with the quasi-cyclic LP codes in \cite{panteleev2022liftedproduct}. In \Cref{sec: decoder}, we describe how to enumerate short cycles in the 2D-SC codes constructed from hypergraph product codes, first analyzing cycles of length $4$, $6$ and $8$, then developing an algorithm to minimize the number of these cycles. In \Cref{sec: simulation results}, we compare the performance of our optimized SC-QLDPC codes with that of non-optimized SC-QLDPC codes with the BP decoder described in \cite{poulin2008iterative}. In \Cref{sec:conclusion}, we conclude and discuss future research directions.

\section{Preliminaries}
\label{sec: preliminaries}

In this section, we briefly introduce the construction of classical SC-LDPC codes and stabilizer codes. In the remainder of this paper, let $\mathbb{N}$, $\mathbb{N}^*$ represent the nonnegative integers and positive integers, respectively. Let $\mathbb{F}_2$ be the Galois field of order $2$. Denote the residue modulo $d$ by $(\cdot)_d$. Denote the all zero matrix and all one matrix of size $m\times n$ by $\mathbf{0}_{m\times n}$, $\mathbf{1}_{m\times n}$, respectively. Let $(\mathbf{A})_{i,j}$ represent the $(i,j)$-th entry of matrix $\mathbf{A}$. The Kronecker product of matrices is denoted by $\otimes$.

\subsection{Classical SC-LDPC Codes}
\label{subsec: sc-ldpc codes}

Given a binary matrix $\mathbf{H}^{\textrm{P}}$, the parity check matrix of a \textbf{quasi-cyclic} (QC-)LDPC code is constructed by replacing each element $\left(\mathbf{H}^{\textrm{P}}\right)_{i,j}$ in $\mathbf{H}^{\textrm{P}}$ with a $z\times z$ block, where each block is either the zero matrix (if and only if $\left(\mathbf{H}^{\textrm{P}}\right)_{i,j}$=0) or some power of the circulant matrix defined in (\ref{eqn: circulant shift matrix}). The Tanner graph associated with $\mathbf{H}^{\textrm{P}}$ is referred to as the \textbf{protograph} of the QC code, and the QC code is said to be \textbf{lifted from} the protograph. 
\begin{equation}\label{eqn: circulant shift matrix}
    \sigma_z=\left[\begin{array}{cccc}
       0 & \cdots & 0 & 1 \\
       1 & \cdots & 0 & 0 \\
       \vdots & \ddots & \vdots & \vdots \\
       0 & \cdots & 1 & 0 
    \end{array}\right].
\end{equation}
An QC-SC code is determined by three $\gamma\times\kappa$ matrices $\mathbf{\Pi}$, $\mathbf{P}$ and $\mathbf{L}$. The \textbf{partitioning matrix} $\mathbf{P}$ specifies how the \textbf{base matrix} $\mathbf{B}$ is split into $m+1$ component matrices $\mathbf{H}^{\textrm{P}}_i$ of the protograph. The \textbf{lifting matrix $\mathbf{L}$} specifies how $\mathbf{H}^{\textrm{P}}_i$ are lifted into component matrices $\mathbf{H}_i$ in the QC-SC code, where $i=0,\dots,m$. The entries of $\mathbf{P}$ range from $0$ through $m$ and the entries of $\mathbf{L}$ range from $0$ through $z-1$. For $1\leq s\leq \gamma$, $1\leq t\leq \kappa$, the $(s,t)$-th element of $\mathbf{H}^{\textrm{P}}_i$ and the $(s,t)$-th block of $\mathbf{H}_i$ are specified as follows:
\begin{equation}
\begin{split}
    \left(\mathbf{H}^{\textrm{P}}_i\right)_{s,t}&=\begin{cases}
    (\mathbf{H})_{s,t},&\textrm{ if }(\mathbf{P})_{s,t}=i,\\
    0,&\textrm{ otherwise.}
    \end{cases},
    \textrm{ and }\\
    \mathbf{H}_{i;s,t}&=\begin{cases}
    \sigma_z^{(\mathbf{L})_{s,t}},&\textrm{ if }\left(\mathbf{H}^{\textrm{P}}_i\right)_{s,t}=1,\\
    \mathbf{0}_{z\times z},&\textrm{ otherwise.}
    \end{cases}
\end{split}
\end{equation}
The component matrices are vertically concatenated into a \textbf{replica}, and these replicas are coupled horizontally. SC codes are categorized as \textbf{tail-biting} (TB) or \textbf{non-tail-biting} (NTB) depending on the form taken by the parity check matrix, as shown in (\ref{eqn: intro_sc}). The \textbf{memory} of the SC code is $m$ and the number of replicas is referred to as the \textbf{coupling length} and is denoted by $L$.
\begin{widetext}
\begin{equation}
\label{eqn: intro_sc}
\mathbf{H}_{\textrm{TB}}=\left[\begin{array}{ccccccc}
\mathbf{H}_0 & \mathbf{0} & \cdots & \mathbf{0} & \mathbf{H}_m & \cdots & \mathbf{H}_1 \\
\mathbf{H}_1 & \mathbf{H}_0 & \mathbf{0} & \cdots & \mathbf{0} & \ddots & \vdots \\
\vdots & \mathbf{H}_1 & \mathbf{H}_0 & \ddots & \ddots & \ddots & \mathbf{H}_m \\
\mathbf{H}_m & \vdots & \ddots & \ddots & \mathbf{0} & \ddots & \mathbf{0} \\
\mathbf{0} & \mathbf{H}_m & \cdots & \mathbf{H}_1 & \mathbf{H}_0 & \ddots & \vdots \\
\vdots & \ddots & \ddots & \vdots & \ddots & \ddots & \mathbf{0} \\
\mathbf{0} & \cdots & \mathbf{0} & \mathbf{H}_m & \cdots & \mathbf{H}_1 & \mathbf{H}_0 \\
\end{array}\right],\  
\mathbf{H}_{\textrm{NTB}}=\left[\begin{array}{cccc}
\mathbf{H}_0 & \mathbf{0} & \cdots & \mathbf{0} \\
\mathbf{H}_1 & \mathbf{H}_0 & \ddots & \vdots \\
\vdots & \mathbf{H}_1 & \ddots & \mathbf{0} \\
\mathbf{H}_m & \vdots & \ddots & \mathbf{H}_0 \\
\mathbf{0} & \mathbf{H}_m & \ddots & \mathbf{H}_1 \\
\vdots & \ddots & \ddots & \vdots  \\
\mathbf{0} & \cdots & \mathbf{0} & \mathbf{H}_m  \\
\end{array}\right].
\end{equation}
\end{widetext}
\begin{exam}\label{exam: classical SC} Let $z=3$, $\gamma=2$, $\kappa=3$, $m=1$, and $L=4$. The $2\times 3$ base matrix $\mathbf{B}$, partitioning matrix $\mathbf{P}$, and lifting matrix $\mathbf{L}$ are given by:
\begin{equation}
\begin{split}
\mathbf{B}=\left[\begin{array}{ccc}
1 & 1 & 1\\
1 & 1 & 1
\end{array}\right],\ 
\mathbf{P}=\left[\begin{array}{ccc}
0 & 1 & 0\\
1 & 0 & 1
\end{array}\right],\\
\textrm{ and }
\mathbf{L}=\left[\begin{array}{ccc}
1 & 0 & 2\\
0 & 2 & 1
\end{array}\right].
\end{split}
\end{equation}

The protograph $\mathbf{H}^{\textup{P}}$ specified by $\mathbf{B}$ and $\mathbf{P}$ is given by:
\begin{equation}
\mathbf{H}^{\textup{P}}=\left[\begin{array}{cccc}
\mathbf{H}_0^{\textup{P}} & \mathbf{0} & \mathbf{0} & \mathbf{H}_1^{\textup{P}} \\
\mathbf{H}_1^{\textup{P}} & \mathbf{H}_0^{\textup{P}} & \mathbf{0} & \mathbf{0} \\
\mathbf{0} & \mathbf{H}_1^{\textup{P}} & \mathbf{H}_0^{\textup{P}} & \mathbf{0} \\
\mathbf{0} & \mathbf{0} & \mathbf{H}_1^{\textup{P}} & \mathbf{H}_0^{\textup{P}}
\end{array}\right],
\end{equation}
where
\begin{equation}
\mathbf{H}_0^{\textup{P}}=\left[\begin{array}{ccc}
1 & 0 & 1\\
0 & 1 & 0
\end{array}\right],\ 
\mathbf{H}_1^{\textup{P}}=\left[\begin{array}{ccc}
0 & 1 & 0\\
1 & 0 & 1
\end{array}\right].
\end{equation}
The parity check matrix $\mathbf{H}$ lifted from $\mathbf{H}^{\textup{P}}$ with respect to $\mathbf{L}$ is given by:
\begin{equation}
\mathbf{H}=\left[\begin{array}{cccc}
\mathbf{H}_0 & \mathbf{0} & \mathbf{0} & \mathbf{H}_1 \\
\mathbf{H}_1 & \mathbf{H}_0 & \mathbf{0} & \mathbf{0} \\
\mathbf{0} & \mathbf{H}_1 & \mathbf{H}_0 & \mathbf{0} \\
\mathbf{0} & \mathbf{0} & \mathbf{H}_1 & \mathbf{H}_0
\end{array}\right],
\end{equation}

where
\begin{equation}
\begin{split}
\mathbf{H}_0=\left[\begin{array}{ccc|ccc|ccc}
0 & 0 & 1 & 0 & 0 & 0 & 0 & 1 & 0\\
1 & 0 & 0 & 0 & 0 & 0 & 0 & 0 & 1\\
0 & 1 & 0 & 0 & 0 & 0 & 1 & 0 & 0\\
\hline
0 & 0 & 0 & 0 & 1 & 0 & 0 & 0 & 0\\
0 & 0 & 0 & 0 & 0 & 1 & 0 & 0 & 0\\
0 & 0 & 0 & 1 & 0 & 0 & 0 & 0 & 0\\
\end{array}\right],\\
\mathbf{H}_1=\left[\begin{array}{ccc|ccc|ccc}
0 & 0 & 0 & 1 & 0 & 0 & 0 & 0 & 0\\
0 & 0 & 0 & 0 & 1 & 0 & 0 & 0 & 0\\
0 & 0 & 0 & 0 & 0 & 1 & 0 & 0 & 0\\
\hline
1 & 0 & 0 & 0 & 0 & 0 & 0 & 0 & 1\\
0 & 1 & 0 & 0 & 0 & 0 & 1 & 0 & 0\\
0 & 0 & 1 & 0 & 0 & 0 & 0 & 1 & 0\\
\end{array}\right].
\end{split}
\end{equation}
\end{exam}

\subsection{Two-Dimensional SC Codes}

A memory $m_1$ SC code is said to be a \textbf{two-dimensional (2D) SC code} if each component matrix $\mathbf{H}_i$, $i=0,\dots,m_1$, is itself a memory $m_2$ SC code with component matrices $\mathbf{H}_{i,j}$, $j=0,\dots,m_2$. 

\begin{exam} Let $m_1=1$ and $m_2=1$. Let $\mathbf{H}_{0,0}=\mathbf{A}$, $\mathbf{H}_{0,1}=\mathbf{B}$, $\mathbf{H}_{1,0}=\mathbf{C}$, and $\mathbf{H}_{1,1}=\mathbf{D}$. The 2D-SC code is given by:
\begin{widetext}
\begin{equation}
\label{eqn: 2d sc}
\mathbf{H}_{\textrm{2D-SC}}=\left[\begin{array}{c|c|c|c}
\begin{array}{cccc}
\mathbf{A} & \mathbf{0} & \cdots & \mathbf{B} \\
\mathbf{B} & \mathbf{A} & \cdots & \mathbf{0} \\
\vdots & \ddots & \ddots & \vdots \\
\mathbf{0} & \cdots & \mathbf{B} & \mathbf{A}
\end{array} & \mathbf{0} & \cdots & \begin{array}{cccc}
\mathbf{C} & \mathbf{0} & \cdots & \mathbf{D} \\
\mathbf{D} & \mathbf{C} & \cdots & \mathbf{0} \\
\vdots & \ddots & \ddots & \vdots \\
\mathbf{0} & \cdots & \mathbf{D} & \mathbf{C}
\end{array} \\
\hline
\begin{array}{cccc}
\mathbf{C} & \mathbf{0} & \cdots & \mathbf{D} \\
\mathbf{D} & \mathbf{C} & \cdots & \mathbf{0} \\
\vdots & \ddots & \ddots & \vdots \\
\mathbf{0} & \cdots & \mathbf{D} & \mathbf{C}
\end{array}  & \begin{array}{cccc}
\mathbf{A} & \mathbf{0} & \cdots & \mathbf{B} \\
\mathbf{B} & \mathbf{A} & \cdots & \mathbf{0} \\
\vdots & \ddots & \ddots & \vdots \\
\mathbf{0} & \cdots & \mathbf{B} & \mathbf{A}
\end{array} & \cdots & \mathbf{0} \\
\hline
\vdots & \ddots & \ddots & \vdots \\
\hline
\mathbf{0} & \cdots & \begin{array}{cccc}
\mathbf{C} & \mathbf{0} & \cdots & \mathbf{D} \\
\mathbf{D} & \mathbf{C} & \cdots & \mathbf{0} \\
\vdots & \ddots & \ddots & \vdots \\
\mathbf{0} & \cdots & \mathbf{D} & \mathbf{C}
\end{array}  & \begin{array}{cccc}
\mathbf{A} & \mathbf{0} & \cdots & \mathbf{B} \\
\mathbf{B} & \mathbf{A} & \cdots & \mathbf{0} \\
\vdots & \ddots & \ddots & \vdots \\
\mathbf{0} & \cdots & \mathbf{B} & \mathbf{A}
\end{array}
\end{array}\right].
\end{equation}
\end{widetext}
\end{exam}

\subsection{The Pauli Group}
\label{subsec: stabilizer codes}

An $2\times 2$ Hermitian matrix can be uniquely expressed as a real linear combination of the four single qubit \textbf{Pauli matrices}:
\begin{equation}
\label{eqn: pauli group}
\begin{split}
I=\left[\begin{array}{cc}
1 & 0 \\
0 & 1
\end{array}\right]&,\ 
X=\left[\begin{array}{cc}
0 & 1 \\
1 & 0
\end{array}\right],\\
Y=\left[\begin{array}{cc}
0 & -i \\
i & 0
\end{array}\right]&,\ 
Z=\left[\begin{array}{cc}
1 & 0 \\
0 & -1
\end{array}\right].\ 
\end{split}
\end{equation}
The matrices satisfy $X^2=Y^2=Z^2=\mathbf{I}_2$, $XY=-YX$, $ZX=-XZ$, and $YZ=-ZY$.

Let $\mathbf{A}\otimes\mathbf{B}$ denote the Kronecker (tensor) product of two  matrices $\mathbf{A}$ and $\mathbf{B}$. Let $\mathbf{a}=\left[a_1,\dots,a_N\right]$ and $\mathbf{b}=\left[b_1,\dots,b_N\right]$ be binary vectors. Define
\begin{equation}
\begin{split}
    &D(\mathbf{a},\mathbf{b})=X^{a_1}Z^{b_1}\otimes \cdots \otimes X^{a_N}Z^{b_N},\\
    & E(\mathbf{a},\mathbf{b})=i^{\mathbf{a}\mathbf{b}^{\textrm{T}}} D(\mathbf{a},\mathbf{b}).
\end{split}
\end{equation}
Then,
\begin{equation}
\begin{split}
E(\mathbf{a},\mathbf{b})^2&=i^{2\mathbf{a}\mathbf{b}^{\textrm{T}}} D(\mathbf{a},\mathbf{b})^2\\
&=i^{2\mathbf{a}\mathbf{b}^{\textrm{T}}} \left(i^{2\mathbf{a}\mathbf{b}^{\textrm{T}}}  \mathbf{I}_{2^N}\right)=\mathbf{I}_{2^N}.
\end{split}
\end{equation}

The \textbf{$N$-qubit Pauli Group} $\mathcal{P}_N$ is given by:
\begin{equation}
\begin{split}
    \mathcal{P}_N=\{i^k D(\mathbf{a},\mathbf{b})| \mathbf{a},\mathbf{b}\textrm{ are binary vectors}&, \\
    k=0,1,2,3&\}.
\end{split}
\end{equation}

We use the Dirac notation $|\cdot\rangle$ to represent the basis states of a single qubit in $\mathbb{C}^2$. The Pauli matrices act on a single qubit as $X|0\rangle=|1\rangle$, $X|1\rangle=|0\rangle$, $Z|0\rangle=|0\rangle$, and $Z|1\rangle=-|1\rangle$. For any $\mathbf{v}=\left[v_1,\dots,v_N\right]$ in $\mathbb{F}_2^N$, we define:
\begin{equation}
    |\mathbf{v}\rangle=|v_1\rangle \otimes |v_2\rangle \cdots \otimes |v_N\rangle,
\end{equation}
the standard basis vector with $1$ the position indexed by $\mathbf{v}$ and zeros elsewhere. We may write an arbitrary $N$ qubit quantum state as 
\begin{equation}
    |\phi\rangle=\sum_{\mathbf{v}\in\mathbb{F}_2^N} \alpha_{\mathbf{v}} |\mathbf{v}\rangle,
\end{equation}
where $\alpha_{\mathbf{v}}\in\mathbb{C}$ and $\sum_{\mathbf{v}\in\mathbb{F}_2^N} |\alpha_{\mathbf{v}}|^2=1$.

The \textbf{symplectic inner product} is $\langle\left(\mathbf{a},\mathbf{b}\right),\left(\mathbf{c},\mathbf{d}\right)\rangle_s=\mathbf{a}\mathbf{b}^{\textrm{T}}+\mathbf{c}\mathbf{d}^{\textrm{T}}$. Since $XZ=-ZX$ we have:
\begin{equation}
    E(\mathbf{a},\mathbf{b})E(\mathbf{c},\mathbf{d})=(-1)^{\langle\left(\mathbf{a},\mathbf{b}\right),\left(\mathbf{c},\mathbf{d}\right)\rangle_s}E(\mathbf{c},\mathbf{d})E(\mathbf{a},\mathbf{b}).
\end{equation}

\subsection{Stabilizer Codes}
We define a \textbf{stabilizer group} $S$ to be a commutative subgroup of the Pauli group $\mathcal{P}_N$, where every group element is Hermitian and no group element is $-I_{2^N}$. We say $S$ has a dimension $r$ if it can be generated by $r$ independent elements as $\mathcal{S}=\langle \nu_iE(\mathbf{c}_i,\mathbf{d}_i)|i=1,2,\dots,r\rangle$, where $\nu_i=\pm 1$ and $c_i,d_i\in\mathbb{F}_2^N$. Since $\mathcal{S}$ is commutative we must have:
\begin{equation}
\langle\left(\mathbf{c}_i,\mathbf{d}_i\right),\left(\mathbf{c}_j,\mathbf{d}_j\right)\rangle_s=\mathbf{c}_i\mathbf{d}_j^{\textrm{T}}+\mathbf{d}_i\mathbf{c}_j^{\textrm{T}}=0.
\end{equation}

Given a stabilizer group $\mathcal{S}$, the corresponding \textbf{stabilizer code} is the fixed subspace $\mathcal{V}(\mathcal{S})=\{|\phi\rangle\in\mathbb{C}^{2N}\big|g|\phi\rangle=|\phi\rangle\textrm{ for all }g\in\mathcal{S}\}$. We refer to the subspace $\mathcal{V}(\mathcal{S})$ as an $\left[\left[N,K\right]\right]$ stabilizer code, because it encodes $K=N-r$ logical qubits into $N$ physical qubits.

\section{Toric Codes as SC codes}

A \textbf{Toric code} is a CSS code defined on a $d\times d$ grid. In Fig.~\ref{fig Toric code}, edges represent qubits, edges incident to a face represent $X$-type generators, and edges incident to a vertex represent $Z$-type generators. The geometry of the grid implies that any pair of stabilizer generator commutes.

We organize the rows of the parity check matrix $\mathbf{H}$ for the $3\times 3$ toric code so that the structure of $\mathbf{H}$ resembles that of a 2D-SC code, as shown in (\ref{eqn: toric code eqn detailed}). Given the labels in Fig.~\ref{fig Toric code}, column $i$ represents the qubit associated with edge $i$. Row $2i-1$ represents the stabilizer generator associated with face $i$, and row $2i$ represents the generator associated with vertex $i$. For example, row $15$ represents $X_{10}\otimes X_{13}\otimes X_{14}\otimes X_{15}$ and row $10$ represents $Z_2\otimes Z_3\otimes Z_4\otimes Z_7$, and face $8$. 

\begin{figure}[hbtp]
\centering
\resizebox{0.3\textwidth}{!}{\begin{tikzpicture}[
  rdv/.style={circle, draw=black, fill=black!5, very thick, minimum size=7mm},
  rdvg/.style={circle, draw=black!30, fill=black!3, very thick, minimum size=7mm},
  rdvb/.style={circle, draw=blue, fill=blue!5, very thick, minimum size=7mm},
  sqv/.style={rectangle, draw=black, fill=black!5, very thick, minimum size=5mm},
  sqvr/.style={rectangle, draw=red, fill=red!5, very thick, minimum size=5mm},
  every edge quotes/.style = {auto},
  ]

  \node[rdv] (v1) {1};
  \node[rdv] (v2) [right=of v1] {2};
  \node[rdv] (v3) [right=of v2] {3};
  \node[rdvg] (v4) [right=of v3] {1};

  \node[rdv] (v5) [below=of v1] {6};
  \node[rdv] (v6) [below=of v2] {4};
  \node[rdvb] (v7) [below=of v3] {5};
  \node[rdvg] (v8) [below=of v4] {6};

  \node[rdv] (v9) [below=of v5] {8};
  \node[rdv] (v10) [below=of v6] {9};
  \node[rdv] (v11) [below=of v7] {7};
  \node[rdvg] (v12) [below=of v8] {8};

  \node[rdvg] (v13) [below=of v9] {1};
  \node[rdvg] (v14) [below=of v10] {2};
  \node[rdvg] (v15) [below=of v11] {3};
  \node[rdvg] (v16) [below=of v12] {1};

  \node[sqv] (f1) [below left=3mm of v2] {1};
  \node[sqv] (f2) [below left=3mm of v3] {2};
  \node[sqv] (f3) [below left=3mm of v4] {3};

  \node[sqv] (f4) [below left=3mm of v6] {6};
  \node[sqv] (f5) [below left=3mm of v7] {4};
  \node[sqv] (f6) [below left=3mm of v8] {5};

  \node[sqvr] (f7) [below left=3mm of v10] {8};
  \node[sqv] (f8) [below left=3mm of v11] {9};
  \node[sqv] (f9) [below left=3mm of v12] {7};

  \draw[-,very thick] (v1.south) edge["5"] (v5.north);
  \draw[-,very thick] (v2.south) edge["1"] (v6.north);
  \draw[-,very thick,color=blue] (v3.south) edge["3"] (v7.north);
  \draw[-,very thick,color=black!30] (v4.south) edge["5"] (v8.north);

  \draw[-,very thick] (v5.south) edge["9"] (v9.north);
  \draw[-,very thick] (v6.south) edge["11"] (v10.north);
  \draw[-,very thick,color=blue] (v7.south) edge["7"] (v11.north);
  \draw[-,very thick,color=black!30] (v8.south) edge["9"] (v12.north);

  \draw[-,very thick,color=red] (v9.south) edge["13"] (v13.north);
  \draw[-,very thick,color=red] (v10.south) edge["15"] (v14.north);
  \draw[-,very thick] (v11.south) edge["17"] (v15.north);
  \draw[-,very thick,color=black!30] (v12.south) edge["13"] (v16.north);

  \draw[-,very thick,color=black!30] (v1.east) edge["14"] (v2.west);
  \draw[-,very thick,color=black!30] (v2.east) edge["16"] (v3.west);
  \draw[-,very thick,color=black!30] (v3.east) edge["18"] (v4.west);

  \draw[-,very thick] (v5.east) edge["6"] (v6.west);
  \draw[-,very thick,color=blue] (v6.east) edge["2"] (v7.west);
  \draw[-,very thick,color=blue] (v7.east) edge["4"] (v8.west);

  \draw[-,very thick,color=red] (v9.east) edge["10"] (v10.west);
  \draw[-,very thick] (v10.east) edge["12"] (v11.west);
  \draw[-,very thick] (v11.east) edge["8"] (v12.west);

  \draw[-,very thick,color=red] (v13.east) edge["14"] (v14.west);
  \draw[-,very thick] (v14.east) edge["16"] (v15.west);
  \draw[-,very thick] (v15.east) edge["18"] (v16.west);

\end{tikzpicture}}
  \caption{A toric code defines on a $3\times 3$ grid. Vertices are indicated by circles and faces are indicated by squares. Vertices/edges with the same index are glued together to create a torus with $18$ edges, $9$ vertices, and $9$ faces. Vertex $5$ specifies the generator $Z_2\otimes Z_3\otimes Z_4\otimes Z_7$ and face $8$ specifies generator $X_{10}\otimes X_{13}\otimes X_{14}\otimes X_{15}$.}   
  \label{fig Toric code}
\end{figure}

We have
\begin{equation}\label{eqn: toric eqn detailed}
    \mathbf{H}=\left[\begin{array}{ccc|ccc|ccc}
    A &   & B &   &   &   & C &   & D \\
    B & A &   &   &   &   & D & C &   \\
      & B & A &   &   &   &   & D & C \\
    \hline
    C &   & D & A &   & B &   &   &   \\
    D & C &   & B & A &   &   &   &   \\
      & D & C &   & B & A &   &   &   \\
    \hline
      &   &   & C &   & D & A &   & B \\
      &   &   & D & C &   & B & A &   \\
      &   &   &   & D & C &   & B & A \\
    \end{array}\right],
\end{equation}
where 
\begin{equation}\label{eqn: toric abcd}
\begin{split}
\mathbf{A}=\left[\begin{array}{cc}
X & I\\
I & I
\end{array}\right]&,\ 
\mathbf{B}=\left[\begin{array}{cc}
X & X\\
Z & I
\end{array}\right],\\ 
\mathbf{C}=\left[\begin{array}{cc}
I & X\\
Z & Z
\end{array}\right]&,\ 
\mathbf{D}=\left[\begin{array}{cc}
I & I\\
I & Z
\end{array}\right].
\end{split}
\end{equation}

\begin{widetext}
\begin{equation}\label{eqn: toric code eqn detailed}
\scalebox{0.72}{$\left[\begin{array}{c|c|c}
\begin{array}{c|c|c}
\begin{array}{cc}
X & I\\
I & I
\end{array}& & \begin{array}{cc}
X & X\\
Z & I
\end{array} \\
\hline
\begin{array}{cc}
X & X\\
Z & I
\end{array} & \begin{array}{cc}
X & I\\
I & I
\end{array} & \\
\hline
& \begin{array}{cc}
X & X\\
Z & I
\end{array} & \begin{array}{cc}
X & I\\
I & I
\end{array} \\
\end{array}& & \begin{array}{c|c|c}
\begin{array}{cc}
I & X\\
Z & Z
\end{array} & & \begin{array}{cc}
I & I\\
I & Z
\end{array}\\
\hline
\begin{array}{cc}
I & I\\
I & Z
\end{array} & \begin{array}{cc}
I & X\\
Z & Z
\end{array} & \\
\hline
& \begin{array}{cc}
I & I\\
I & Z
\end{array}& \begin{array}{cc}
I & X\\
Z & Z
\end{array} \\
\end{array} \\
\hline 
\begin{array}{c|c|c}
\begin{array}{cc}
I & X\\
Z & Z
\end{array} & & \begin{array}{cc}
I & I\\
I & Z
\end{array}\\
\hline
\begin{array}{cc}
I & I\\
I & \textcolor{blue}{Z}
\end{array} & \begin{array}{cc}
I & X\\
\textcolor{blue}{Z} & \textcolor{blue}{Z}
\end{array} & \\
\hline
& \begin{array}{cc}
I & I\\
I & Z
\end{array}& \begin{array}{cc}
I & X\\
Z & Z
\end{array} \\
\end{array}& \begin{array}{c|c|c}
\begin{array}{cc}
X & I\\
I & I
\end{array}& & \begin{array}{cc}
X & X\\
Z & I
\end{array} \\
\hline
\begin{array}{cc}
X & X\\
\textcolor{blue}{Z} & I
\end{array} & \begin{array}{cc}
X & I\\
I & I
\end{array} & \\
\hline
& \begin{array}{cc}
X & X\\
Z & I
\end{array} & \begin{array}{cc}
X & I\\
I & I
\end{array} \\
\end{array} & \\
\hline 
& \begin{array}{c|c|c}
\begin{array}{cc}
I & X\\
Z & Z
\end{array} & & \begin{array}{cc}
I & I\\
I & Z
\end{array}\\
\hline
\begin{array}{cc}
I & I\\
I & Z
\end{array} & \begin{array}{cc}
I & \textcolor{red}{X}\\
Z & Z
\end{array} & \\
\hline
& \begin{array}{cc}
I & I\\
I & Z
\end{array}& \begin{array}{cc}
I & X\\
Z & Z
\end{array} \\
\end{array}& \begin{array}{c|c|c}
\begin{array}{cc}
X & I\\
I & I
\end{array}& & \begin{array}{cc}
X & X\\
Z & I
\end{array} \\
\hline
\begin{array}{cc}
\textcolor{red}{X} & \textcolor{red}{X}\\
Z & I
\end{array} & \begin{array}{cc}
\textcolor{red}{X} & I\\
I & I
\end{array} & \\
\hline
& \begin{array}{cc}
X & X\\
Z & I
\end{array} & \begin{array}{cc}
X & I\\
I & I
\end{array} \\
\end{array}  \\
\end{array}\right]$.}
\end{equation}
\end{widetext}

\begin{rem} \emph{(Extension to $d\times d$ grids)} The face/vertex labels at level $i$ are obtained by adding $d$ to the face/vertex labels at level $i-1$, then cycling to the right. Thus, a face and the vertex at the top left of the face receives the same label $l$ (just as in Fig.~\ref{fig Toric code}). To be precise, the $(i,j)$-th face, $i,j=1,\dots,d$, is labelled by $\big((i-1)d+(j-i)_d+1\big)$, where $(\cdot)_d$ denotes the residue modulo $d$. 

The edges at the left and at the bottom of the $(i,j)$-th face are then labelled consecutively as $2\big(i-1)d+(j-i-1)_d\big)+1$ and $2\big(i-1)d+(j-i-1)_d\big)+2$. The parity check matrix $\mathbf{H}$ of the $d\times d$ toric code takes the form in (\ref{eqn: 2d sc}), where $\mathbf{A}$, $\mathbf{B}$, $\mathbf{C}$, and $\mathbf{D}$ are given by (\ref{eqn: toric abcd}). It has the structure of a 2D-SC code with $(m_1,m_2,L_1,L_2)=(1,1,d,d)$.
\end{rem}

\section{Commutative Algebra}
\label{sec: quantum 2d-sc codes}

We begin by extending the symplectic inner product to matrices for which every entry is either the zero matrix $\mathbf{0}_2$ or a $2\times 2$ Pauli matrix. Given an $m_1\times n$ matrix $\mathbf{P}=\left[P_{i,k}\right]$ and an $m_2\times n$ matrix $\mathbf{Q}=\left[Q_{j,k}\right]$ define $\langle P,Q\rangle_s$ to be the $m_1\times m_2$ binary matrix given by:
\begin{equation}
\begin{split}
    \bigl(\langle \mathbf{P},\mathbf{Q}\rangle_s\bigr)_{i,j}&=\bigg\langle \big[\bigotimes_{k: P_{i,k}\neq\mathbf{0}_2} P_{i,k}\big], \big[\bigotimes_{k: Q_{i,k}\neq\mathbf{0}_2} Q_{j,k}\big]\bigg\rangle_s\\
    &=\sum_{k: P_{i,k}\neq \mathbf{0}_2\textrm{ and }Q_{i,k}\neq \mathbf{0}_2} \langle P_{i,k},Q_{j,k}\rangle_s.
\end{split}
\end{equation}

\begin{exam}\label{exam matrix multiplication} Let $m_1=1$, $m_2=2$, and $n=4$. Specify $\mathbf{P}\in\mathcal{P}^{1\times 4}$, $\mathbf{Q}\in\mathcal{P}^{2\times 4}$ as follows:
\begin{equation}
\begin{split}
\mathbf{P}&=\left[\begin{array}{cccc}
X & Y & Z & I
\end{array}\right],\\
\mathbf{Q}&=\left[\begin{array}{cccc}
I & X & Z & Y \\
Z & Y & X & X
\end{array}\right].   
\end{split}
\end{equation}
Then, 
\begin{widetext}
\begin{equation}
\begin{split}
\langle \mathbf{P},\mathbf{Q} \rangle_s=&\left[\begin{array}{cc}
\langle X,I \rangle_s+\langle Y,X \rangle_s+\langle Z,Z \rangle_s+\langle I,Y \rangle_s &
\langle X,Z \rangle_s+\langle Y,Y \rangle_s+\langle Z,X \rangle_s+\langle I,X \rangle_s
\end{array}\right]\\
=&\left[\begin{array}{cc}
0+1+0+0 & 1+0+1+0
\end{array}\right]=\left[\begin{array}{cc}
1 & 0
\end{array}\right].
\end{split}
\end{equation}
\end{widetext}
\end{exam}

\begin{exam}\label{exam toric code theo: 1} (\textbf{Toric codes as 2D-SC codes}) We verify that the matrix $\mathbf{H}$ given below (recall (\ref{eqn: 2d sc})) is the parity check matrix of a stabilizer code by checking five symplectic inner products. Every row of $\mathbf{H}$ determines a Pauli matrix, and the requirement that generators commute translates to the requirement that five symplectic inner product vanish. The matrices $\mathbf{A},\mathbf{B},\mathbf{C},\mathbf{D}$ are defined in (\ref{eqn: toric abcd}).

\begin{widetext}
\begin{equation*}
\mathbf{H}=\left[\begin{array}{c|c|c|c}
\begin{array}{cccc}
\cl\tikznode{11}{$\mathbf{A}$} & \mathbf{0} & \cdots & \mathbf{B} \\
\mathbf{B} & \mathbf{A} & \cdots & \mathbf{0} \\
\vdots & \ddots & \ddots & \vdots \\
\mathbf{0} & \cdots & \cl\tikznode{31}{$\mathbf{B}$} & \cl\tikznode{32}{$\mathbf{A}$}
\end{array} & \mathbf{0} & \cdots & \begin{array}{cccc}
\mathbf{C} & \mathbf{0} & \cdots & \mathbf{D} \\
\mathbf{D} & \mathbf{C} & \cdots & \mathbf{0} \\
\vdots & \ddots & \ddots & \vdots \\
\mathbf{0} & \cdots & \mathbf{D} & \cl\tikznode{21}{$\mathbf{C}$}
\end{array} \\
\hline
\begin{array}{cccc}
\mathbf{C} & \mathbf{0} & \cdots & \mathbf{D} \\
\cl\tikznode{12}{$\mathbf{D}$} & \mathbf{C} & \cdots & \mathbf{0} \\
\vdots & \ddots & \ddots & \vdots \\
\mathbf{0} & \cdots & \cl\tikznode{33}{$\mathbf{D}$} & \cl\tikznode{34}{$\mathbf{C}$}
\end{array}  & \begin{array}{cccc}
\mathbf{A} & \mathbf{0} & \cdots & \mathbf{B} \\
\mathbf{B} & \mathbf{A} & \cdots & \mathbf{0} \\
\vdots & \ddots & \ddots & \vdots \\
\mathbf{0} & \cdots & \mathbf{B} & \mathbf{A}
\end{array} & \cdots & \mathbf{0} \\
\hline
\vdots & \ddots & \ddots & \vdots \\
\hline
\mathbf{0} & \cdots & \begin{array}{cccc}
\cl\tikznode{41}{$\mathbf{C}$} & \mathbf{0} & \cdots & \mathbf{D} \\
\cl\tikznode{42}{$\mathbf{D}$} & \mathbf{C} & \cdots & \mathbf{0} \\
\vdots & \ddots & \ddots & \vdots \\
\mathbf{0} & \cdots & \cl\tikznode{51}{$\mathbf{D}$} & \cl\tikznode{52}{$\mathbf{C}$}
\end{array}  & \begin{array}{cccc}
\cl\tikznode{43}{$\mathbf{A}$} & \mathbf{0} & \cdots & \cl\tikznode{22}{$\mathbf{B}$} \\
\cl\tikznode{44}{$\mathbf{B}$} & \mathbf{A} & \cdots & \mathbf{0} \\
\vdots & \ddots & \ddots & \vdots \\
\mathbf{0} & \cdots & \cl\tikznode{53}{$\mathbf{B}$} & \cl\tikznode{54}{$\mathbf{A}$}
\end{array}
\end{array}\right].
\end{equation*}
\end{widetext}

\begin{tikzpicture}[remember picture,overlay,blue,rounded corners]
  \draw[<-,shorten <=1pt] (11)
    -- +(-0.4,0)
    |- +(-1,0.4)
    coordinate (p1)
    node[left] {$(1)$};
  \draw[<-,shorten <=1pt] (12)
    -- +(-0.4,0)
    |- (p1);

  \draw[<-,shorten <=1pt] (21)
    -- +(0.4,0)
    |- +(1,0.4)
    coordinate (p2)
    node[right] {$(2)$};
  \draw[<-,shorten <=1pt] (22)
    -- +(0.4,0)
    |- (p2);

  \draw[<-,shorten <=1pt] (33)
    -- +(-0.4,0)
    |- +(-2.5,-0.4)
    coordinate (p3)
    node[left] {$(3)$};
  \draw[<-,shorten <=1pt] (32)
    -- +(-0.4,0)
    |- (p3);
  \draw[<-,shorten <=1pt] (31)
    -- +(-0.4,0)
    |- (p3);
  \draw[<-,shorten <=1pt] (34)
    -- +(-0.4,0)
    |- (p3);

  \draw[<-,shorten <=1pt] (42)
    -- +(0.4,0)
    |- +(6.4,-0.4)
    coordinate (p4)
    node[right] {$(4)$};
  \draw[<-,shorten <=1pt] (41)
    -- +(0.4,0)
    |- (p4);
  \draw[<-,shorten <=1pt] (43)
    -- +(0.4,0)
    |- (p4);
  \draw[<-,shorten <=1pt] (44)
    -- +(0.4,0)
    |- (p4);

  \draw[<-,shorten <=1pt] (54)
    -- +(0.4,0)
    |- +(1,-0.4)
    coordinate (p5)
    node[right] {$(5)$};
  \draw[<-,shorten <=1pt] (52)
    -- +(0.4,0)
    |- (p5);
  \draw[<-,shorten <=1pt] (53)
    -- +(0.4,0)
    |- (p5);
  \draw[<-,shorten <=1pt] (51)
    -- +(0.4,0)
    |- (p5);
 
\end{tikzpicture}

We require 
\begin{enumerate}[label=(\arabic*)]
  \item $\langle \mathbf{A},\mathbf{D}\rangle_s=\mathbf{0}_{2\times 2}$;
  \item $\langle \mathbf{B},\mathbf{C}\rangle_s=\mathbf{0}_{2\times 2}$;
  \item $\langle \mathbf{A},\mathbf{C}\rangle_s+\langle \mathbf{B},\mathbf{D}\rangle_s=\mathbf{0}_{2\times 2}$;
  \item $\langle \mathbf{A},\mathbf{B}\rangle_s+\langle \mathbf{C},\mathbf{D}\rangle_s=\mathbf{0}_{2\times 2}$;
  \item $\langle \mathbf{A},\mathbf{A}\rangle_s+\langle \mathbf{B},\mathbf{B}\rangle_s+\langle \mathbf{C},\mathbf{C}\rangle_s+\langle \mathbf{D},\mathbf{D}\rangle_s=\mathbf{0}_{2\times 2}$.
\end{enumerate}

\end{exam}

\begin{rem} More generally, if $\mathbf{A}$, $\mathbf{B}$, $\mathbf{C}$, $\mathbf{D}$ are $r\times n$ matrices satisfying the above conditions (with $\mathbf{0}_{2\times 2}$ replaced by $\mathbf{0}_{r\times r}$), then the 2D-SC code with component matrices $(\mathbf{H}_{0,0},\mathbf{H}_{0,1},\mathbf{H}_{1,0},\mathbf{H}_{1,1})=(\mathbf{A}, \mathbf{B}, \mathbf{C}, \mathbf{D})$ is a stabilizer code.
\end{rem}

\begin{theo}\label{theo: Quantum 2D SC Codes} Let $\mathcal{P}$ be the Pauli group on a single qubit, and let $\mathbf{H}_{i,j}\in \mathcal{P}^{r\times n}$, $i=0,\dots,m_1$, $j=0,\dots,m_2$, be the component matrices of a 2D-SC code with parameters $(m_1,m_2,L_1,L_2)$. If the following conditions are satisfied:
\begin{equation}\label{eqn quantum 2d sc codes 1}
\begin{split}
\sum\nolimits_{i=0}^{m_1-d_1}\sum\nolimits_{j=0}^{m_2-d_2} \langle \mathbf{H}_{i,j},\mathbf{H}_{i+d_1,j+d_2} \rangle =\mathbf{0}_{r\times r}&,\\
\forall 0\leq d_1\leq m_1, 0\leq d_2\leq m_2&,
\end{split}
\end{equation}
\begin{equation}\label{eqn quantum 2d sc codes 2}
\begin{split}
\sum\nolimits_{i=0}^{m_1-d_1}\sum\nolimits_{j=0}^{m_2-d_2} \langle \mathbf{H}_{i+d_1,j},\mathbf{H}_{i,j+d_2} \rangle =\mathbf{0}_{r\times r}&,\\
\forall 0\leq d_1\leq m_1, 0\leq d_2\leq m_2&,
\end{split}
\end{equation}
then the 2D-SC code is a stabilizer code.
\end{theo}

\begin{rem} When $r=n=2$, $m_1=m_2=1$, and $L_1=L_2=d$, then the conditions (\ref{eqn quantum 2d sc codes 1}) and (\ref{eqn quantum 2d sc codes 2}) reduce to conditions (1)-(5) in \Cref{exam toric code theo: 1}.
\end{rem}

\begin{proof} Conditions (\ref{eqn quantum 2d sc codes 1}) translates to verifying pairwise orthogonality of blocks of rows, as shown in Fig.~\ref{fig: theorem 2d sc}. Condition (\ref{eqn quantum 2d sc codes 2}) translates to a similar visualization, and we omit the details. 
\end{proof}

\begin{figure}
\centering
\resizebox{0.48\textwidth}{!}{\begin{tikzpicture}[darkstyle/.style={circle,draw,fill=gray!10,very thick,minimum size=20}]

\foreach \x [count=\xi] in {-20,0}{
    \foreach \y [count=\yi] in {-10,10}{
        \node[] (a\xi\yi) at (0.5*\x,0.5*\y) {};
        \pgfmathtruncatemacro{\label}{2*\yi+\xi-2}
        \node[] (na\xi\yi) at (0.5*\x,0.5*\y+0.5) {};
        \node[] (wa\xi\yi) at (0.5*\x-0.5,0.5*\y) {};
        \node[] (nwa\xi\yi) at (0.5*\x-0.5,0.5*\y+0.5) {};
    }
}

\node[] (aw) at (-20*0.5,-2*0.5) {};
\node[] (aww) at (-20*0.5-0.5,-2*0.5) {};
\node[] (ae) at (0,-2*0.5) {};
\node[] (as) at (-12*0.5,-10*0.5) {};
\node[] (an) at (-12*0.5,10*0.5) {};
\node[] at (-10*0.5,10*0.5+1) {\Large $L_1$};
\node[] at (-20*0.5-1,4*0.5) {\Large $m_1$};

\node[] (block1) at (-16*0.5,3*0.5) {};
\node[] (block1label) at (-16*0.5-0.3,3*0.5+0.3) {\Large $\mathbf{H}_{i}$};
\draw[draw=black, very thick] (block1) rectangle ++(1,-1);

\node[] (block2) at (-16*0.5,-2*0.5) {};
\node[] (block2label) at (-16*0.5-0.3,-2*0.5+0.3) {\Large $\mathbf{H}_{i+d_1}$};
\draw[draw=black, very thick] (block2) rectangle ++(1,-1);

\draw[color=black, very thick] (a11)--(a12)--(a22)--(a21)--(a11);
\draw[color=gray!50, very thick] (a12)--(a21);
\draw[color=gray!50, very thick] (aw)--(as);
\draw[color=gray!50, very thick] (an)--(ae);

\draw[stealth-stealth, color=blue, thick] (na12)--(na22);
\draw[stealth-stealth, color=blue, thick] (aww)--(wa12);

\foreach \x [count=\xi] in {10,20}{
    \foreach \y [count=\yi] in {3,13}{
        \node[] (b\xi\yi) at (0.5*\x,0.5*\y) {};
        \pgfmathtruncatemacro{\label}{2*\yi+\xi-2}
        \node[] (nb\xi\yi) at (0.5*\x,0.5*\y+0.25) {};
        \node[] (wb\xi\yi) at (0.5*\x-0.25,0.5*\y) {};
        \node[] (nwb\xi\yi) at (0.5*\x-0.5,0.5*\y+0.5) {};
    }
}

\node[] (bw) at (10*0.5,7*0.5) {};
\node[] (bww) at (10*0.5-0.25,7*0.5) {};
\node[] (be) at (20*0.5,7*0.5) {};
\node[] (bs) at (14*0.5,3*0.5) {};
\node[] (bn) at (14*0.5,13*0.5) {};
\node[] at (15*0.5,13*0.5+0.75) {\Large $L_2$};
\node[] at (10*0.5-0.75,10*0.5) {\Large $m_2$};

\node[] (blockb1) at (12*0.5,9.5*0.5) {};
\node[] (blockb1label) at (12*0.5-0.3,9.5*0.5+0.3) {\Large $\mathbf{H}_{i,j}$};
\draw[draw=black, very thick] (blockb1) rectangle ++(0.5,-0.5);

\draw[color=black, very thick] (b11)--(b12)--(b22)--(b21)--(b11);
\draw[color=gray!50, very thick] (b12)--(b21);
\draw[color=gray!50, very thick] (bw)--(bs);
\draw[color=gray!50, very thick] (bn)--(be);

\draw[stealth-stealth, color=blue, thick] (bww)--(wb12);
\draw[stealth-stealth, color=blue, thick] (nb12)--(nb22);

\foreach \x [count=\xi] in {10,20}{
    \foreach \y [count=\yi] in {-13,-3}{
        \node[] (c\xi\yi) at (0.5*\x,0.5*\y) {};
        \pgfmathtruncatemacro{\label}{2*\yi+\xi-2}
        \node[] (nc\xi\yi) at (0.5*\x,0.5*\y+0.25) {};
        \node[] (wc\xi\yi) at (0.5*\x-0.25,0.5*\y) {};
        \node[] (nwc\xi\yi) at (0.5*\x-0.5,0.5*\y+0.5) {};
    }
}

\node[] (cw) at (10*0.5,-9*0.5) {};
\node[] (cww) at (10*0.5-0.25,-9*0.5) {};
\node[] (ce) at (20*0.5,-9*0.5) {};
\node[] (cs) at (14*0.5,-13*0.5) {};
\node[] (cn) at (14*0.5,-3*0.5) {};
\node[] at (15*0.5,-3*0.5+0.75) {\Large $L_2$};
\node[] at (10*0.5-0.75,-6*0.5) {\Large $m_2$};

\node[] (blockc2) at (12*0.5,-9*0.5) {};
\node[] (blockc2label) at (12*0.5+0.5,-9*0.5+0.3) {\Large $\mathbf{H}_{i+d_1,j+d_2}$};
\draw[draw=black, very thick] (blockc2) rectangle ++(0.5,-0.5);

\draw[color=black, very thick] (c11)--(c12)--(c22)--(c21)--(c11);
\draw[color=gray!50, very thick] (c12)--(c21);
\draw[color=gray!50, very thick] (cw)--(cs);
\draw[color=gray!50, very thick] (cn)--(ce);

\draw[stealth-stealth, color=blue, thick] (cww)--(wc12);
\draw[stealth-stealth, color=blue, thick] (nc12)--(nc22);

\draw[-stealth, color=red, thick] (-16*0.5+1.2,3*0.5-0.5)--(9*0.5,8*0.5);
\draw[-stealth, color=red, thick] (-16*0.5+1.2,-2*0.5-0.5)--(9*0.5,-8*0.5);

\end{tikzpicture}}
\caption{Verifying pairwise orthogonality of blocks of rows in (\ref{eqn quantum 2d sc codes 1}).}
\label{fig: theorem 2d sc}
\end{figure}

\begin{rem} In the classical world, non-tail-biting (NTB) codes are associated with better thresholds than tail-biting (TB) codes, due to the existence of low degree check nodes at the two ends of the coupling chain (see the discussion of low threshold saturation in SC codes \cite{5695130,kumar2014threshold,lentmaier2010iterative}).
However, due to the commutativity conditions required by stabilizer codes, in the quantum world, NTB code suffer from low minimum distance. This is because component matrices must pairwise commute, and the block length of the stabilizer code $\mathcal{C}'$ generated by the component matrices is small ($n=N/L$). Any product of stabilizers in $\mathcal{C}'$ gives rise to a non-correctable error pattern in the NTB code. 
\end{rem}

\subsection{SC Codes on Quantum Hardware}\label{subsec: hardware}
Limitations on qubit connectivity present a significant challenge for implementing QLDPC codes across various quantum hardware platforms. However, the convolutional structure of SC codes allows them to be employed efficiently on diverse quantum architectures.

Superconducting platforms, such as those adopted by companies like IBM and Google, only permit entanglement between neighboring qubits, necessitating extensive swap gates when implementing QLDPC codes with arbitrary connections. IBM has proposed partitioning the Tanner graph into multiple planar subgraphs, where the count of these subgraphs, termed “thickness,” influences the implementation efficiency of QLDPC codes \cite{bravyi2024high}. For an SC code with memory parameters $(m_1,m_2)$, the thickness and maximum entanglement distances within each planar subgraph depend on $m_1$ and $m_2$. SC codes with short memories thus allow efficient implementations in superconducting architectures.

Neutral atom arrays, such as those developed by QuEra and PASQAL, offer improved QLDPC code implementation through reconfigurable atom arrays (RAAs) \cite{xu2024constant,wang2024atomique}. RAAs consist of a fixed atom array using spatial light modulators (SLM) and multiple movable arrays using acoustic-optic deflectors (AOD). The AOD arrays enable aligned qubit movement along rows and columns, though paths cannot cross within a single AOD array. As qubits in AOD and SLM arrays move in union, any pair within the Rydberg interaction distance can simultaneously form two-qubit gates using a Rydberg laser. By allocating qubits across AOD arrays and arranging their movements, all required two-qubit gates in QLDPC codes can be realized. The cost and fidelity of RAA implementation depend on atom movement overhead, making high parallelism essential. SC codes naturally support this parallelism by distributing ancilla qubits from different replicas across AOD arrays with identical qubit layouts. Given the arrangement of the base block code, the SC code’s configuration follows naturally.

For example, Fig.~\ref{fig: RAA1} and Fig.~\ref{fig: RAA2} show the implementation of edges in the component matrices $\mathbf{H}_{0,0}$ and $\mathbf{H}_{2,1}$, respectively, in an SC code with a base code size of $(9,4)$, memories $m_1 = m_2 = 2$, and coupling lengths $L_1 = L_2 = 5$. Blue and red dots represent data qubits and ancilla qubits, respectively. The data qubits are arranged by the SLM as a fixed $5 \times 5$ grid, with each $3 \times 3$ sub-array corresponding to the data qubits in a replica. The ancilla qubits are organized into a large AOD array, which is divided into $3 \times 3$ sub-arrays of size $2 \times 2$, and 16 smaller AOD arrays, each containing a $2 \times 2$ qubit layout, allowing for mobility to facilitate two-qubit gate formation. Qubits within Rydberg distance are enclosed in black ovals.

In Fig.\ref{fig: RAA1}, qubits within each $2 \times 2$ AOD array are positioned at the four corners of their corresponding $3 \times 3$ data qubit sub-array, enabling four two-qubit gates in $\mathbf{H}_{0,0}$. In Fig.\ref{fig: RAA2}, to implement two-qubit gates in $\mathbf{H}_{2,1}$, each AOD array is shifted two grids to the left and one grid up in a rotated layout. Within each $2 \times 2$ sub-array, the rows and columns are adjusted so that qubits in the second column align with the top and bottom qubits of the central column in the corresponding $3 \times 3$ data qubit sub-array, while qubits in the first column remain separate from all data nodes, thus activating two two-qubit gates in $\mathbf{H}_{2,1}$. 

In general, SC codes with small memories support highly parallel implementation on RAAs.

\begin{figure}
\centering
\resizebox{0.45\textwidth}{!}{\begin{tikzpicture}
    \def\bluesize{1} 
    \def\redscale{0.9} 

    \foreach \i in {0,1,...,5} {
        \draw[blue, thick] (\i, 0) -- (\i, 5);
        \draw[blue, thick] (0, \i) -- (5, \i);
    }

    \draw[red, thick] (2 + 0.5*\bluesize - 0.5*\redscale*\bluesize, 0.5*\bluesize - 0.5*\redscale*\bluesize) rectangle ++(2*\bluesize+\redscale*\bluesize, 2*\bluesize+\redscale*\bluesize);

    \foreach \x in {0,1,...,4}{
        \foreach \y in {0,1,...,4}{
            \ifnum\x<2
                \draw[red, thick] 
                    (\x + 0.5*\bluesize - 0.5*\redscale*\bluesize, 
                     \y + 0.5*\bluesize - 0.5*\redscale*\bluesize)
                    rectangle ++(\redscale*\bluesize, \redscale*\bluesize);
            \else
                \ifnum\y>2
                    \draw[red, thick] 
                        (\x + 0.5*\bluesize - 0.5*\redscale*\bluesize, 
                         \y + 0.5*\bluesize - 0.5*\redscale*\bluesize)
                        rectangle ++(\redscale*\bluesize, \redscale*\bluesize);
                \fi
            \fi

            \coordinate (blue11) at (\x + 0.2*\bluesize, \y + 0.2*\bluesize);
            \coordinate (blue21) at (\x + 0.2*\bluesize, \y + 0.5*\bluesize);
            \coordinate (blue31) at (\x + 0.2*\bluesize, \y + 0.8*\bluesize);
            \coordinate (blue12) at (\x + 0.5*\bluesize, \y + 0.2*\bluesize);
            \coordinate (blue22) at (\x + 0.5*\bluesize, \y + 0.5*\bluesize);
            \coordinate (blue32) at (\x + 0.5*\bluesize, \y + 0.8*\bluesize);            \coordinate (blue13) at (\x + 0.8*\bluesize, \y + 0.2*\bluesize);
            \coordinate (blue23) at (\x + 0.8*\bluesize, \y + 0.5*\bluesize);
            \coordinate (blue33) at (\x + 0.8*\bluesize, \y + 0.8*\bluesize);
            
            \coordinate (red11) at (\x + 0.2*\bluesize + 0.07*\bluesize, \y + 0.2*\bluesize);
            \coordinate (red12) at (\x + 0.8*\bluesize + 0.07*\bluesize, \y + 0.2*\bluesize);
            \coordinate (red21) at (\x + 0.2*\bluesize + 0.07*\bluesize, \y + 0.8*\bluesize);
            \coordinate (red22) at (\x + 0.8*\bluesize + 0.07*\bluesize, \y + 0.8*\bluesize);

            \node[draw=blue, fill=blue, circle, inner sep=0.5pt] at (blue11) {};
            \node[draw=blue, fill=blue, circle, inner sep=0.5pt] at (blue21) {};
            \node[draw=blue, fill=blue, circle, inner sep=0.5pt] at (blue31) {};
            \node[draw=blue, fill=blue, circle, inner sep=0.5pt] at (blue12) {};
            \node[draw=blue, fill=blue, circle, inner sep=0.5pt] at (blue22) {};
            \node[draw=blue, fill=blue, circle, inner sep=0.5pt] at (blue32) {};
            \node[draw=blue, fill=blue, circle, inner sep=0.5pt] at (blue13) {};
            \node[draw=blue, fill=blue, circle, inner sep=0.5pt] at (blue23) {};
            \node[draw=blue, fill=blue, circle, inner sep=0.5pt] at (blue33) {};

            \node[draw=red, fill=red, circle, inner sep=0.5pt] at (red11) {};
            \node[draw=red, fill=red, circle, inner sep=0.5pt] at (red21) {};
            \node[draw=red, fill=red, circle, inner sep=0.5pt] at (red12) {};
            \node[draw=red, fill=red, circle, inner sep=0.5pt] at (red22) {};

            \draw[black, thick, rounded corners=1mm] 
                ($(blue11)!0.5!(red11)$) ellipse [x radius=0.1, y radius=0.06];
            \draw[black, thick, rounded corners=1mm] 
                ($(blue31)!0.5!(red21)$) ellipse [x radius=0.1, y radius=0.06];
            \draw[black, thick, rounded corners=1mm] 
                ($(blue13)!0.5!(red12)$) ellipse [x radius=0.1, y radius=0.06];
            \draw[black, thick, rounded corners=1mm] 
                ($(blue33)!0.5!(red22)$) ellipse [x radius=0.1, y radius=0.06];
        }
    }
\end{tikzpicture}}
\caption{Layout of an SC code with a $(9,4)$ base code, memories $m_1 = m_2 = 2$, and coupling lengths $L_1 = L_2 = 5$. Blue dots represent data qubits in a fixed spatial light modulator (SLM) array, while red dots denote ancilla qubits in movable acousto-optic deflector (AOD) arrays. Qubit pairs within each black oval are positioned within the Rydberg distance, enabling four two-qubit gates in $\mathbf{H}_{0,0}$ of the corresponding replica via a Rydberg laser.}
\label{fig: RAA1}
\end{figure}

\begin{figure}
\centering
\resizebox{0.45\textwidth}{!}{\begin{tikzpicture}
    \def\bluesize{1} 
    \def\redscale{0.9} 

    \foreach \i in {0,1,...,5} {
        \draw[blue, thick] (\i, 0) -- (\i, 5);
        \draw[blue, thick] (0, \i) -- (5, \i);
    }

    \draw[red, thick] (0.5*\bluesize - 0.5*\redscale*\bluesize, 1+0.5*\bluesize - 0.5*\redscale*\bluesize) rectangle ++(2*\bluesize+\redscale*\bluesize, 2*\bluesize+\redscale*\bluesize);

    \foreach \x in {0,1,...,4}{
        \foreach \y in {0,1,...,4}{
            \ifnum\x>2
                \draw[red, thick] 
                    (\x + 0.5*\bluesize - 0.5*\redscale*\bluesize, 
                     \y + 0.5*\bluesize - 0.5*\redscale*\bluesize)
                    rectangle ++(\redscale*\bluesize, \redscale*\bluesize);
            \else
                \ifnum\y>3
                    \draw[red, thick] 
                        (\x + 0.5*\bluesize - 0.5*\redscale*\bluesize, 
                         \y + 0.5*\bluesize - 0.5*\redscale*\bluesize)
                        rectangle ++(\redscale*\bluesize, \redscale*\bluesize);
                \else
                    \ifnum\y<1
                    \draw[red, thick] 
                        (\x + 0.5*\bluesize - 0.5*\redscale*\bluesize, 
                         \y + 0.5*\bluesize - 0.5*\redscale*\bluesize)
                        rectangle ++(\redscale*\bluesize, \redscale*\bluesize);
                    \fi
                \fi
            \fi

            \coordinate (blue11) at (\x + 0.2*\bluesize, \y + 0.2*\bluesize);
            \coordinate (blue21) at (\x + 0.2*\bluesize, \y + 0.5*\bluesize);
            \coordinate (blue31) at (\x + 0.2*\bluesize, \y + 0.8*\bluesize);
            \coordinate (blue12) at (\x + 0.5*\bluesize, \y + 0.2*\bluesize);
            \coordinate (blue22) at (\x + 0.5*\bluesize, \y + 0.5*\bluesize);
            \coordinate (blue32) at (\x + 0.5*\bluesize, \y + 0.8*\bluesize);            \coordinate (blue13) at (\x + 0.8*\bluesize, \y + 0.2*\bluesize);
            \coordinate (blue23) at (\x + 0.8*\bluesize, \y + 0.5*\bluesize);
            \coordinate (blue33) at (\x + 0.8*\bluesize, \y + 0.8*\bluesize);
            \coordinate (red12) at (\x + 0.5*\bluesize + 0.07*\bluesize, \y + 0.2*\bluesize);
            \coordinate (red22) at (\x + 0.5*\bluesize + 0.07*\bluesize, \y + 0.8*\bluesize);
            \coordinate (red11) at (\x + 0.35*\bluesize, \y + 0.2*\bluesize);
            \coordinate (red21) at (\x + 0.35*\bluesize, \y + 0.8*\bluesize);

            \node[draw=blue, fill=blue, circle, inner sep=0.5pt] at (blue11) {};
            \node[draw=blue, fill=blue, circle, inner sep=0.5pt] at (blue21) {};
            \node[draw=blue, fill=blue, circle, inner sep=0.5pt] at (blue31) {};
            \node[draw=blue, fill=blue, circle, inner sep=0.5pt] at (blue12) {};
            \node[draw=blue, fill=blue, circle, inner sep=0.5pt] at (blue22) {};
            \node[draw=blue, fill=blue, circle, inner sep=0.5pt] at (blue32) {};
            \node[draw=blue, fill=blue, circle, inner sep=0.5pt] at (blue13) {};
            \node[draw=blue, fill=blue, circle, inner sep=0.5pt] at (blue23) {};
            \node[draw=blue, fill=blue, circle, inner sep=0.5pt] at (blue33) {};
            \node[draw=red, fill=red, circle, inner sep=0.5pt] at (red11) {};
            \node[draw=red, fill=red, circle, inner sep=0.5pt] at (red21) {};
            \node[draw=red, fill=red, circle, inner sep=0.5pt] at (red12) {};
            \node[draw=red, fill=red, circle, inner sep=0.5pt] at (red22) {};

            \draw[black, thick, rounded corners=1mm] 
                ($(blue12)!0.5!(red12)$) ellipse [x radius=0.1, y radius=0.06];
            \draw[black, thick, rounded corners=1mm] 
                ($(blue32)!0.5!(red22)$) ellipse [x radius=0.1, y radius=0.06];
        }
    }
\end{tikzpicture}}
\caption{Implementation of gates in $\mathbf{H}_{2,1}$ of the SC code shown in Fig.~\ref{fig: RAA1}. Each AOD array is shifted two grids left and one grid up in a rotated layout. Qubit pairs enclosed within each black oval are positioned within the Rydberg distance, allowing the formation of four two-qubit gates in $\mathbf{H}_{2,1}$ of the corresponding replica via a Rydberg laser.}
\label{fig: RAA2}
\end{figure}

\color{black}

\subsection{Characteristic Function}
\label{subsec: general characteristic }

In \Cref{theo: Quantum 2D SC Codes}, we showed that a 2D-SC code is a stabilizer code if and only if conditions (\ref{eqn quantum 2d sc codes 1}) and (\ref{eqn quantum 2d sc codes 2}) are satisfied. We now use the language of two-dimensional convolution to translate these conditions into a simple algebraic form that facilitates the construction and analysis of a large class of codes. We begin with the example of the $d\times d$ toric code.

We represent the $d\times d$ toric code as the polynomial $\mathbf{F}(U,V)$ in two indeterminates $U,V$, with coefficients in $\mathcal{P}^{2\times 2}$ given by 
\begin{equation}\label{eqn: characteristic 1}
\mathbf{F}(U,V)=\mathbf{A}+\mathbf{B}V+\mathbf{C}U+\mathbf{D}UV,
\end{equation}
where $\mathbf{A}$, $\mathbf{B}$, $\mathbf{C}$, and $\mathbf{D}$ are given by (\ref{eqn: toric abcd}). Formal multiplication gives
\begin{equation}\label{eqn: characteristic 2}
\begin{split}
&\mathbf{F}(U,V)\mathbf{F}(U^{-1},V^{-1})^{\textrm{T}}\\
=&(\mathbf{A}+\mathbf{B}V+\mathbf{C}U+\mathbf{D}UV)\\
&(\mathbf{A}+\mathbf{B}V^{-1}+\mathbf{C}U^{-1}+\mathbf{D}U^{-1}V^{-1})^{\textrm{T}}\\
=&(\mathbf{A}\mathbf{A}^{\textrm{T}}+\mathbf{B}\mathbf{B}^{\textrm{T}}+\mathbf{C}\mathbf{C}^{\textrm{T}}+\mathbf{D}\mathbf{D}^{\textrm{T}})\\
&+(\mathbf{A}\mathbf{B}^{\textrm{T}}+\mathbf{C}\mathbf{D}^{\textrm{T}})V^{-1}+(\mathbf{B}\mathbf{A}^{\textrm{T}}+\mathbf{D}\mathbf{C}^{\textrm{T}})V\\
&+(\mathbf{A}\mathbf{C}^{\textrm{T}}+\mathbf{B}\mathbf{D}^{\textrm{T}})U^{-1}+(\mathbf{C}\mathbf{A}^{\textrm{T}}+\mathbf{D}\mathbf{B}^{\textrm{T}})U\\
&+\mathbf{A}\mathbf{D}^{\textrm{T}}U^{-1}V^{-1}+\mathbf{D}\mathbf{A}^{\textrm{T}}UV\\
&+\mathbf{B}\mathbf{C}^{\textrm{T}}VU^{-1}+\mathbf{C}\mathbf{B}^{\textrm{T}}UV^{-1}.\\
\end{split}
\end{equation}

There is a $1$--$1$ correspondence between terms on the right side of (\ref{eqn: characteristic 2}) and the conditions satisfied by $\mathbf{A}$, $\mathbf{B}$, $\mathbf{C}$, and $\mathbf{D}$ in \Cref{exam toric code theo: 1}, which implies that the $d\times d$ toric code is a stabilizer code.

We now extend the symplectic inner product to polynomials $\mathbf{F}(U,V)$ in two indeterminates $U,V$ with Pauli matrix coefficients. Let $\mathbf{P}(U,V)=\sum\nolimits_{i,j\in\mathbb{Z}} \mathbf{P}_{i,j} U^i V^j$ with $\mathbf{P}_{i,j} \in \mathcal{P}^{r_1\times n}$ and let $\mathbf{Q}(U,V)=\sum\nolimits_{i,j\in\mathbb{Z}} \mathbf{Q}_{i,j} U^i V^j$ with $\mathbf{Q}_{i,j} \in \mathcal{P}^{r_2\times n}$. 

\begin{defi} The symplectic inner product $\langle \mathbf{P},\mathbf{Q}\rangle_s$ is the polynomial in two indeterminates $\mathbf{U},\mathbf{V}$ with coefficients in $\mathbb{F}_2^{r_1\times r_2}$ given by:
\begin{equation}
\langle \mathbf{P},\mathbf{Q}\rangle_s=\sum\nolimits_{k,l\in \mathbb{Z}} \Biggl(\sum\nolimits_{i,j\in\mathbb{Z}} \langle \mathbf{P}_{i,j}, \mathbf{Q}_{k-i,l-j}\rangle_s \Biggr) U^k V^l.
\end{equation}
\end{defi}

\begin{rem} The five conditions satisfied by $\mathbf{A}$, $\mathbf{B}$, $\mathbf{C}$, and $\mathbf{D}$, that imply the toric code is a stabilizer code, reduce to the single condition:
\begin{equation}
\Bigl \langle \mathbf{F}(U,V), \mathbf{F}(U^{-1},V^{-1}) \Bigr \rangle_s=\mathbf{0}_{2 \times 2}.
\end{equation}
\end{rem}

\begin{defi} The \textbf{characteristic function} of a 2D-SC code is the polynomial
\begin{equation}
\mathbf{F}(U,V)=\sum\nolimits_{i=0}^{m_1}\sum\nolimits_{j=0}^{m_2} \mathbf{H}_{i,j} U^i V^j.
\end{equation}
\end{defi}
\begin{theo}\label{theo: condition characteristic function} A 2D-SC code is a stabilizer code if and only if the characteristic function $\mathbf{F}(U,V)$ satisfies:
\begin{equation}\label{eqn condition characteristic function}
\Bigl \langle \mathbf{F}(U,V), \mathbf{F}(U^{-1},V^{-1}) \Bigr \rangle_s=\mathbf{0}_{r \times r}.
\end{equation}
\end{theo}

\begin{proof} We rewrite conditions (\ref{eqn quantum 2d sc codes 1}) and (\ref{eqn quantum 2d sc codes 2}) of \Cref{theo: Quantum 2D SC Codes} as
\begin{widetext}
\begin{equation}
\sum\nolimits_{i=\max\{0,-d_1\}}^{\min\{m_1-d_1,m_1\}}\sum\nolimits_{j=\max\{0,-d_2\}}^{\min\{m_2-d_2,m_2\}} \langle \mathbf{H}_{i,j}, \mathbf{H}_{i+d_1,j+d_2} \rangle=\mathbf{0}_{r \times r}.
\end{equation}
\end{widetext}
Now,
\begin{widetext}
\begin{equation} 
\begin{split}
& \Bigl \langle \mathbf{F}(U,V), \mathbf{F}(U^{-1},V^{-1}) \Bigr \rangle_s \\
=& \Bigl \langle  \sum\nolimits_{i=0}^{m_1}\sum\nolimits_{j=0}^{m_2} \mathbf{H}_{i,j} U^i V^j, \sum\nolimits_{i=0}^{m_1}\sum\nolimits_{j=0}^{m_2} \mathbf{H}_{i,j} U^{-i} V^{-j} \Bigr \rangle_s \\
=& \sum\nolimits_{d_1=-m_1}^{m_1}\sum\nolimits_{d_2=-m_2}^{m_2} \Bigl(\sum\nolimits_{i=\max\{0,-d_1\}}^{\min\{m_1-d_1,m_1\}}\sum\nolimits_{j=\max\{0,-d_2\}}^{\min\{m_2-d_2,m_2\}} \langle \mathbf{H}_{i,j}, \mathbf{H}_{i+d_1,j+d_2} \rangle_s \Bigr) U^{-d_1}V^{-d_2}\\
=& \mathbf{0}_{r \times r} \\
\end{split}
\end{equation}
\end{widetext}
The result follows from \Cref{theo: Quantum 2D SC Codes}.
\end{proof}

Condition (\ref{eqn condition characteristic function}) in \Cref{theo: condition characteristic function} facilitates the construction of 2D-SC codes that are stabilizer codes, and we provide two examples below.
\begin{widetext}
\begin{equation}\label{eqn: other example base}
\begin{split}
\mathbf{F}(U,V)=\left[\begin{array}{ccc}
X(1+U) & Y(1+V) & Z\\
Y(UV) & Z(V) & X(1+UV)\\
Z(U+UV+U^2V^2) & X(V+UV^2+U^2V^2) & Y(U^2V+U)\\
\end{array}\right].
\end{split}
\end{equation}
\begin{equation}
\begin{split}
\mathbf{F}(U,V)=\left[\begin{array}{cccc}
Y(1+UV) &0 & Z(U+V) & X(1+U)\\
0 & Z(1+UV) & X(V+UV) & Y(U+V) \\
X(V+UV) & Y(U+V) & 0 & Z(1+UV)\\
Z(U+V) & X(1+U) & Y(1+UV) & 0 \\
\end{array}\right].
\end{split}
\end{equation}
\end{widetext}

We now explain how characteristic functions facilitate the analysis of minimum distance of 2D-SC codes with coupling lengths $L_1$ and $L_2$. We define the characteristic function of a codeword $\mathbf{c}\in\mathcal{P}^{nL_1L_2}$ by 
\begin{equation}
\mathbf{c}\left(U,V\right)=\sum_{i=0}^{L_1-1}\sum_{j=0}^{L_2-1} \mathbf{c}_{i,j}U^iV^j,
\end{equation}
where each $\mathbf{c}_{i,j}$ is a sub-sequence of $\mathbf{c}$ with length $n$, and $\mathbf{c}=(\mathbf{c}_0,\mathbf{c}_1,\dots,\mathbf{c}_{L_1-1})$ where $\mathbf{c}_i=(\mathbf{c}_{i,0},\mathbf{c}_{i,1},\dots,\mathbf{c}_{i,L_2-1})$, $i=0,1,\dots,L_1-1$. The weight of the codeword (characteristic function) is the number of non-zero terms in $\mathbf{c}(U,V)$.

We also need to define characteristic vectors of vectors in the stabilizer group. Let $\mathbf{s}\in \mathbb{F}_2^{rL_1L_2}$ be an indicator function of the stabilizer $\mathbf{c}=\prod\nolimits_{i:s_i=1}\mathbf{g}_i$, where $\mathbf{g}_1,\mathbf{g}_2,\dots,\mathbf{g}_{rL_1L_2}$ are the generators in the parity check matrix of a SC-QLDPC code. Suppose $\mathbf{s}$ is divided into $L_1L_2$ sub-sequences of length $r$ by $\mathbf{s}=(\mathbf{s}_0,\mathbf{s}_1,\dots,\mathbf{s}_{L_1-1})$ where $\mathbf{s}_i=(\mathbf{s}_{i,0},\mathbf{s}_{i,1},\dots,\mathbf{s}_{i,L_2-1})$, $i=0,1,\dots,L_1-1$. Then the characteristic function of the stabilizer $\mathbf{c}$ can be represented by 
\begin{equation}
    \mathbf{c}\left(U,V\right)=\mathbf{s}(U,V) \mathbf{F}(U^{-1},V^{-1}),
\end{equation}
where $U^{L_1}=V^{L_2}=1$ and
\begin{equation}
    \mathbf{s}(U,V)=\sum_{i=0}^{L_1-1}\sum_{j=0}^{L_2-1} \mathbf{s}_{i,j}U^iV^j.
\end{equation}

\begin{theo}\label{theo: distances} The minimum distance of a 2D-SC code with characteristic function $\mathbf{F}(U,V)$ is the minimum weight of a characteristic function $\mathbf{c}(U,V)$ that satisfies:
\begin{widetext}
\begin{equation}
\bigg\langle\mathbf{F}\left(U,V\right),\mathbf{c}\left(U,V\right)\bigg\rangle=\mathbf{0}_{r\times 1},\ \mathbf{c}\left(U,V\right)\notin\Bigl\{\mathbf{s}(U,V)\mathbf{F}(U^{-1},V^{-1})\Big| \mathbf{s}(U,V)\in \mathbb{F}_2\left[U,V\right]/\bigl\langle U^{L_1},V^{L_2}\bigr\rangle\Bigr\}.
\end{equation}
\end{widetext}
\end{theo}

\begin{proof} The requirement that a codeword commutes with every stabilizer generator reduces to
\begin{widetext}
\begin{equation}
\sum\nolimits_{k=0}^{m_1}\sum\nolimits_{l=0}^{m_2}\mathbf{H}_{k,l}\mathbf{c}_{k+i,l+j}=\mathbf{0}_{r\times 1},\ 
\forall 0\leq i\leq L_1-1, 0\leq j\leq L_2-1.
\end{equation}  
\end{widetext}
Now,
\begin{widetext}
\begin{equation} \label{eqn: minimum distance}
\begin{split}
\bigg\langle\mathbf{F}\left(U,V\right),\mathbf{c}\left(U,V\right)\bigg\rangle_s=&\bigg\langle\sum\nolimits_{i=0}^{m_1}\sum\nolimits_{j=0}^{m_2} \mathbf{H}_{i,j} U^iV^j,\sum\nolimits_{i=0}^{L_1-1}\sum\nolimits_{j=0}^{L_2-1} \mathbf{c}_{i,j}U^{i}V^{j}\bigg\rangle_s\\
=& \sum\nolimits_{i=0}^{L_1-1}\sum\nolimits_{j=0}^{L_2-1} \Bigg( \sum\nolimits_{k=0}^{m_1}\sum\nolimits_{l=0}^{m_2} \mathbf{H}_{k,l}\mathbf{c}_{i-k,j-l}\Bigg) U^{i}V^{j}\\
=& \mathbf{0}_{r\times 1},
\end{split}
\end{equation}
\end{widetext}
and the theorem follows from (\ref{eqn: minimum distance}).
\end{proof}

\begin{exam} The characteristic function of the $d\times d $ toric code is given by
\begin{equation}
\mathbf{F}\left(U,V\right)=\left[\begin{array}{cc}
X(1+V) & X(V+U)\\
Z(V+U) & Z(1+V)U
\end{array}\right].
\end{equation}
The characteristic function $\mathbf{c}\left(U,V\right)$ of a codeword satisfies
\begin{widetext}
\begin{equation}\label{eqn: exam dis condi}
\begin{split}
\bigg\langle\mathbf{F}\left(U,V\right),\mathbf{c}\left(U,V\right)\bigg\rangle_s=&\bigg\langle\mathbf{F}\left(U,V\right),\left[\begin{array}{cc}
Xf_X+Yf_Y+Zf_Z & Xg_X+Yg_Y+Zg_Z
\end{array}\right]\bigg\rangle_s\\
=&\left[\begin{array}{c}
(1+V)(f_Y+f_Z)+(V+U)(g_Y+g_Z)\\
(V+U)(f_X+f_Y)+(1+V)U(g_X+g_Y)
\end{array}\right]=\mathbf{0}_{r\times1},
\end{split}
\end{equation}
\end{widetext}
where $f_X,f_Y,f_Z$ and $g_X,g_Y,g_Z$ are polynomials in indeterminates $U,V$ with binary coefficients.

We notice that $(f_X,f_Y,f_Z)=(UV,U,V)$, $(g_X,g_Y,g_Z)=(U,V,1)$ satisfy (\ref{eqn: exam dis condi}) and is thus a codeword with characteristic function $\mathbf{c}(U,V)=\left[XUV+YU+ZV,XU+YV+Z\right]$.

Set $\mathbf{s}(U,V)=\left[UV, UV\right]$, and observe that $\mathbf{c}(U,V)=\mathbf{s}(U,V)\mathbf{F}\left(U^{-1},V^{-1}\right)$. It follows that the associated codeword is a stabilizer.
\end{exam}

Given a polynomial $\mathbf{F}(U,V)$ with coefficients in $\mathcal{P}^{r\times n}$, let $m_1,m_2$ respectively be the maximal degrees of the indeterminates $U,V$. We replace $U^sV^t$ in $\mathbf{F}(U,V)$ by $U^{m_1-s}V^{m_2-t}$ to obtain the \textbf{complementary polynomial} $\bar{\mathbf{F}}(U,V)$. We observe 
\begin{equation}\label{eqn: complementary}
    \bar{\mathbf{F}}(U,V)=U^{m_1}V^{m_2} \mathbf{F}(U^{-1},V^{-1}).
\end{equation}

\begin{exam}\label{cons: HGP} \emph{(2D-SC HGP Codes)} Given $\mathbf{A}(U,V)\in \mathbb{F}_2^{r_1\times n_1}\left[U,V\right]$ and $\mathbf{B}(U,V)\in \mathbb{F}_2^{r_2\times n_2}\left[U,V\right]$, define:
\begin{widetext}
\begin{equation}
\mathbf{F}(U,V)=\left[\begin{array}{cc}
X(\mathbf{I}_{n_2\times n_2}\otimes \mathbf{A}(U,V)) & X(\bar{\mathbf{B}}(U,V)^T \otimes \mathbf{I}_{r_1 \times r_1})\\
\vspace{2pt}
Z(\mathbf{B}(U,V) \otimes \mathbf{I}_{n_1 \times n_1}) & Z(\mathbf{I}_{r_2\times r_2}\otimes \bar{\mathbf{A}}(U,V)^T)
\end{array}\right].
\end{equation}
\end{widetext}
\end{exam}

We now show that the 2D-SC hypergraph product code is a stabilizer code by verifying that 
\begin{widetext}
\begin{equation}
\begin{split}
    \Bigl \langle \mathbf{F}(U,V), \mathbf{F}(U^{-1},V^{-1}) \Bigr \rangle_s=\mathbf{0}_{(r_1n_2+r_2n_1)\times(r_1n_2+r_2n_1)}.
    \end{split}
\end{equation}
\end{widetext}
We have
\begin{widetext}
\begin{equation*}
\begin{split}
&\Bigl \langle \mathbf{F}(U,V), \mathbf{F}(U^{-1},V^{-1}) \Bigr \rangle_s= \left[\begin{array}{c|c}
\mathbf{0}_{r_1n_2\times r_1n_2} & \begin{array}{r}\mathbf{B}(U^{-1},V^{-1})^T \otimes \mathbf{A}(U,V)\\+\bar{\mathbf{B}}(U,V)^T \otimes \bar{\mathbf{A}}(U^{-1},V^{-1})\end{array}\\
\hline
\begin{array}{r}\mathbf{B}(U,V) \otimes \mathbf{A}(U^{-1},V^{-1})^T\\+\bar{\mathbf{B}}(U^{-1},V^{-1}) \otimes \bar{\mathbf{A}}(U,V)^T \end{array}& \mathbf{0}_{r_2n_1\times r_2n_1}
\end{array}\right].\\
\end{split}
\end{equation*}
\end{widetext}

Given a polynomial $\mathbf{F}(U,V)$ with coefficients in $\mathcal{P}^{r\times n}$, let $m_1,m_2$ respectively be the maximal degrees of the indeterminates $U,V$. We replace $U^sV^t$ in $\mathbf{F}(U,V)$ by $U^{m_1-s}V^{m_2-t}$ to obtain the \textbf{complementary polynomial} $\bar{\mathbf{F}}(U,V)$. We observe 
\begin{equation}\label{eqn: complementary}
    \bar{\mathbf{F}}(U,V)=U^{m_1}V^{m_2} \mathbf{F}(U^{-1},V^{-1}).
\end{equation}

\begin{exam}\label{cons: HGP} \emph{(2D-SC HGP Codes)} Given $\mathbf{A}(U,V)\in \mathbb{F}_2^{r_1\times n_1}\left[U,V\right]$ and $\mathbf{B}(U,V)\in \mathbb{F}_2^{r_2\times n_2}\left[U,V\right]$, define:
\begin{widetext}
\begin{equation}
\mathbf{F}(U,V)=\left[\begin{array}{cc}
X(\mathbf{I}_{n_2\times n_2}\otimes \mathbf{A}(U,V)) & X(\bar{\mathbf{B}}(U,V)^T \otimes \mathbf{I}_{r_1 \times r_1})\\
\vspace{2pt}
Z(\mathbf{B}(U,V) \otimes \mathbf{I}_{n_1 \times n_1}) & Z(\mathbf{I}_{r_2\times r_2}\otimes \bar{\mathbf{A}}(U,V)^T)
\end{array}\right].
\end{equation}
\end{widetext}

\end{exam}

We now show that the 2D-SC hypergraph product code is a stabilizer code by verifying that $\Bigl \langle \mathbf{F}(U,V), \mathbf{F}(U^{-1},V^{-1}) \Bigr \rangle_s=\mathbf{0}_{(r_1n_2+r_2n_1)\times(r_1n_2+r_2n_1)}$. We have
\begin{widetext}
\begin{equation*}
\begin{split}
&\Bigl \langle \mathbf{F}(U,V), \mathbf{F}(U^{-1},V^{-1}) \Bigr \rangle_s= \left[\begin{array}{c|c}
\mathbf{0}_{r_1n_2\times r_1n_2} & \begin{array}{r}\mathbf{B}(U^{-1},V^{-1})^T \otimes \mathbf{A}(U,V)\\+\bar{\mathbf{B}}(U,V)^T \otimes \bar{\mathbf{A}}(U^{-1},V^{-1})\end{array}\\
\hline
\begin{array}{r}\mathbf{B}(U,V) \otimes \mathbf{A}(U^{-1},V^{-1})^T\\+\bar{\mathbf{B}}(U^{-1},V^{-1}) \otimes \bar{\mathbf{A}}(U,V)^T \end{array}& \mathbf{0}_{r_2n_1\times r_2n_1}
\end{array}\right].\\
\end{split}
\end{equation*}
\end{widetext}

It follows from (\ref{eqn: complementary}) that the off diagonal matrices vanish. At the bottom left we obtain
\begin{widetext}
\begin{equation}
    \mathbf{B}(U,V)\otimes \mathbf{A}(U^{-1},V^{-1})^{\textrm{T}}+U^{-m_1}V^{-m_2}\mathbf{B}(U,V)\otimes \Bigl(U^{m_1}V^{m_2}\mathbf{A}(U^{-1},V^{-1})^{\textrm{T}}\Bigr)=\mathbf{0}_{r_2n_1\times n_1r_2}.
\end{equation}
\end{widetext}
At the top right we obtain
\begin{widetext}
\begin{equation}
    \mathbf{B}(U^{-1},V^{-1})^{\textrm{T}}\otimes \mathbf{A}(U,V)+U^{m_1}V^{m_2}\mathbf{B}(U^{-1},V^{-1})^{\textrm{T}}\otimes \Bigl(U^{-m_1}V^{-m_2}\mathbf{A}(U,V)\Bigr)=\mathbf{0}_{r_1n_2\times r_2n_1}.
\end{equation}
\end{widetext}

\begin{rem}\label{rem: QA-LP Codes}(\textbf{Quasi-abelian lifted product codes are finite-dimensional SC-HGP codes})
    Observe that the aforementioned construction resembles the representation of lifted product codes. Using the fact that every finitely generated abelian group is a direct product of cyclic groups, quasi-abelian lifted product (LP) codes proposed in \cite{panteleev2022liftedproduct} can be viewed as finite-dimensional SC-HGP codes. This means that it is possible to efficiently construct high-performance quasi-abelian LP codes through cycle optimization methods developed for classical SC-LDPC codes. \Cref{sec: decoder} provides a framework for the 2D case.
\end{rem}

\begin{cons}\label{cons: xyz} (\textbf{2D-SC XYZ Codes}) Given $\mathbf{A}(U,V)\in \mathbb{F}_2^{r_1\times n_1}\left[U,V\right]$, $\mathbf{B}(U,V)\in \mathbb{F}_2^{r_2\times n_2}\left[U,V\right]$, and $\mathbf{C}(U,V)\in \mathbb{F}_2^{r_3\times n_3}\left[U,V\right]$, define:
\begin{widetext}
\begin{equation}
\mathbf{F}(U,V)=\left[\begin{array}{llll}
X(\mathbf{A}\otimes\mathbf{I}_{n_2}\otimes\mathbf{I}_{n_3}) & I_{r_1n_2n_3\times n_1r_2r_3} & Z(\mathbf{I}_{r_1}\otimes\mathbf{I}_{n_2}\otimes\bar{\mathbf{C}}^T) & Y(\mathbf{I}_{r_1}\otimes\bar{\mathbf{B}}^T\otimes\mathbf{I}_{n_3}) \\
Y(\mathbf{I}_{n_1}\otimes\mathbf{B}\otimes\mathbf{I}_{n_3}) & Z(\mathbf{I}_{n_1}\otimes\mathbf{I}_{r_2}\otimes\bar{\mathbf{C}}^T) & I_{n_1r_2n_3\times r_1n_2r_3} & X(\bar{\mathbf{A}}^T\otimes\mathbf{I}_{r_2}\otimes\mathbf{I}_{n_3}) \\
Z(\mathbf{I}_{n_1}\otimes\mathbf{I}_{n_2}\otimes\mathbf{C}) & Y(\mathbf{I}_{n_1}\otimes\bar{\mathbf{B}}^T\otimes\mathbf{I}_{r_3}) & X(\bar{\mathbf{A}}^T\otimes\mathbf{I}_{n_2}\otimes\mathbf{I}_{r_3}) & I_{n_1n_2r_3\times r_1r_2n_3} \\
I_{r_1r_2r_3\times n_1n_2n_3} & X(\mathbf{A}\otimes\mathbf{I}_{r_2}\otimes\mathbf{I}_{r_3}) & Y(\mathbf{I}_{r_1}\otimes\mathbf{B}\otimes\mathbf{I}_{r_3}) & Z(\mathbf{I}_{r_1}\otimes\mathbf{I}_{r_2}\otimes\mathbf{C}) 
\end{array}\right].
\end{equation}
\end{widetext}
The 2D-SC XYZ code obtained from $\mathbf{F}(U,V)$ is a stabilizer code.
\end{cons}

\begin{exam}\label{exam: other example lifting} Let $e,g,y\in\mathbb{N}^*$, with $\sigma=\sigma_z$ as defined in (\ref{eqn: circulant shift matrix}). $\mathbf{F}(U,V)$ below is a lifting of the characteristic function in (\ref{eqn: other example base}).
\begin{widetext}
\begin{equation}
\begin{split}
&\mathbf{F}(U,V)\\
=&\left[\begin{array}{ccc}
X(\sigma^{e+y}+\sigma^{2e-g}U) & Y(\sigma^e+\sigma^eV) & Z\sigma^e\\
Y(\sigma^e UV) & Z(\sigma^gV) & X(\sigma^g+\sigma^{e-y}UV)\\
Z(\sigma^{g+y}U+\sigma^{g+y}UV+\sigma^eU^2V^2) & X(\sigma^{2g-e+y}V+\sigma^{g}UV^2+\sigma^{e-y}U^2V^2) & Y(\sigma^{e-y}U^2V+\sigma^gU)\\
\end{array}\right].
\end{split}
\end{equation}
\end{widetext}
\end{exam}

\begin{rem} (\textbf{SC-Generalized Bicycle Codes as 1D-SC-QLDPC Codes}) Note that we can also specify characteristic functions of SC-QLDPC codes with arbitrary dimensions by changing the number of indeterminates (which should be equal to the dimension). Several codes in the literature could be categorized as 1D-SC-QLDPC codes, including Generalized Bicycle (GB) codes. The characteristic function of a GB code has the general form:
\begin{equation}
    \mathbf{F}(U)=\left[\begin{array}{cc}
      Xa(U)   & Xb(U) \\
      Z U^Lb(U^{-1})  & ZU^La(U^{-1})
    \end{array}\right].
\end{equation}

In particular, the characteristic function of the $\left[\left[126,28\right]\right]$ GB code in \cite{panteleev2021degenerate} is specified as follows with $m=62$ and $L=63$:
\begin{widetext}
\begin{equation}
\begin{split}
   \mathbf{F}(U)&=\left[\begin{array}{cc}
      X(1+U+U^{14}+U^{16}+U^{22})   & X(1+U^3+U^{13}+U^{20}+U^{42}) \\
      Z(1+U^{63-3}+U^{63-13}+U^{63-20}+U^{63-42} )   & Z(1+U^{63-1}+U^{63-14}+U^{63-16}+U^{63-22})
    \end{array}\right]\\
    &=\left[\begin{array}{cc}
      X(1+U+U^{14}+U^{16}+U^{22})   & X(1+U^3+U^{13}+U^{20}+U^{42}) \\
      Z(1+U^{60}+U^{50}+U^{43}+U^{21} )   & Z(1+U^{62}+U^{49}+U^{47}+U^{41})
    \end{array}\right].
\end{split}
\end{equation}
\end{widetext}
\end{rem}

\section{Code Optimization}
\label{sec: decoder}

In this section, we develop a framework for optimizing the SC-HGP codes introduced in \Cref{subsec: general characteristic }. We start from the alternating sum condition that specifies when a closed path in the base matrix/protograph gives rise to a closed path/cycle in the protograph/Tanner graph. We then develop an optimization algorithm that efficiently reduces the number of short cycles.

\subsection{Cycle Conditions}
\label{subsec: cycle condition}

In classical coding theory, short cycles in Tanner graphs create problems for belief propagation (BP) decoders in both the waterfall and error floor regions. Reducing the number of cycles in the Tanner graph is a major goal in optimizing LDPC codes, but efficiently reducing the number of cycles in general LDPC codes is not an easy task. Progressive edge growth (PEG) is an algorithm that efficiently improves the \textbf{girth} (the minimum length of a cycle) of the Tanner graph. However, girth is not the only metric that determines performance since multiplicities of cycles are also important. For example, a code with girth $6$ than contains $1000$ cycles-$6$ can perform worse than a code with $1$ cycle-$4$ and $2$ cycles-$6$ if the channel is good enough. The convolutional structure of SC codes implies that the number of cycles can be represented by the base matrix, partitioning matrix, and the lifting matrix in a compact way, which facilitates enumeration and optimization of short cycles. 

In particular, each cycle in the Tanner graph is lifted from a closed path in the protograph, and the latter arises from another closed path in the base matrix. Closed paths are cycles that allow repetitive nodes, they are also referred to as \textbf{cycle candidates}. Whether a cycle candidate in the base matrix/protograph gives rise to a cycle in the protograph/Tanner graph is determined by the alternating sum of entries along the cycle in the partitioning/lifting matrix.

\begin{exam} The matrix $\mathbf{H}$ given in (\ref{eqn: cyc_condition_exam}) specifies a 1D-TB code with $m=2$ and $L=4$. The base matrix $\mathbf{H}_0+\mathbf{H}_1+\mathbf{H}_2$ is the $3\times 4$ matrix with each entry equal to $1$. 
Entries $0,1,2$ in the partitioning matrix $\mathbf{P}$ give rise to a $1$ at the same location in the component matrices $\mathbf{H}_0,\mathbf{H}_1$, and $\mathbf{H}_2$, respectively. We highlight cycle candidates of lengths $4$, $6$, and $8$ (also referred to as cycles-$4$, cycles-$6$, and cycles-$8$ candidates in the rest of the paper) in blue, red, and green respectively.
\begin{widetext}
    \begin{equation}\label{eqn: cyc_condition_exam}
\mathbf{H} = \left[\begin{array}{c|c|c|c}
\begin{array}{cccc}
\cs\tikznode{64}{$1$} & \cs\tikznode{63}{$1$} & 0 & 0 \\
0 & 0 & 1 & 0 \\
0 & 0 & 0 & \cl\tikznode{11}{$1$} \\
\end{array} & \begin{array}{cccc}
0 & 0 & 0 & 0 \\
0 & 0 & 0 & 0 \\
0 & 0 & 0 & 0 \\
\end{array} & \begin{array}{cccc}
0 & 0 & \cl\tikznode{43}{$1$} & 0 \\
0 & \ce\tikznode{85}{$1$} & 0 & 0 \\
0 & \ce\tikznode{86}{$1$} & 0 & 0 \\
\end{array} & \begin{array}{cccc}
0 & 0 & 0 & \cl\tikznode{42}{$1$} \\
\ce\tikznode{84}{$1$} & 0 & 0 & 1 \\
1 & 0 & \ce\tikznode{14}{$1$} & 0 \\
\end{array} \\
\hline
\begin{array}{cccc}
0 & 0 & 0 & \cl\tikznode{12}{$1$} \\
1 & 0 & 0 & 1 \\
\cs\tikznode{65}{$1$} & 0 & \cl\tikznode{24}{$1$} & 0 \\
\end{array} & \begin{array}{cccc}
1 & 1 & 0 & 0 \\
0 & 0 & 1 & 0 \\
0 & 0 & 0 & \cs\tikznode{21}{$1$} \\
\end{array} & \begin{array}{cccc}
0 & 0 & 0 & 0 \\
0 & 0 & 0 & 0 \\
0 & 0 & 0 & 0 \\
\end{array} & \begin{array}{cccc}
0 & 0 & \cl\tikznode{13}{$1$} & 0 \\
0 & 1 & 0 & 0 \\
0 & 1 & 0 & 0 \\
\end{array} \\
\hline
\begin{array}{cccc}
0 & 0 & \cl\tikznode{23}{$1$} & 0 \\
0 & \cs\tikznode{62}{$1$} & 0 & 0 \\
0 & 1 & 0 & 0 \\
\end{array} & \begin{array}{cccc}
0 & 0 & 0 & \cl\tikznode{22}{$1$} \\
1 & 0 & 0 & \cs\tikznode{61}{$1$} \\
1 & 0 & \cl\tikznode{34}{$1$} & 0 \\
\end{array} & \begin{array}{cccc}
1 & 1 & 0 & 0 \\
0 & 0 & 1 & 0 \\
0 & 0 & 0 & \cl\tikznode{31}{$1$} \\
\end{array} & \begin{array}{cccc}
0 & 0 & 0 & 0 \\
0 & 0 & 0 & 0 \\
0 & 0 & 0 & 0 \\
\end{array} \\
\hline
\begin{array}{cccc}
0 & 0 & 0 & 0 \\
0 & 0 & 0 & 0 \\
0 & 0 & 0 & 0 \\
\end{array} & \begin{array}{cccc}
0 & 0 & \cl\tikznode{33}{$1$} & 0 \\
0 & 1 & 0 & 0 \\
0 & 1 & 0 & 0 \\
\end{array} & \begin{array}{cccc}
0 & 0 & 0 & \ce\tikznode{32}{$1$} \\
1 & 0 & 0 & \ce\tikznode{82}{$1$} \\
1 & 0 & \cl\tikznode{44}{$1$} & 0 \\
\end{array} & \begin{array}{cccc}
\ce\tikznode{83}{$1$} & 1 & 0 & 0 \\
0 & 0 & \ce\tikznode{81}{$1$} & 0 \\
0 & 0 & 0 & \cl\tikznode{41}{$1$} \\
\end{array}
\end{array}\right],\hspace{10pt}
\begin{array}{rl}
     &\mathbf{P}=\left[\begin{array}{cccc}
     0 & 0 & \cl\tikznode{f1}{$2$} & \cl\tikznode{f2}{$1$} \\
     1 & 2 & 0 & 1 \\
     1 & 2 & \cl\tikznode{f4}{$1$} & \cl\tikznode{f3}{$0$} \\
     \end{array}\right]\\
     &\\
     &\mathbf{P}=\left[\begin{array}{cccc}
     \cs\tikznode{s1}{$0$} & \cs\tikznode{s2}{$0$} & 2 & 1 \\
     1 & \cs\tikznode{s3}{$2$} & 0 & \cs\tikznode{s4}{$1$} \\
     \cs\tikznode{s6}{$1$} & 2 & 1 & \cs\tikznode{s5}{$0$} \\
     \end{array}\right].\\
     &\\
     &\mathbf{P}=\left[\begin{array}{cccc}
     \ce\tikznode{ei1}{$0$} & 0 & 2 & \ce\tikznode{ei2}{$1$} \\
     \ce\tikznode{ei8}{$1$} & \ce\tikznode{ei7}{$2$} & \ce\tikznode{ei4}{$0$} & \ce\tikznode{ei3}{$1$} \\
     1 & \ce\tikznode{ei6}{$2$} & \ce\tikznode{ei5}{$1$} & 0 \\
     \end{array}\right]
\end{array}
\end{equation}

\begin{tikzpicture}[remember picture,overlay,blue,rounded corners]
\draw[draw=blue,thick] (11)--(12)--(13)--(14)--(11);
\draw[draw=blue,thick] (21)--(22)--(23)--(24)--(21);
\draw[draw=blue,thick] (31)--(32)--(33)--(34)--(31);
\draw[draw=blue,thick] (41)--(42)--(43)--(44)--(41);
\draw[draw=mygreen,thick] (81)--(82)--(32)--(83)--(84)--(85)--(86)--(14)--(81);
\draw[draw=red,thick] (61)--(62)--(63)--(64)--(65)--(21)--(61);
\draw[draw=blue,thick] (f1)--(f2)--(f3)--(f4)--(f1);
\draw[draw=mygreen,thick] (ei1)--(ei2)--(ei3)--(ei4)--(ei5)--(ei6)--(ei7)--(ei8)--(ei1);
\draw[draw=red,thick] (s1)--(s2)--(s3)--(s4)--(s5)--(s6)--(s1);
\end{tikzpicture}

\end{widetext}
The alternating sums of the blue, red, and green cycle candidates are
\begin{equation*}
\begin{split}
&2-1+0-1=0,\\
&0-0+2-1+0-1=0,\\
&0-1+1-0+1-2+2-1=0.
\end{split}
\end{equation*}
All give rise to cycles in the Tanner graph, and each cycle candidate gives rise to four cycles.
\end{exam}

\Cref{lemma: cycle condition} provides a necessary and sufficient condition for a cycle candidate in the base matrix to give rise to a cycle in the protograph (see \cite{battaglioni2017design,fossorier2004quasicyclic} for more details). Since we have not constrained the SC-QLDPC codes to be quasi-cyclic, the protograph is exactly the Tanner graph.

\begin{lemma}\label{lemma: cycle condition} (\textbf{alternating sum condition}) Let $g\geq 2$, and let $(j_1,i_1,j_2,i_2,\dots,j_g,i_g)$ be a cycle-$2g$ candidate $C$ in the base matrix, where $j_{g+1}=j_1$ and $(i_{k},j_{k})$, $(i_{k},j_{k+1})$, $1\leq k\leq g$, are nodes of $C$ in the partitioning matrix $\mathbf{P}$. Then $C$ becomes a cycle-$2g$ in the protograph if and only if
\begin{equation}
\begin{split}
\sum\nolimits_{k=1}^{g}\Bigl((\mathbf{P})_{i_{k},j_{k}}-(\mathbf{P})_{i_{k},j_{k+1}}\Bigr)&=0
\end{split}
\end{equation}
for NTB codes, or
\begin{equation}\label{eqn: partition cycle}
\begin{split}
\sum\nolimits_{k=1}^{g}\Bigl((\mathbf{P})_{i_{k},j_{k}}-(\mathbf{P})_{i_{k},j_{k+1}}\Bigr)&=0\ (\hspace*{-1em}\mod L)
\end{split}
\end{equation}
for TB codes.
\end{lemma}

The $k$-th term in the alternating sum is the horizontal displacement of node $2k$ with respect to node $(2k-1)$ (interpreted $\mod L$ for TB codes) on $C$. It therefore follows that when $C$ is a closed path, the alternating sum is zero. For a TB 2D-SC code, the alternating sum of the terms $(\mathbf{P})_{i_k,j_k}-(\mathbf{P})_{i_k,j_{k+1}}$ is required to be $(0,0)\mod (L_1,L_2)$.

\subsection{Optimization of SC-QLDPC Codes}
\label{subsec: optimization of SC-QLDPC}

In this section, we focus on optimization of cycles of length $4,6,8$ in 2D-SC-HGP codes (essentially quasi-abelian lifted product codes on a group that is a product of two cyclic groups of sizes $L_1$ and $L_2$) introduced in \Cref{sec: quantum 2d-sc codes}. Recall that the characteristic function $\mathbf{F}(U,V)$ is given by 
\begin{widetext}
\begin{equation}\label{eqn: cyc opt sc hgp}
\mathbf{F}(U,V)=\left[\begin{array}{cc}
X(\mathbf{I}_{n_2\times n_2}\otimes \mathbf{A}(U,V)) & X(\bar{\mathbf{B}}(U,V)^T \otimes \mathbf{I}_{r_1 \times r_1})\\
\vspace{2pt}
Z(\mathbf{B}(U,V) \otimes \mathbf{I}_{n_1 \times n_1}) & Z(\mathbf{I}_{r_2\times r_2}\otimes \bar{\mathbf{A}}(U,V)^T)
\end{array}\right],
\end{equation}
\end{widetext}
where $\mathbf{A}(U,V)\in \mathbb{F}_2^{r_1\times n_1}\left[U,V\right]$ and $\mathbf{B}(U,V)\in \mathbb{F}_2^{r_2\times n_2}\left[U,V\right]$. 

We now illustrate how to derive the partitioning matrix $\mathbf{P}$ from a characteristic function $\mathbf{F}(U,V)$ of the form (\ref{eqn: cyc opt sc hgp}). Note that $\mathbf{F}(U,V)$ is completely determined by $\mathbf{A}(U,V)$ and $\mathbf{B}(U,V)$ since $\bar{\mathbf{A}}(U,V)$ and $\bar{\mathbf{B}}(U,V)$ are simply $\mathbf{A}(U,V)$ and $\mathbf{B}(U,V)$, with $U^iV^j$ replaced by $U^{m_1-i}V^{m_2-j}$. For simplicity, we consider an example where the $(s,t)$-th entry of $\mathbf{F}(U,V)$ is a monomial $F_{s,t} U^iV^j$. The $(s,t)$-th entry of the partitioning matrix $\mathbf{P}$ is then $(i,j)$, and $\mathbf{P}$ is completely determined by two partitioning matrices $\mathbf{P}_a$ and $\mathbf{P}_b$ corresponding to $\mathbf{A}(U,V)$ and $\mathbf{B}(U,V)$ respectively ($\mathbf{P}$ is essentially the hypergraph product of $\mathbf{P}_a$ and $\mathbf{P}_b$).

\begin{exam} Here $r_1=r_2=2$, $n_1=n_2=3$, with base matrices
\begin{equation*}
    \mathbf{A}=\begin{bmatrix}
1&0&1\\
1&1&0
\end{bmatrix}\textrm{ and }\mathbf{B}=\begin{bmatrix}
1&1&0\\
0&1&1
\end{bmatrix}.
\end{equation*}
The characteristic functions are given by
\begin{equation*}
\mathbf{A}(U,V)=\begin{bmatrix}
U^{a_1}V^{a_1'} & 0  & U^{a_2}V^{a_2'} \\
U^{a_3}V^{a_3'} & U^{a_4}V^{a_4'} &  0
\end{bmatrix},
\end{equation*}
\begin{equation*}
\mathbf{B}(U,V)=\begin{bmatrix}
U^{b_1}V^{b_1'} & U^{b_2}V^{b_2'} &  0 \\
0 & U^{b_3}V^{b_3'} & U^{b_4}V^{b_4'}
\end{bmatrix}.
\end{equation*}
The two partitioning matrices $\mathbf{P}_a$ and $\mathbf{P}_b$ are given by
\begin{equation*}
    \mathbf{P}_a=\begin{bmatrix}
(a_1,a_1')& &(a_2,a_2')\\
(a_3,a_3')&(a_4,a_4')&
\end{bmatrix},
\end{equation*}
\begin{equation*}
\mathbf{P}_b=\begin{bmatrix}
(b_1,b_1')&(b_2,b_2')& \\
 & (b_3,b_3') & (b_4,b_4') 
\end{bmatrix}.
\end{equation*}

The projection of the partitioning matrix $\mathbf{P}$ onto the first dimension is then given by (\ref{eqn: example rigid cycle}). Variable nodes (VNs) are labelled with $(j_1,j_2)$ or $(i_1,i_2)$, check nodes (CNs) are labelled with $(i_1,j_2)$ or $(j_1,i_2)$. The Kronecker product structure is evident -- three copies of $\mathbf{P}_a$ along the diagonal of the top left block, and each entry of $\mathbf{P}_b$ is replaced by a $3\times 3$ diagonal matrix. We distinguish four types of node/edge corresponding to the four blocks in (\ref{eqn: cyc opt sc hgp}):

\begin{enumerate}
    \item $(j_1,j_2)$--$(i_1,j_2)$, where $i_1$--$j_1$ is an edge in $\mathbf{A}$, corresponding to entries in the top left block $\mathbf{I}_{n_2\times n_2}\otimes \mathbf{A}(U,V)$;
    \item $(j_1,j_2)$--$(j_1,i_2)$, where $i_2$--$j_2$ is an edge in $\mathbf{B}$, corresponding to entries in the bottom left block $\mathbf{B}(U,V) \otimes \mathbf{I}_{n_1 \times n_1}$;
    \item $(i_1,i_2)$--$(i_1,j_2)$, where $i_1$--$j_1$ is an edge in $\mathbf{B}$, corresponding to entries in the top right block $\bar{\mathbf{B}}(U,V)^T \otimes \mathbf{I}_{r_1 \times r_1}$;
    \item $(i_1,i_2)$--$(j_1,i_2)$, where $i_1$--$j_1$ is an edge in $\mathbf{A}$, corresponding to entries in the bottom right block $\mathbf{I}_{r_2\times r_2}\otimes \bar{\mathbf{A}}(U,V)^T$.
\end{enumerate}

The Kronecker product structure makes it possible to label the entries of the partitioning matrix $\mathbf{P}$. For example, the type (2) blue shaded entry $b_3$ is in the second row and second column of $\mathbf{P}_b$, so $i_2=j_2=2$. It is the third diagonal entry so $j_1=3$. Similarly, the type (1)
 blue shaded entry $a_2$ is in the second copy of $\mathbf{P}_a$ so $j_2=2$. It is in the entry in row $1$ and column $3$ of $\mathbf{P}_a$ so $i_1=1$ and $j_1=3$.
\end{exam}

We now enumerate cycles-$4$, cycles-$6$, and cycles-$8$ in SC-HGP codes. According to \Cref{lemma: cycle condition}, a cycle-$2g$ candidate $x_1,x_2,\dots,x_{2g}$ becomes a cycle-$2g$ in the SC-HGP code if $\sum\nolimits_{i=1}^{2g} (-1)^ix_i=0\mod L$. 

We highlighted three cycle candidates of lengths $4$ (blue), $6$ (red), and $8$ (green) in (\ref{eqn: example rigid cycle}), respectively. Recall the alternating sum conditions for each of them giving rise to a cycle in the SC-HGP codes: the blue cycle needs $a_2-(m-b_3)+(m-a_2)-b_3=0\mod L$, the red cycle needs $a_3-a_4+b_4-(m-a_4)+(m-a_3)-b_4=0\mod L$, and the green cycle needs $a_3-a_4+b_1-b_2+a_4-a_3+b_2-b_1=0\mod L$. Notice that parameters in the alternating sums for highlighted cycle candidates cancel so the cycle conditions are always satisfied regardless of the exact assignment of $\mathbf{P}_a$ and $\mathbf{P}_b$ once $\mathbf{A}$ and $\mathbf{B}$ are fixed. As an result, the multiplicities of these cycles are only dependent on the base matrices $\mathbf{A}$ and $\mathbf{B}$: we call cycles with this property \textbf{rigid cycles}. For example, each pair of nonzero entries in $\mathbf{A}$ and $\mathbf{B}$ results in $L$ rigid cycles-$4$ in the SC-HGP code, thus the total number of rigid cycles-$4$ is $4\times 4\times L=16L$. 

Similarly, we refer to the cycles whose multiplicities are dependent on the choice of $\mathbf{P}_a$ and $\mathbf{P}_b$ as \textbf{flexible cycles}. \Cref{exam: cycle SC-HGP} shows some examples of flexible cycle-$6$ and cycles-$8$ candidates in an SC-HGP code. Then, our objective is to find $\mathbf{P}_a$ and $\mathbf{P}_b$ that minimize the total number of flexible cycles in the resultant SC-HGP code, given the structure of an underlying HGP base matrix and the SC parameters, i.e., given $\mathbf{A}$, $\mathbf{B}$, $m_1$, $m_2$, $L_1$, and $L_2$. In \Cref{exam: cycle SC-HGP}, we show how flexible cycles can either be a single cycle candidate from $\mathbf{A}$ or $\mathbf{B}$, or be decomposed into a pair of node/edge/cycle candidate from $\mathbf{A}$ and $\mathbf{B}$.  

\begin{widetext}
\vspace{2em}
\begin{equation}\label{eqn: example rigid cycle}
\hspace*{5em}
\mathbf{P}=\left[
\begin{array}{c|c}
\begin{array}{CCC|CCC|CCC}
a_1 &  & a_2 &  &  &  &  &  &  \\
\ce\tikznode{81}{$a_3$} & \ce\tikznode{82}{$a_4$} &  &  &  &  &  &  &  \\
\hline
 &  &  & a_1 &  & \cl\tikznode{11}{$a_2$} &  &  &  \\
 &  &  & \ce\tikznode{86}{$a_3$} & \ce\tikznode{85}{$a_4$} &  &  &  &  \\
\hline
 &  &  &  &  &  & a_1 &  & a_2 \\
 &  &  &  &  &  & \cs\tikznode{61}{$a_3$} & \cs\tikznode{62}{$a_4$} &  \\
\end{array} & 
\begin{array}{cc|cc}
m-b_1 &  &  &  \\
 & m-b_1 &  &  \\
\hline
m-b_2 &  & \cl\tikznode{12}{$m-b_3$} &  \\
 & m-b_2 &  & m-b_3 \\
\hline
 &  & m-b_4 &  \\
 &  &  & m-b_4 \\
\end{array}\\
\hline
\begin{array}{CCC|CCC|CCC}
\ce\tikznode{88}{$b_1$} &  &  & \ce\tikznode{87}{$b_2$} &  &  &  &  &  \\
 & \ce\tikznode{83}{$b_1$} &  &  & \ce\tikznode{84}{$b_2$} &  &  &  &  \\
 &  & b_1 &  &  & b_2 &  &  &  \\
\hline
 &  &  & b_3 &  &  & \cs\tikznode{66}{$b_4$} &  &  \\
 &  &  &  & b_3 &  &  & \cs\tikznode{63}{$b_4$} &  \\
 &  &  &  &  & \cl\tikznode{21}{$b_3$} &  &  & b_4 \\
\end{array} & 
\begin{array}{cc|cc}
m-a_1 & m-a_3 &  &  \\
 & m-a_4 &  &  \\
m-a_2 &  &  &  \\
\hline
 &  & m-a_1 & \cs\tikznode{65}{$m-a_3$} \\
 &  &  & \cs\tikznode{64}{$m-a_4$} \\
 &  & \cl\tikznode{22}{$m-a_2$} &  \\
\end{array}\\
\end{array}
\right].
\end{equation}

\begin{tikzpicture}[remember picture,overlay,blue,rounded corners]
\draw[draw=mygreen,thick] (81)--(82)--(83)--(84)--(85)--(86)--(87)--(88)--(81);
\draw[draw=red,thick] (61)--(62)--(63)--(64)--(65)--(66)--(61);
  \draw[->,shorten <=1pt] (11)
    |- +(-0.6,-0.2)
    -- +(-4.4,-0.2)
    coordinate (p1)
    node[left] {$(i_1,j_2)=(1,2)$};
  \draw[->,shorten <=1pt] (12)
    |- +(-0.6,-0.2)
    -- (p1);

  \draw[->,shorten <=1pt] (21)
    |- +(-0.6,0.2)
    -- +(-4.4,0.2)
    coordinate (p2)
    node[left] {$(j_1,i_2)=(3,2)$};
  \draw[->,shorten <=1pt] (22)
    |- +(-0.6,0.2)
    -- (p2);

  \draw[->,shorten <=1pt] (11)
    -| +(-0.3,0.6)
    -- +(-0.3,1.2)
    coordinate (p3)
    node[above] {$(j_1,j_2)=(3,2)$};
  \draw[->,shorten <=1pt] (21)
    -| +(-0.3,0.6)
    -- (p3);

  \draw[->,shorten <=1pt] (12)
    -| +(0.8,0.6)
    -- +(0.8,1.2)
    coordinate (p4)
    node[above] {$(i_1,i_2)=(1,2)$};
  \draw[->,shorten <=1pt] (22)
    -| +(0.8,0.6)
    -- (p4);
\end{tikzpicture}
\end{widetext}

\begin{figure*}[!t]
    \centering
    \subfigure[Cycles-$4$ induced flexible cycle candidates.]{
        \resizebox{0.45\textwidth}{!}{\begin{tikzpicture}[darkstyle/.style={circle,draw,fill=gray!10,very thick,minimum size=20}]
\draw[color=black, very thick] (-25*0.5,0)--(30*0.5,0);
\draw[color=black, very thick] (-25*0.5,25*0.5)--(30*0.5,25*0.5);
\draw[color=black, very thick] (-25*0.5,-30*0.5)--(30*0.5,-30*0.5);
\draw[color=black, very thick] (0,-30*0.5)--(0,25*0.5);
\draw[color=black, very thick] (-25*0.5,-30*0.5)--(-25*0.5,25*0.5);
\draw[color=black, very thick] (30*0.5,-30*0.5)--(30*0.5,25*0.5);

\foreach \x in {-20,-15,-8,-3,8,13,20,25}{
    \draw[color=gray!50, very thick] (-25*0.5,-0.5*\x)--(30*0.5,-0.5*\x);
    \draw[color=gray!50, very thick] (0.5*\x,-30*0.5)--(0.5*\x,25*0.5);
}

\draw[draw=black,very thick] (-20*0.5,-50*0.5) rectangle ++(10,7);
\draw[draw=black,very thick] (5*0.5,-50*0.5) rectangle ++(10,7);

\draw[color=gray!50, very thick] (-20*0.5,-50*0.5+1)--(-20*0.5+10,-50*0.5+1);
\draw[color=gray!50, very thick] (-20*0.5,-50*0.5+5)--(-20*0.5+10,-50*0.5+5);
\draw[color=gray!50, very thick] (-20*0.5+3,-50*0.5)--(-20*0.5+3,-50*0.5+7);
\draw[color=gray!50, very thick] (-20*0.5+7,-50*0.5)--(-20*0.5+7,-50*0.5+7);

\node[] at (-20*0.5-1,-50*0.5+5) {\Huge $i_1$};
\node[] at (-20*0.5-1,-50*0.5+1) {\Huge $i'_1$};
\node[] at (-20*0.5+3,-50*0.5+8) {\Huge $j_1$};
\node[] at (-20*0.5+7,-50*0.5+8) {\Huge $j'_1$};

\foreach \x [count=\xi] in {3,7}{
    \foreach \y [count=\yi] in {5,1}{
        \filldraw [black] (-20*0.5+\x,-50*0.5+\y) circle (5pt);
        \pgfmathtruncatemacro{\label}{2*\xi+\yi-2}
        \node[] at (-20*0.5+\x-0.5,-50*0.5+\y+0.5) {\Huge $a_{\label}$};
    }
}

\draw[color=gray!50, very thick] (5*0.5,-50*0.5+2)--(5*0.5+10,-50*0.5+2);
\draw[color=gray!50, very thick] (5*0.5,-50*0.5+6)--(5*0.5+10,-50*0.5+6);
\draw[color=gray!50, very thick] (5*0.5+2,-50*0.5)--(5*0.5+2,-50*0.5+7);
\draw[color=gray!50, very thick] (5*0.5+7,-50*0.5)--(5*0.5+7,-50*0.5+7);

\node[] at (5*0.5-1,-50*0.5+6) {\Huge $i_2$};
\node[] at (5*0.5-1,-50*0.5+1) {\Huge $i'_2$};
\node[] at (5*0.5+2,-50*0.5+8) {\Huge $j_2$};
\node[] at (5*0.5+7,-50*0.5+8) {\Huge $j'_2$};

\foreach \x [count=\xi] in {2,7}{
    \foreach \y [count=\yi] in {6,2}{
        \filldraw [black] (5*0.5+\x,-50*0.5+\y) circle (5pt);
        \pgfmathtruncatemacro{\label}{2*\xi+\yi-2}
        \node[] at (5*0.5+\x+0.5,-50*0.5+\y-0.5) {\Huge $b_{\label}$};
    }
}

    \node[] at (-28*0.5,20*0.5) {\Huge $(i_1,j_2)$};
    \node[] at (-28*0.5,15*0.5) {\Huge $(i'_1,j_2)$};
    \node[] at (-28*0.5,8*0.5) {\Huge $(i_1,j'_2)$};
    \node[] at (-28*0.5,3*0.5) {\Huge $(i'_1,j_2)$};
    \node[] at (-28*0.5,-8*0.5) {\Huge $(j_1,i_2)$};
    \node[] at (-28*0.5,-13*0.5) {\Huge $(j'_1,i_2)$};
    \node[] at (-28*0.5,-20*0.5) {\Huge $(j_1,i'_2)$};
    \node[] at (-28*0.5,-25*0.5) {\Huge $(j_1',i'_2)$};

    \node[] at (-20*0.5,27*0.5) {\Huge $(j_1,j_2)$};
    \node[] at (-15*0.5,27*0.5) {\Huge $(j'_1,j_2)$};
    \node[] at (-8*0.5,27*0.5) {\Huge $(j_1,j'_2)$};
    \node[] at (-3*0.5,27*0.5) {\Huge $(j'_1,j'_2)$};
    \node[] at (8*0.5,27*0.5) {\Huge $(i_1,i_2)$};
    \node[] at (13*0.5,27*0.5) {\Huge $(i'_1,i_2)$};
    \node[] at (20*0.5,27*0.5) {\Huge $(i_1,i'_2)$};
    \node[] at (25*0.5,27*0.5) {\Huge $(i'_1,i'_2)$};
\foreach \x [count=\xi] in {8,20}{
    \foreach \y [count=\yi] in {20,8}{
        \filldraw [blue] (0.5*\x,0.5*\y) circle (5pt);
        \pgfmathtruncatemacro{\label}{2*\yi+\xi-2}
        \node[] (ulb1\label) at (0.5*\x-0.5,0.5*\y+0.5) {\Huge $m-b_{\label}$};
        \pgfmathtruncatemacro{\label}{2*\yi+\xi-2}
        \node[] (urb1\label) at (0.5*\x+0.5,0.5*\y+0.5) {};
        \pgfmathtruncatemacro{\label}{2*\yi+\xi-2}
        \node[] (blb1\label) at (0.5*\x-0.5,0.5*\y-0.5) {};
        \pgfmathtruncatemacro{\label}{2*\yi+\xi-2}
        \node[] (brb1\label) at (0.5*\x+0.5,0.5*\y-0.5) {};
    }
}

\foreach \x [count=\xi] in {13,25}{
    \foreach \y [count=\yi] in {15,3}{
        \filldraw [blue] (0.5*\x,0.5*\y) circle (5pt);
        \pgfmathtruncatemacro{\label}{2*\yi+\xi-2}
        \node[] (ulb2\label) at (0.5*\x-0.5,0.5*\y+0.5) {\Huge $m-b_{\label}$};
        \pgfmathtruncatemacro{\label}{2*\yi+\xi-2}
        \node[] (urb2\label) at (0.5*\x+0.5,0.5*\y+0.5) {};
        \pgfmathtruncatemacro{\label}{2*\yi+\xi-2}
        \node[] (blb2\label) at (0.5*\x-0.5,0.5*\y-0.5) {};
        \pgfmathtruncatemacro{\label}{2*\yi+\xi-2}
        \node[] (brb2\label) at (0.5*\x+0.5,0.5*\y-0.5) {};
    }
}

\foreach \x [count=\xi] in {-20,-8}{
    \foreach \y [count=\yi] in {-8,-20}{
        \filldraw [blue] (0.5*\x,0.5*\y) circle (5pt);
        \pgfmathtruncatemacro{\label}{2*\xi+\yi-2}
        \node[] (ulb3\label) at (0.5*\x-0.5,0.5*\y+0.5) {\Huge $b_{\label}$};
        \pgfmathtruncatemacro{\label}{2*\xi+\yi-2}
        \node[] (urb3\label) at (0.5*\x+0.5,0.5*\y+0.5) {};
        \pgfmathtruncatemacro{\label}{2*\xi+\yi-2}
        \node[] (blb3\label) at (0.5*\x-0.5,0.5*\y-0.5) {};
        \pgfmathtruncatemacro{\label}{2*\xi+\yi-2}
        \node[] (brb3\label) at (0.5*\x+0.5,0.5*\y-0.5) {};
    }
}

\foreach \x [count=\xi] in {-15,-3}{
    \foreach \y [count=\yi] in {-13,-25}{
        \filldraw [blue] (0.5*\x,0.5*\y) circle (5pt);
        \pgfmathtruncatemacro{\label}{2*\xi+\yi-2}
        \node[] (ulb4\label) at (0.5*\x-0.5,0.5*\y+0.5) {\Huge $b_{\label}$};
        \pgfmathtruncatemacro{\label}{2*\xi+\yi-2}
        \node[] (urb4\label) at (0.5*\x+0.5,0.5*\y+0.5) {};
        \pgfmathtruncatemacro{\label}{2*\xi+\yi-2}
        \node[] (blb4\label) at (0.5*\x-0.5,0.5*\y-0.5) {};
        \pgfmathtruncatemacro{\label}{2*\xi+\yi-2}
        \node[] (brb4\label) at (0.5*\x+0.5,0.5*\y-0.5) {};
    }
}

\foreach \x [count=\xi] in {-20,-15}{
    \foreach \y [count=\yi] in {20,15}{
        \filldraw [red] (0.5*\x,0.5*\y) circle (5pt);
        \pgfmathtruncatemacro{\label}{2*\xi+\yi-2}
        \node[] (ula1\label) at (0.5*\x-0.5,0.5*\y+0.5) {\Huge $a_{\label}$};
        \pgfmathtruncatemacro{\label}{2*\xi+\yi-2}
        \node[] (ura1\label) at (0.5*\x+0.5,0.5*\y+0.5) {};
        \pgfmathtruncatemacro{\label}{2*\xi+\yi-2}
        \node[] (bla1\label) at (0.5*\x-0.5,0.5*\y-0.5) {};
        \pgfmathtruncatemacro{\label}{2*\xi+\yi-2}
        \node[] (bra1\label) at (0.5*\x+0.5,0.5*\y-0.5) {};
    }
}

\foreach \x [count=\xi] in {-8,-3}{
    \foreach \y [count=\yi] in {8,3}{
        \filldraw [red] (0.5*\x,0.5*\y) circle (5pt);
        \pgfmathtruncatemacro{\label}{2*\xi+\yi-2}
        \node[] (ula2\label) at (0.5*\x-0.5,0.5*\y+0.5) {\Huge $a_{\label}$};
        \pgfmathtruncatemacro{\label}{2*\xi+\yi-2}
        \node[] (ura2\label) at (0.5*\x+0.5,0.5*\y+0.5) {};
        \pgfmathtruncatemacro{\label}{2*\xi+\yi-2}
        \node[] (bla2\label) at (0.5*\x-0.5,0.5*\y-0.5) {};
        \pgfmathtruncatemacro{\label}{2*\xi+\yi-2}
        \node[] (bra2\label) at (0.5*\x+0.5,0.5*\y-0.5) {};
    }
}

\foreach \x [count=\xi] in {8,13}{
    \foreach \y [count=\yi] in {-8,-13}{
        \filldraw [red] (0.5*\x,0.5*\y) circle (5pt);
        \pgfmathtruncatemacro{\label}{2*\yi+\xi-2}
        \node[] at (0.5*\x-0.5+\xi-1,0.5*\y+0.5) {\Huge $m-a_{\label}$};
        \pgfmathtruncatemacro{\label}{2*\yi+\xi-2}
        \node[] (ula3\label) at (0.5*\x-0.5,0.5*\y+0.5) {};
        \pgfmathtruncatemacro{\label}{2*\yi+\xi-2}
        \node[] (ura3\label) at (0.5*\x+0.5,0.5*\y+0.5) {};
        \pgfmathtruncatemacro{\label}{2*\yi+\xi-2}
        \node[] (bla3\label) at (0.5*\x-0.5,0.5*\y-0.5) {};
        \pgfmathtruncatemacro{\label}{2*\yi+\xi-2}
        \node[] (bra3\label) at (0.5*\x+0.5,0.5*\y-0.5) {};
    }
}

\foreach \x [count=\xi] in {20,25}{
    \foreach \y [count=\yi] in {-20,-25}{
        \filldraw [red] (0.5*\x,0.5*\y) circle (5pt);
        \pgfmathtruncatemacro{\label}{2*\yi+\xi-2}
        \node[] at (0.5*\x-0.5+\xi-1,0.5*\y+0.5) {\Huge $m-a_{\label}$};
        \pgfmathtruncatemacro{\label}{2*\yi+\xi-2}
        \node[] (ula4\label) at (0.5*\x-0.5,0.5*\y+0.5) {};
        \pgfmathtruncatemacro{\label}{2*\yi+\xi-2}
        \node[] (ura4\label) at (0.5*\x+0.5,0.5*\y+0.5) {};
        \pgfmathtruncatemacro{\label}{2*\yi+\xi-2}
        \node[] (bla4\label) at (0.5*\x-0.5,0.5*\y-0.5) {};
        \pgfmathtruncatemacro{\label}{2*\yi+\xi-2}
        \node[] (bra4\label) at (0.5*\x+0.5,0.5*\y-0.5) {};
    }
}

\draw[color=red,very thick] (ula11)--(ulb11)--(ulb13)--(ulb14)--(ula43)--(ula44)--(ulb22)--(ula12)--(ula11);
\draw[color=blue,very thick] (bra11)--(brb12)--(bra43)--(bra44)--(brb24)--(brb23)--(brb21)--(bra12)--(bra11);

\draw[color=orange,very thick] (ula21)--(ula23)--(ulb43)--(ula34)--(ula32)--(ulb33)--(ula21);
\draw[color=mygreen,very thick] (bra22)--(bra24)--(bra23)--(brb13)--(bra31)--(brb33)--(bra22);

\end{tikzpicture}}}
    \hspace{10pt} 
    \subfigure[Cycles-$6$ induced flexible cycle candidates.]{
        \resizebox{0.45\textwidth}{!}{\begin{tikzpicture}[darkstyle/.style={circle,draw,fill=gray!10,very thick,minimum size=20}]
\draw[color=black, very thick] (-25*0.5,0)--(30*0.5,0);
\draw[color=black, very thick] (-25*0.5,25*0.5)--(30*0.5,25*0.5);
\draw[color=black, very thick] (-25*0.5,-30*0.5)--(30*0.5,-30*0.5);
\draw[color=black, very thick] (0,-30*0.5)--(0,25*0.5);
\draw[color=black, very thick] (-25*0.5,-30*0.5)--(-25*0.5,25*0.5);
\draw[color=black, very thick] (30*0.5,-30*0.5)--(30*0.5,25*0.5);

\foreach \x in {-20,-15,-8,8,13,20}{
    \draw[color=gray!50, very thick] (-25*0.5,-0.5*\x)--(30*0.5,-0.5*\x);
    \draw[color=gray!50, very thick] (0.5*\x,-30*0.5)--(0.5*\x,25*0.5);
}

\draw[draw=black,very thick] (-20*0.5,-50*0.5) rectangle ++(10,7);
\draw[draw=black,very thick] (5*0.5,-50*0.5) rectangle ++(10,7);

\draw[color=gray!50, very thick] (-20*0.5,-50*0.5+1)--(-20*0.5+10,-50*0.5+1);
\draw[color=gray!50, very thick] (-20*0.5,-50*0.5+4)--(-20*0.5+10,-50*0.5+4);
\draw[color=gray!50, very thick] (-20*0.5,-50*0.5+6)--(-20*0.5+10,-50*0.5+6);
\draw[color=gray!50, very thick] (-20*0.5+2,-50*0.5)--(-20*0.5+2,-50*0.5+7);
\draw[color=gray!50, very thick] (-20*0.5+6,-50*0.5)--(-20*0.5+6,-50*0.5+7);
\draw[color=gray!50, very thick] (-20*0.5+9,-50*0.5)--(-20*0.5+9,-50*0.5+7);

\node[] at (-20*0.5-1,-50*0.5+6) {\Huge $i_1$};
\node[] at (-20*0.5-1,-50*0.5+4) {\Huge $i'_1$};
\node[] at (-20*0.5-1,-50*0.5+1) {\Huge $i''_1$};
\node[] at (-20*0.5+2,-50*0.5+8) {\Huge $j_1$};
\node[] at (-20*0.5+6,-50*0.5+8) {\Huge $j'_1$};
\node[] at (-20*0.5+9,-50*0.5+8) {\Huge $j''_1$};

\filldraw [black] (-20*0.5+2,-50*0.5+6) circle (5pt);
\node[] at (-20*0.5+2-0.5,-50*0.5+6+0.5) {\Huge $a_{1}$};
\filldraw [black] (-20*0.5+6,-50*0.5+6) circle (5pt);
\node[] at (-20*0.5+6-0.5,-50*0.5+6+0.5) {\Huge $a_{2}$};
\filldraw [black] (-20*0.5+6,-50*0.5+4) circle (5pt);
\node[] at (-20*0.5+6-0.5,-50*0.5+4+0.5) {\Huge $a_{3}$};
\filldraw [black] (-20*0.5+9,-50*0.5+4) circle (5pt);
\node[] at (-20*0.5+9-0.5,-50*0.5+4+0.5) {\Huge $a_{4}$};
\filldraw [black] (-20*0.5+9,-50*0.5+1) circle (5pt);
\node[] at (-20*0.5+9-0.5,-50*0.5+1+0.5) {\Huge $a_{5}$};
\filldraw [black] (-20*0.5+2,-50*0.5+1) circle (5pt);
\node[] at (-20*0.5+2-0.5,-50*0.5+1+0.5) {\Huge $a_{6}$};

\draw[color=gray!50, very thick] (5*0.5,-50*0.5+3)--(5*0.5+10,-50*0.5+3);
\draw[color=gray!50, very thick] (5*0.5+6,-50*0.5)--(5*0.5+6,-50*0.5+7);

\node[] at (5*0.5-1,-50*0.5+3) {\Huge $i_2$};
\node[] at (5*0.5+6,-50*0.5+8) {\Huge $j_2$};

\filldraw [black] (5*0.5+6,-50*0.5+3) circle (5pt);
\node[] at (5*0.5+6-0.5,-50*0.5+3+0.5) {\Huge $b$};

    \node[] at (-28*0.5,20*0.5) {\Huge $(i_1,j_2)$};
    \node[] at (-28*0.5,15*0.5) {\Huge $(i'_1,j_2)$};
    \node[] at (-28*0.5,8*0.5) {\Huge $(i''_1,j_2)$};

    \node[] at (-28*0.5,-8*0.5) {\Huge $(j_1,i_2)$};
    \node[] at (-28*0.5,-13*0.5) {\Huge $(j'_1,i_2)$};
    \node[] at (-28*0.5,-20*0.5) {\Huge $(j''_1,i_2)$};

    \node[] at (-20*0.5,27*0.5) {\Huge $(j_1,j_2)$};
    \node[] at (-15*0.5,27*0.5) {\Huge $(j'_1,j_2)$};
    \node[] at (-8*0.5,27*0.5) {\Huge $(j''_1,j_2)$};
    
    \node[] at (8*0.5,27*0.5) {\Huge $(i_1,i_2)$};
    \node[] at (13*0.5,27*0.5) {\Huge $(i'_1,i_2)$};
    \node[] at (20*0.5,27*0.5) {\Huge $(i''_1,i_2)$};

\foreach \x in {8}
    \foreach \y in {20}{
        \filldraw [blue] (0.5*\x,0.5*\y) circle (5pt);
        \node[] (ulb11) at (0.5*\x-0.5,0.5*\y+0.5) {\Huge $m-b$};
        \node[] (brb11) at (0.5*\x+0.5,0.5*\y-0.5) {$ $};
        \node[] (b11) at (0.5*\x,0.5*\y) {$ $};
        \filldraw [blue] (-0.5*\y,-0.5*\x) circle (5pt);
        \node[] (ulb21) at (-0.5*\y-0.5,-0.5*\x+0.5) {\Huge $b$};
        \node[] (brb21) at (-0.5*\y+0.5,-0.5*\x-0.5) {$ $};
        \node[] (b21) at (-0.5*\y,-0.5*\x) {$ $};
    }

\foreach \x in {13}
    \foreach \y in {15}{
        \filldraw [blue] (0.5*\x,0.5*\y) circle (5pt);
        \node[] (ulb12) at (0.5*\x-0.5,0.5*\y+0.5) {\Huge $m-b$};
        \node[] (brb12) at (0.5*\x+0.5,0.5*\y-0.5) {$ $};
        \node[] (b12) at (0.5*\x,0.5*\y) {$ $};
        \filldraw [blue] (-0.5*\y,-0.5*\x) circle (5pt);
        \node[] (ulb22) at (-0.5*\y-0.5,-0.5*\x+0.5) {\Huge $b$};
        \node[] (brb22) at (-0.5*\y+0.5,-0.5*\x-0.5) {$ $};
        \node[] (b22) at (-0.5*\y,-0.5*\x) {$ $};
}

\foreach \x in {20}
    \foreach \y in {8}{
        \filldraw [blue] (0.5*\x,0.5*\y) circle (5pt);
        \node[] (ulb13) at (0.5*\x-0.5,0.5*\y+0.5) {\Huge $m-b$};
        \node[] (brb13) at (0.5*\x+0.5,0.5*\y-0.5) {$ $};
        \node[] (b13) at (0.5*\x,0.5*\y) {$ $};
        \filldraw [blue] (-0.5*\y,-0.5*\x) circle (5pt);
        \node[] (ulb23) at (-0.5*\y-0.5,-0.5*\x+0.5) {\Huge $b$};
        \node[] (brb23) at (-0.5*\y+0.5,-0.5*\x-0.5) {$ $};
        \node[] (b23) at (-0.5*\y,-0.5*\x) {$ $};
}

\foreach \x in {-20}
    \foreach \y in {20}{
        \filldraw [red] (0.5*\x,0.5*\y) circle (5pt);
        \node[] (ula11) at (0.5*\x-0.5,0.5*\y+0.5) {\Huge $a_1$};
        \node[] (bra11) at (0.5*\x+0.5,0.5*\y-0.5) {$ $};
        \node[] (a11) at (0.5*\x,0.5*\y) {$ $};
        \filldraw [red] (14+0.5*\x,-14+0.5*\y) circle (5pt);
        \node[] (ula21) at (14+0.5*\x-0.5,-14+0.5*\y+0.5) {\Huge $m-a_1$};
        \node[] (bra21) at (14+0.5*\x+0.5,-14+0.5*\y-0.5) {$ $};
        \node[] (a21) at (14+0.5*\x,-14+0.5*\y) {$ $};
}

\foreach \x in {-15}
    \foreach \y in {20}{
        \filldraw [red] (0.5*\x,0.5*\y) circle (5pt);
        \node[] (ula12) at (0.5*\x-0.5,0.5*\y+0.5) {\Huge $a_2$};
        \node[] (bra12) at (0.5*\x+0.5,0.5*\y-0.5) {$ $};
        \node[] (a12) at (0.5*\x,0.5*\y) {$ $};
        \filldraw [red] (11.5+0.5*\x,-16.5+0.5*\y) circle (5pt);
        \node[] (ula22) at (11.5+0.5*\x-0.5,-16.5+0.5*\y+0.5) {\Huge $m-a_2$};
        \node[] (bra22) at (11.5+0.5*\x+0.5,-16.5+0.5*\y-0.5) {$ $};
        \node[] (a22) at (11.5+0.5*\x,-16.5+0.5*\y) {$ $};
        }

\foreach \x in {-15}
    \foreach \y in {15}{
        \filldraw [red] (0.5*\x,0.5*\y) circle (5pt);
        \node[] (ula13) at (0.5*\x-0.5,0.5*\y+0.5) {\Huge $a_3$};
        \node[] (bra13) at (0.5*\x+0.5,0.5*\y-0.5) {$ $};
        \node[] (a13) at (0.5*\x,0.5*\y) {$ $};
        \filldraw [red] (14+0.5*\x,-14+0.5*\y) circle (5pt);
        \node[] at (14+0.5*\x-0.5+1,-14+0.5*\y+0.5) {\Huge $m-a_3$};
        \node[] (ula23) at (14+0.5*\x-0.5,-14+0.5*\y+0.5) {$ $};
        \node[] (bra23) at (14+0.5*\x+0.5,-14+0.5*\y-0.5) {$ $};
        \node[] (a23) at (14+0.5*\x,-14+0.5*\y) {$ $};
        }

\foreach \x in {-8}
    \foreach \y in {15}{
        \filldraw [red] (0.5*\x,0.5*\y) circle (5pt);
        \node[] (ula14) at (0.5*\x-0.5,0.5*\y+0.5) {\Huge $a_4$};
        \node[] (bra14) at (0.5*\x+0.5,0.5*\y-0.5) {$ $};
        \node[] (a14) at (0.5*\x,0.5*\y) {$ $};
        \filldraw [red] (10.5+0.5*\x,-17.5+0.5*\y) circle (5pt);
        \node[] (ula24) at (10.5+0.5*\x-0.5,-17.5+0.5*\y+0.5) {\Huge $m-a_4$};
        \node[] (bra24) at (10.5+0.5*\x+0.5,-17.5+0.5*\y-0.5) {$ $};
        \node[] (a24) at (10.5+0.5*\x,-17.5+0.5*\y) {$ $};
        }

\foreach \x in {-8}
    \foreach \y in {8}{
        \filldraw [red] (0.5*\x,0.5*\y) circle (5pt);
        \node[] (ula15) at (0.5*\x-0.5,0.5*\y+0.5) {\Huge $a_5$};
        \node[] (bra15) at (0.5*\x+0.5,0.5*\y-0.5) {$ $};
        \node[] (a15) at (0.5*\x,0.5*\y) {$ $};
        \filldraw [red] (14+0.5*\x,-14+0.5*\y) circle (5pt);
        \node[] (ula25) at (14+0.5*\x-0.5,-14+0.5*\y+0.5) {\Huge $m-a_5$};
        \node[] (bra25) at (14+0.5*\x+0.5,-14+0.5*\y-0.5) {$ $};
        \node[] (a25) at (14+0.5*\x,-14+0.5*\y) {$ $};
}

\foreach \x in {-20}
    \foreach \y in {8}{
        \filldraw [red] (0.5*\x,0.5*\y) circle (5pt);
        \node[] (ula16) at (0.5*\x-0.5,0.5*\y+0.5) {\Huge $a_6$};
        \node[] (bra16) at (0.5*\x+0.5,0.5*\y-0.5) {$ $};
        \node[] (a16) at (0.5*\x,0.5*\y) {$ $};
        \filldraw [red] (20+0.5*\x,-8+0.5*\y) circle (5pt);
        \node[] (ula26) at (20+0.5*\x-0.5,-8+0.5*\y+0.5) {\Huge $m-a_6$};
        \node[] (bra26) at (20+0.5*\x+0.5,-8+0.5*\y-0.5) {$ $};
        \node[] (a26) at (20+0.5*\x,-8+0.5*\y) {$ $};
}

\draw[color=red,very thick] (ula11)--(ulb11)--(ula22)--(ula23)--(ulb12)--(ula14)--(ula15)--(ula16)--(ula11);
\draw[color=mygreen,very thick] (bra11)--(bra12)--(bra13)--(brb12)--(bra24)--(brb23)--(bra15)--(bra16)--(bra11);
\draw[color=blue,very thick] (a11)--(a12)--(b22)--(a23)--(a24)--(a25)--(b13)--(a16)--(a11);
\end{tikzpicture}}}
    \caption{
        Examples show how flexible cycle candidates of SC-HGP codes can be decomposed into two different components from $\mathbf{P}_a$ and $\mathbf{P}_b$. The matrices in the bottom panels of (a) and (b) are $\mathbf{P}_a$ and $\mathbf{P}_b$, and those in the top panels are $\mathbf{P}$. Nodes from copies of $\mathbf{A}$ and $\mathbf{B}$ are highlighted by red nodes and blue nodes, respectively. The number at each node is the value in its partitioning matrix $\mathbf{P}_a$ or $\mathbf{P}_b$. The red and blue cycles in (a) shows two flexible cycles-$8$ candidates where each of them is resulting from a pair of cycles-$4$ candidates from $\mathbf{P}_a$ and $\mathbf{P}_b$. Let $s_a=a_1-a_2+a_4-a_3$, $s_b=b_1-b_2+b_4-b_3$, the alternating sums associated with the blue and red cycle candidates are $s_a-s_b$ and $s_a+s_b$, respectively. The orange and green cycles in (a) shows two flexible cycles-$6$ candidates where each of them is resulting from a cycle-$4$ candidate from $\mathbf{P}_a$ and node $b_3$ from $\mathbf{P}_b$. The alternating sums of both of them are $s_a$. (b) shows three flexible cycles-$8$ candidates where each of them is resulting from a cycle-$6$ candidate from $\mathbf{P}_a$ and an entry $b$ from $\mathbf{P}_b$. The alternating sums of all three cases are $a_1-a_2+a_3-a_4+a_5-a_6$.
    }
    \label{fig: cycle 8}
\end{figure*}

\begin{exam}\label{exam: cycle SC-HGP} Fig.~\ref{fig: cycle 8} shows how flexible cycle (candidates) in SC-HGP codes can be decomposed into components in $\mathbf{P}_a$ and $\mathbf{P}_b$.

We calculate alternating sums of the orange and green cycle-$6$ candidates in Fig.~\ref{fig: cycle 8}(a) to obtain 
\begin{equation*}
\begin{split}
&a_1-b_3+(m-a_2)-(m-a_4)+b_3-a_3\\
=&a_1+a_4-a_2-a_3,\textrm{ and }\\
&a_4-a_3+(m-b_3)-(m-a_1)+b_3-a_2\\
=&a_1+a_4-a_2-a_3.
\end{split}
\end{equation*} 
A flexible cycle-$6$ arise from these cycle candidates can be decomposed into a cycle-$4$ candidate in $\mathbf{P}_a$ that satisfies the alternating sum condition $a_1+a_4-a_2-a_3=0\mod L$ along with an entry $b_3$ in $\mathbf{P}_b$.

Let $s_a=a_1-a_2+a_4-a_3$ and $s_b=b_1-b_2+b_4-b_3$. We calculate alternating sums of the red and blue cycle-$8$ candidates in Fig.~\ref{fig: cycle 8}(a) to obtain 
\begin{equation*}
\begin{split}
&a_1-(m-b_1)+(m-b_3)-(m-b_4)+(m-a_3)\\
&-(m-a_4)+(m-b_2)-a_2=s_a+s_b,\textrm{ and } \\
&a_1-(m-b_2)+(m-a_3)-(m-a_4)+(m-b_4)\\
&-(m-b_3))+(m-b_1)-a_2=s_a-s_b.
\end{split}
\end{equation*} 

The condition of at least one of the cycle candidates determines a cycle-$8$ in the SC-HGP code is 
\begin{equation}
    |a_1+a_4-a_2-a_3|=|b_1+b_4-b_2-b_3|=S,
\end{equation}
and if $S=0$, both paths determine cycles-$8$. The associated cycles-$8$ can be decomposed into a pair of cycles-$4$ candidate in $\mathbf{P}_a$ and $\mathbf{P}_b$ with identical or opposite alternating sums.

We calculate the alternating sum of the three cycle-$8$ candidates in Fig.~\ref{fig: cycle 8}(b) to obtain
\begin{equation*}
\begin{split}
   &a_1-a_2+a_3-(m-b)+(m-a_4)-b\\
   &+a_5-a_6=(a_1-a_2+a_3-a_4+a_5-a_6).
\end{split}
\end{equation*}
A flexible cycles-$8$ arise from these cycle candidates can be decomposed into a cycle-$6$ candidate in $\mathbf{P}_a$ that satisfies the alternating sum condition $a_1-a_2+a_3-a_4+a_5-a_6=0\mod L$ along with an entry $b$ in $\mathbf{P}_b$.
\end{exam}

\begin{figure*}[!t]
\centering

\subfigure[$\{\{2,2\},\{0,0\}\}$.]{\resizebox{0.32\textwidth}{!}{\begin{tikzpicture}[darkstyle/.style={circle,draw,fill=gray!10,very thick,minimum size=20}]
\draw[color=black, very thick] (-25*0.5,0)--(30*0.5,0);
\draw[color=black, very thick] (-25*0.5,25*0.5)--(30*0.5,25*0.5);
\draw[color=black, very thick] (-25*0.5,-30*0.5)--(30*0.5,-30*0.5);
\draw[color=black, very thick] (0,-30*0.5)--(0,25*0.5);
\draw[color=black, very thick] (-25*0.5,-30*0.5)--(-25*0.5,25*0.5);
\draw[color=black, very thick] (30*0.5,-30*0.5)--(30*0.5,25*0.5);

\foreach \x in {-20,-15,-8,-3,8,13,20,25}{
    \draw[color=gray!50, very thick] (-25*0.5,-0.5*\x)--(30*0.5,-0.5*\x);
    \draw[color=gray!50, very thick] (0.5*\x,-30*0.5)--(0.5*\x,25*0.5);
}

\foreach \x [count=\xi] in {8,20}{
    \foreach \y [count=\yi] in {20,8}{
        \filldraw [blue] (0.5*\x,0.5*\y) circle (5pt);
        \pgfmathtruncatemacro{\label}{2*\yi+\xi-2}
        \node[] (ulb1\label) at (0.5*\x-0.5,0.5*\y+0.5) {\Huge $m-b_{\label}$};
        \pgfmathtruncatemacro{\label}{2*\yi+\xi-2}
        \node[] (urb1\label) at (0.5*\x+0.5,0.5*\y+0.5) {};
        \pgfmathtruncatemacro{\label}{2*\yi+\xi-2}
        \node[] (blb1\label) at (0.5*\x-0.5,0.5*\y-0.5) {};
        \pgfmathtruncatemacro{\label}{2*\yi+\xi-2}
        \node[] (brb1\label) at (0.5*\x+0.5,0.5*\y-0.5) {};
    }
}

\foreach \x [count=\xi] in {13,25}{
    \foreach \y [count=\yi] in {15,3}{
        \filldraw [blue] (0.5*\x,0.5*\y) circle (5pt);
        \pgfmathtruncatemacro{\label}{2*\yi+\xi-2}
        \node[] (ulb2\label) at (0.5*\x-0.5,0.5*\y+0.5) {\Huge $m-b_{\label}$};
        \pgfmathtruncatemacro{\label}{2*\yi+\xi-2}
        \node[] (urb2\label) at (0.5*\x+0.5,0.5*\y+0.5) {};
        \pgfmathtruncatemacro{\label}{2*\yi+\xi-2}
        \node[] (blb2\label) at (0.5*\x-0.5,0.5*\y-0.5) {};
        \pgfmathtruncatemacro{\label}{2*\yi+\xi-2}
        \node[] (brb2\label) at (0.5*\x+0.5,0.5*\y-0.5) {};
    }
}

\foreach \x [count=\xi] in {-20,-8}{
    \foreach \y [count=\yi] in {-8,-20}{
        \filldraw [blue] (0.5*\x,0.5*\y) circle (5pt);
        \pgfmathtruncatemacro{\label}{2*\xi+\yi-2}
        \node[] (ulb3\label) at (0.5*\x-0.5,0.5*\y+0.5) {\Huge $b_{\label}$};
        \pgfmathtruncatemacro{\label}{2*\xi+\yi-2}
        \node[] (urb3\label) at (0.5*\x+0.5,0.5*\y+0.5) {};
        \pgfmathtruncatemacro{\label}{2*\xi+\yi-2}
        \node[] (blb3\label) at (0.5*\x-0.5,0.5*\y-0.5) {};
        \pgfmathtruncatemacro{\label}{2*\xi+\yi-2}
        \node[] (brb3\label) at (0.5*\x+0.5,0.5*\y-0.5) {};
    }
}

\foreach \x [count=\xi] in {-15,-3}{
    \foreach \y [count=\yi] in {-13,-25}{
        \filldraw [blue] (0.5*\x,0.5*\y) circle (5pt);
        \pgfmathtruncatemacro{\label}{2*\xi+\yi-2}
        \node[] (ulb4\label) at (0.5*\x-0.5,0.5*\y+0.5) {\Huge $b_{\label}$};
        \pgfmathtruncatemacro{\label}{2*\xi+\yi-2}
        \node[] (urb4\label) at (0.5*\x+0.5,0.5*\y+0.5) {};
        \pgfmathtruncatemacro{\label}{2*\xi+\yi-2}
        \node[] (blb4\label) at (0.5*\x-0.5,0.5*\y-0.5) {};
        \pgfmathtruncatemacro{\label}{2*\xi+\yi-2}
        \node[] (brb4\label) at (0.5*\x+0.5,0.5*\y-0.5) {};
    }
}

\foreach \x [count=\xi] in {-20,-15}{
    \foreach \y [count=\yi] in {20,15}{
        \filldraw [red] (0.5*\x,0.5*\y) circle (5pt);
        \pgfmathtruncatemacro{\label}{2*\xi+\yi-2}
        \node[] (ula1\label) at (0.5*\x-0.5,0.5*\y+0.5) {\Huge $a_{\label}$};
        \pgfmathtruncatemacro{\label}{2*\xi+\yi-2}
        \node[] (ura1\label) at (0.5*\x+0.5,0.5*\y+0.5) {};
        \pgfmathtruncatemacro{\label}{2*\xi+\yi-2}
        \node[] (bla1\label) at (0.5*\x-0.5,0.5*\y-0.5) {};
        \pgfmathtruncatemacro{\label}{2*\xi+\yi-2}
        \node[] (bra1\label) at (0.5*\x+0.5,0.5*\y-0.5) {};
    }
}

\foreach \x [count=\xi] in {-8,-3}{
    \foreach \y [count=\yi] in {8,3}{
        \filldraw [red] (0.5*\x,0.5*\y) circle (5pt);
        \pgfmathtruncatemacro{\label}{2*\xi+\yi-2}
        \node[] (ula2\label) at (0.5*\x-0.5,0.5*\y+0.5) {\Huge $a_{\label}$};
        \pgfmathtruncatemacro{\label}{2*\xi+\yi-2}
        \node[] (ura2\label) at (0.5*\x+0.5,0.5*\y+0.5) {};
        \pgfmathtruncatemacro{\label}{2*\xi+\yi-2}
        \node[] (bla2\label) at (0.5*\x-0.5,0.5*\y-0.5) {};
        \pgfmathtruncatemacro{\label}{2*\xi+\yi-2}
        \node[] (bra2\label) at (0.5*\x+0.5,0.5*\y-0.5) {};
    }
}

\foreach \x [count=\xi] in {8,13}{
    \foreach \y [count=\yi] in {-8,-13}{
        \filldraw [red] (0.5*\x,0.5*\y) circle (5pt);
        \pgfmathtruncatemacro{\label}{2*\yi+\xi-2}
        \node[] at (0.5*\x-0.5+\xi-1,0.5*\y+0.5) {\Huge $m-a_{\label}$};
        \pgfmathtruncatemacro{\label}{2*\yi+\xi-2}
        \node[] (ula3\label) at (0.5*\x-0.5,0.5*\y+0.5) {};
        \pgfmathtruncatemacro{\label}{2*\yi+\xi-2}
        \node[] (ura3\label) at (0.5*\x+0.5,0.5*\y+0.5) {};
        \pgfmathtruncatemacro{\label}{2*\yi+\xi-2}
        \node[] (bla3\label) at (0.5*\x-0.5,0.5*\y-0.5) {};
        \pgfmathtruncatemacro{\label}{2*\yi+\xi-2}
        \node[] (bra3\label) at (0.5*\x+0.5,0.5*\y-0.5) {};
    }
}

\foreach \x [count=\xi] in {20,25}{
    \foreach \y [count=\yi] in {-20,-25}{
        \filldraw [red] (0.5*\x,0.5*\y) circle (5pt);
        \pgfmathtruncatemacro{\label}{2*\yi+\xi-2}
        \node[] at (0.5*\x-0.5+\xi-1,0.5*\y+0.5) {\Huge $m-a_{\label}$};
        \pgfmathtruncatemacro{\label}{2*\yi+\xi-2}
        \node[] (ula4\label) at (0.5*\x-0.5,0.5*\y+0.5) {};
        \pgfmathtruncatemacro{\label}{2*\yi+\xi-2}
        \node[] (ura4\label) at (0.5*\x+0.5,0.5*\y+0.5) {};
        \pgfmathtruncatemacro{\label}{2*\yi+\xi-2}
        \node[] (bla4\label) at (0.5*\x-0.5,0.5*\y-0.5) {};
        \pgfmathtruncatemacro{\label}{2*\yi+\xi-2}
        \node[] (bra4\label) at (0.5*\x+0.5,0.5*\y-0.5) {};
    }
}

\draw[color=red,very thick] (ula13)--(ulb12)--(ulb14)--(ula21)--(ula22)--(ulb23)--(ulb21)--(ula14)--(ula13);
\draw[color=blue,very thick] (bra13)--(brb11)--(brb13)--(bra21)--(bra22)--(brb24)--(brb22)--(bra14)--(bra13);

\draw[color=mygreen,very thick] (ula31)--(ula33)--(ulb41)--(ulb42)--(ula44)--(ula42)--(ulb32)--(ulb31)--(ula31);

\end{tikzpicture}}}\hspace{5pt}
\subfigure[$\{\{2,0\},\{1,1\}\}$.]{\resizebox{0.32\textwidth}{!}{\begin{tikzpicture}[darkstyle/.style={circle,draw,fill=gray!10,very thick,minimum size=20}]
\draw[color=black, very thick] (-25*0.5,0)--(30*0.5,0);
\draw[color=black, very thick] (-25*0.5,25*0.5)--(30*0.5,25*0.5);
\draw[color=black, very thick] (-25*0.5,-30*0.5)--(30*0.5,-30*0.5);
\draw[color=black, very thick] (0,-30*0.5)--(0,25*0.5);
\draw[color=black, very thick] (-25*0.5,-30*0.5)--(-25*0.5,25*0.5);
\draw[color=black, very thick] (30*0.5,-30*0.5)--(30*0.5,25*0.5);

\foreach \x in {-20,-15,-8,-3,8,13,20,25}{
    \draw[color=gray!50, very thick] (-25*0.5,-0.5*\x)--(30*0.5,-0.5*\x);
    \draw[color=gray!50, very thick] (0.5*\x,-30*0.5)--(0.5*\x,25*0.5);
}

\foreach \x [count=\xi] in {8,20}{
    \foreach \y [count=\yi] in {20,8}{
        \filldraw [blue] (0.5*\x,0.5*\y) circle (5pt);
        \pgfmathtruncatemacro{\label}{2*\yi+\xi-2}
        \node[] (ulb1\label) at (0.5*\x-0.5,0.5*\y+0.5) {\Huge $m-b_{\label}$};
        \pgfmathtruncatemacro{\label}{2*\yi+\xi-2}
        \node[] (urb1\label) at (0.5*\x+0.5,0.5*\y+0.5) {};
        \pgfmathtruncatemacro{\label}{2*\yi+\xi-2}
        \node[] (blb1\label) at (0.5*\x-0.5,0.5*\y-0.5) {};
        \pgfmathtruncatemacro{\label}{2*\yi+\xi-2}
        \node[] (brb1\label) at (0.5*\x+0.5,0.5*\y-0.5) {};
    }
}

\foreach \x [count=\xi] in {13,25}{
    \foreach \y [count=\yi] in {15,3}{
        \filldraw [blue] (0.5*\x,0.5*\y) circle (5pt);
        \pgfmathtruncatemacro{\label}{2*\yi+\xi-2}
        \node[] (ulb2\label) at (0.5*\x-0.5,0.5*\y+0.5) {\Huge $m-b_{\label}$};
        \pgfmathtruncatemacro{\label}{2*\yi+\xi-2}
        \node[] (urb2\label) at (0.5*\x+0.5,0.5*\y+0.5) {};
        \pgfmathtruncatemacro{\label}{2*\yi+\xi-2}
        \node[] (blb2\label) at (0.5*\x-0.5,0.5*\y-0.5) {};
        \pgfmathtruncatemacro{\label}{2*\yi+\xi-2}
        \node[] (brb2\label) at (0.5*\x+0.5,0.5*\y-0.5) {};
    }
}

\foreach \x [count=\xi] in {-20,-8}{
    \foreach \y [count=\yi] in {-8,-20}{
        \filldraw [blue] (0.5*\x,0.5*\y) circle (5pt);
        \pgfmathtruncatemacro{\label}{2*\xi+\yi-2}
        \node[] (ulb3\label) at (0.5*\x-0.5,0.5*\y+0.5) {\Huge $b_{\label}$};
        \pgfmathtruncatemacro{\label}{2*\xi+\yi-2}
        \node[] (urb3\label) at (0.5*\x+0.5,0.5*\y+0.5) {};
        \pgfmathtruncatemacro{\label}{2*\xi+\yi-2}
        \node[] (blb3\label) at (0.5*\x-0.5,0.5*\y-0.5) {};
        \pgfmathtruncatemacro{\label}{2*\xi+\yi-2}
        \node[] (brb3\label) at (0.5*\x+0.5,0.5*\y-0.5) {};
    }
}

\foreach \x [count=\xi] in {-15,-3}{
    \foreach \y [count=\yi] in {-13,-25}{
        \filldraw [blue] (0.5*\x,0.5*\y) circle (5pt);
        \pgfmathtruncatemacro{\label}{2*\xi+\yi-2}
        \node[] (ulb4\label) at (0.5*\x-0.5,0.5*\y+0.5) {\Huge $b_{\label}$};
        \pgfmathtruncatemacro{\label}{2*\xi+\yi-2}
        \node[] (urb4\label) at (0.5*\x+0.5,0.5*\y+0.5) {};
        \pgfmathtruncatemacro{\label}{2*\xi+\yi-2}
        \node[] (blb4\label) at (0.5*\x-0.5,0.5*\y-0.5) {};
        \pgfmathtruncatemacro{\label}{2*\xi+\yi-2}
        \node[] (brb4\label) at (0.5*\x+0.5,0.5*\y-0.5) {};
    }
}

\foreach \x [count=\xi] in {-20,-15}{
    \foreach \y [count=\yi] in {20,15}{
        \filldraw [red] (0.5*\x,0.5*\y) circle (5pt);
        \pgfmathtruncatemacro{\label}{2*\xi+\yi-2}
        \node[] (ula1\label) at (0.5*\x-0.5,0.5*\y+0.5) {\Huge $a_{\label}$};
        \pgfmathtruncatemacro{\label}{2*\xi+\yi-2}
        \node[] (ura1\label) at (0.5*\x+0.5,0.5*\y+0.5) {};
        \pgfmathtruncatemacro{\label}{2*\xi+\yi-2}
        \node[] (bla1\label) at (0.5*\x-0.5,0.5*\y-0.5) {};
        \pgfmathtruncatemacro{\label}{2*\xi+\yi-2}
        \node[] (bra1\label) at (0.5*\x+0.5,0.5*\y-0.5) {};
    }
}

\foreach \x [count=\xi] in {-8,-3}{
    \foreach \y [count=\yi] in {8,3}{
        \filldraw [red] (0.5*\x,0.5*\y) circle (5pt);
        \pgfmathtruncatemacro{\label}{2*\xi+\yi-2}
        \node[] (ula2\label) at (0.5*\x-0.5,0.5*\y+0.5) {\Huge $a_{\label}$};
        \pgfmathtruncatemacro{\label}{2*\xi+\yi-2}
        \node[] (ura2\label) at (0.5*\x+0.5,0.5*\y+0.5) {};
        \pgfmathtruncatemacro{\label}{2*\xi+\yi-2}
        \node[] (bla2\label) at (0.5*\x-0.5,0.5*\y-0.5) {};
        \pgfmathtruncatemacro{\label}{2*\xi+\yi-2}
        \node[] (bra2\label) at (0.5*\x+0.5,0.5*\y-0.5) {};
    }
}

\foreach \x [count=\xi] in {8,13}{
    \foreach \y [count=\yi] in {-8,-13}{
        \filldraw [red] (0.5*\x,0.5*\y) circle (5pt);
        \pgfmathtruncatemacro{\label}{2*\yi+\xi-2}
        \node[] at (0.5*\x-0.5+\xi-1,0.5*\y+0.5) {\Huge $m-a_{\label}$};
        \pgfmathtruncatemacro{\label}{2*\yi+\xi-2}
        \node[] (ula3\label) at (0.5*\x-0.5,0.5*\y+0.5) {};
        \pgfmathtruncatemacro{\label}{2*\yi+\xi-2}
        \node[] (ura3\label) at (0.5*\x+0.5,0.5*\y+0.5) {};
        \pgfmathtruncatemacro{\label}{2*\yi+\xi-2}
        \node[] (bla3\label) at (0.5*\x-0.5,0.5*\y-0.5) {};
        \pgfmathtruncatemacro{\label}{2*\yi+\xi-2}
        \node[] (bra3\label) at (0.5*\x+0.5,0.5*\y-0.5) {};
    }
}

\foreach \x [count=\xi] in {20,25}{
    \foreach \y [count=\yi] in {-20,-25}{
        \filldraw [red] (0.5*\x,0.5*\y) circle (5pt);
        \pgfmathtruncatemacro{\label}{2*\yi+\xi-2}
        \node[] at (0.5*\x-0.5+\xi-1,0.5*\y+0.5) {\Huge $m-a_{\label}$};
        \pgfmathtruncatemacro{\label}{2*\yi+\xi-2}
        \node[] (ula4\label) at (0.5*\x-0.5,0.5*\y+0.5) {};
        \pgfmathtruncatemacro{\label}{2*\yi+\xi-2}
        \node[] (ura4\label) at (0.5*\x+0.5,0.5*\y+0.5) {};
        \pgfmathtruncatemacro{\label}{2*\yi+\xi-2}
        \node[] (bla4\label) at (0.5*\x-0.5,0.5*\y-0.5) {};
        \pgfmathtruncatemacro{\label}{2*\yi+\xi-2}
        \node[] (bra4\label) at (0.5*\x+0.5,0.5*\y-0.5) {};
    }
}

\draw[color=red,very thick] (ula11)--(ulb11)--(ula33)--(ulb43)--(ulb44)--(ula44)--(ulb22)--(ula12)--(ula11);
\draw[color=blue,very thick] (bra11)--(brb12)--(bra43)--(brb44)--(brb43)--(bra34)--(brb21)--(bra12)--(bra11);

\end{tikzpicture}}}\hspace{5pt}
\subfigure[$\{\{2,1\},\{1,0\}\}$.]{\resizebox{0.32\textwidth}{!}{\begin{tikzpicture}[darkstyle/.style={circle,draw,fill=gray!10,very thick,minimum size=20}]
\draw[color=black, very thick] (-25*0.5,0)--(30*0.5,0);
\draw[color=black, very thick] (-25*0.5,25*0.5)--(30*0.5,25*0.5);
\draw[color=black, very thick] (-25*0.5,-30*0.5)--(30*0.5,-30*0.5);
\draw[color=black, very thick] (0,-30*0.5)--(0,25*0.5);
\draw[color=black, very thick] (-25*0.5,-30*0.5)--(-25*0.5,25*0.5);
\draw[color=black, very thick] (30*0.5,-30*0.5)--(30*0.5,25*0.5);

\foreach \x in {-20,-15,-8,-3,8,13,20,25}{
    \draw[color=gray!50, very thick] (-25*0.5,-0.5*\x)--(30*0.5,-0.5*\x);
    \draw[color=gray!50, very thick] (0.5*\x,-30*0.5)--(0.5*\x,25*0.5);
}

\foreach \x [count=\xi] in {8,20}{
    \foreach \y [count=\yi] in {20,8}{
        \filldraw [blue] (0.5*\x,0.5*\y) circle (5pt);
        \pgfmathtruncatemacro{\label}{2*\yi+\xi-2}
        \node[] (ulb1\label) at (0.5*\x-0.5,0.5*\y+0.5) {\Huge $m-b_{\label}$};
        \pgfmathtruncatemacro{\label}{2*\yi+\xi-2}
        \node[] (urb1\label) at (0.5*\x+0.5,0.5*\y+0.5) {};
        \pgfmathtruncatemacro{\label}{2*\yi+\xi-2}
        \node[] (blb1\label) at (0.5*\x-0.5,0.5*\y-0.5) {};
        \pgfmathtruncatemacro{\label}{2*\yi+\xi-2}
        \node[] (brb1\label) at (0.5*\x+0.5,0.5*\y-0.5) {};
    }
}

\foreach \x [count=\xi] in {13,25}{
    \foreach \y [count=\yi] in {15,3}{
        \filldraw [blue] (0.5*\x,0.5*\y) circle (5pt);
        \pgfmathtruncatemacro{\label}{2*\yi+\xi-2}
        \node[] (ulb2\label) at (0.5*\x-0.5,0.5*\y+0.5) {\Huge $m-b_{\label}$};
        \pgfmathtruncatemacro{\label}{2*\yi+\xi-2}
        \node[] (urb2\label) at (0.5*\x+0.5,0.5*\y+0.5) {};
        \pgfmathtruncatemacro{\label}{2*\yi+\xi-2}
        \node[] (blb2\label) at (0.5*\x-0.5,0.5*\y-0.5) {};
        \pgfmathtruncatemacro{\label}{2*\yi+\xi-2}
        \node[] (brb2\label) at (0.5*\x+0.5,0.5*\y-0.5) {};
    }
}

\foreach \x [count=\xi] in {-20,-8}{
    \foreach \y [count=\yi] in {-8,-20}{
        \filldraw [blue] (0.5*\x,0.5*\y) circle (5pt);
        \pgfmathtruncatemacro{\label}{2*\xi+\yi-2}
        \node[] (ulb3\label) at (0.5*\x-0.5,0.5*\y+0.5) {\Huge $b_{\label}$};
        \pgfmathtruncatemacro{\label}{2*\xi+\yi-2}
        \node[] (urb3\label) at (0.5*\x+0.5,0.5*\y+0.5) {};
        \pgfmathtruncatemacro{\label}{2*\xi+\yi-2}
        \node[] (blb3\label) at (0.5*\x-0.5,0.5*\y-0.5) {};
        \pgfmathtruncatemacro{\label}{2*\xi+\yi-2}
        \node[] (brb3\label) at (0.5*\x+0.5,0.5*\y-0.5) {};
    }
}

\foreach \x [count=\xi] in {-15,-3}{
    \foreach \y [count=\yi] in {-13,-25}{
        \filldraw [blue] (0.5*\x,0.5*\y) circle (5pt);
        \pgfmathtruncatemacro{\label}{2*\xi+\yi-2}
        \node[] (ulb4\label) at (0.5*\x-0.5,0.5*\y+0.5) {\Huge $b_{\label}$};
        \pgfmathtruncatemacro{\label}{2*\xi+\yi-2}
        \node[] (urb4\label) at (0.5*\x+0.5,0.5*\y+0.5) {};
        \pgfmathtruncatemacro{\label}{2*\xi+\yi-2}
        \node[] (blb4\label) at (0.5*\x-0.5,0.5*\y-0.5) {};
        \pgfmathtruncatemacro{\label}{2*\xi+\yi-2}
        \node[] (brb4\label) at (0.5*\x+0.5,0.5*\y-0.5) {};
    }
}

\foreach \x [count=\xi] in {-20,-15}{
    \foreach \y [count=\yi] in {20,15}{
        \filldraw [red] (0.5*\x,0.5*\y) circle (5pt);
        \pgfmathtruncatemacro{\label}{2*\xi+\yi-2}
        \node[] (ula1\label) at (0.5*\x-0.5,0.5*\y+0.5) {\Huge $a_{\label}$};
        \pgfmathtruncatemacro{\label}{2*\xi+\yi-2}
        \node[] (ura1\label) at (0.5*\x+0.5,0.5*\y+0.5) {};
        \pgfmathtruncatemacro{\label}{2*\xi+\yi-2}
        \node[] (bla1\label) at (0.5*\x-0.5,0.5*\y-0.5) {};
        \pgfmathtruncatemacro{\label}{2*\xi+\yi-2}
        \node[] (bra1\label) at (0.5*\x+0.5,0.5*\y-0.5) {};
    }
}

\foreach \x [count=\xi] in {-8,-3}{
    \foreach \y [count=\yi] in {8,3}{
        \filldraw [red] (0.5*\x,0.5*\y) circle (5pt);
        \pgfmathtruncatemacro{\label}{2*\xi+\yi-2}
        \node[] (ula2\label) at (0.5*\x-0.5,0.5*\y+0.5) {\Huge $a_{\label}$};
        \pgfmathtruncatemacro{\label}{2*\xi+\yi-2}
        \node[] (ura2\label) at (0.5*\x+0.5,0.5*\y+0.5) {};
        \pgfmathtruncatemacro{\label}{2*\xi+\yi-2}
        \node[] (bla2\label) at (0.5*\x-0.5,0.5*\y-0.5) {};
        \pgfmathtruncatemacro{\label}{2*\xi+\yi-2}
        \node[] (bra2\label) at (0.5*\x+0.5,0.5*\y-0.5) {};
    }
}

\foreach \x [count=\xi] in {8,13}{
    \foreach \y [count=\yi] in {-8,-13}{
        \filldraw [red] (0.5*\x,0.5*\y) circle (5pt);
        \pgfmathtruncatemacro{\label}{2*\yi+\xi-2}
        \node[] at (0.5*\x-0.5+\xi-1,0.5*\y+0.5) {\Huge $m-a_{\label}$};
        \pgfmathtruncatemacro{\label}{2*\yi+\xi-2}
        \node[] (ula3\label) at (0.5*\x-0.5,0.5*\y+0.5) {};
        \pgfmathtruncatemacro{\label}{2*\yi+\xi-2}
        \node[] (ura3\label) at (0.5*\x+0.5,0.5*\y+0.5) {};
        \pgfmathtruncatemacro{\label}{2*\yi+\xi-2}
        \node[] (bla3\label) at (0.5*\x-0.5,0.5*\y-0.5) {};
        \pgfmathtruncatemacro{\label}{2*\yi+\xi-2}
        \node[] (bra3\label) at (0.5*\x+0.5,0.5*\y-0.5) {};
    }
}

\foreach \x [count=\xi] in {20,25}{
    \foreach \y [count=\yi] in {-20,-25}{
        \filldraw [red] (0.5*\x,0.5*\y) circle (5pt);
        \pgfmathtruncatemacro{\label}{2*\yi+\xi-2}
        \node[] at (0.5*\x-0.5+\xi-1,0.5*\y+0.5) {\Huge $m-a_{\label}$};
        \pgfmathtruncatemacro{\label}{2*\yi+\xi-2}
        \node[] (ula4\label) at (0.5*\x-0.5,0.5*\y+0.5) {};
        \pgfmathtruncatemacro{\label}{2*\yi+\xi-2}
        \node[] (ura4\label) at (0.5*\x+0.5,0.5*\y+0.5) {};
        \pgfmathtruncatemacro{\label}{2*\yi+\xi-2}
        \node[] (bla4\label) at (0.5*\x-0.5,0.5*\y-0.5) {};
        \pgfmathtruncatemacro{\label}{2*\yi+\xi-2}
        \node[] (bra4\label) at (0.5*\x+0.5,0.5*\y-0.5) {};
    }
}

\draw[color=blue,very thick] (bra11)--(brb12)--(bra43)--(brb44)--(bra24)--(brb23)--(brb21)--(bra12)--(bra11);
\draw[color=red,very thick] (ula11)--(ulb11)--(ula33)--(ulb43)--(ula24)--(ulb24)--(ulb22)--(ula12)--(ula11);

\end{tikzpicture}}}

\subfigure[$\{\{1,1\},\{1,1\}\}$.]{\resizebox{0.32\textwidth}{!}{\begin{tikzpicture}[darkstyle/.style={circle,draw,fill=gray!10,very thick,minimum size=20}]
\draw[color=black, very thick] (-25*0.5,0)--(30*0.5,0);
\draw[color=black, very thick] (-25*0.5,25*0.5)--(30*0.5,25*0.5);
\draw[color=black, very thick] (-25*0.5,-30*0.5)--(30*0.5,-30*0.5);
\draw[color=black, very thick] (0,-30*0.5)--(0,25*0.5);
\draw[color=black, very thick] (-25*0.5,-30*0.5)--(-25*0.5,25*0.5);
\draw[color=black, very thick] (30*0.5,-30*0.5)--(30*0.5,25*0.5);

\foreach \x in {-20,-15,-8,-3,8,13,20,25}{
    \draw[color=gray!50, very thick] (-25*0.5,-0.5*\x)--(30*0.5,-0.5*\x);
    \draw[color=gray!50, very thick] (0.5*\x,-30*0.5)--(0.5*\x,25*0.5);
}

\foreach \x [count=\xi] in {8,20}{
    \foreach \y [count=\yi] in {20,8}{
        \filldraw [blue] (0.5*\x,0.5*\y) circle (5pt);
        \pgfmathtruncatemacro{\label}{2*\yi+\xi-2}
        \node[] (ulb1\label) at (0.5*\x-0.5,0.5*\y+0.5) {\Huge $m-b_{\label}$};
        \pgfmathtruncatemacro{\label}{2*\yi+\xi-2}
        \node[] (urb1\label) at (0.5*\x+0.5,0.5*\y+0.5) {};
        \pgfmathtruncatemacro{\label}{2*\yi+\xi-2}
        \node[] (blb1\label) at (0.5*\x-0.5,0.5*\y-0.5) {};
        \pgfmathtruncatemacro{\label}{2*\yi+\xi-2}
        \node[] (brb1\label) at (0.5*\x+0.5,0.5*\y-0.5) {};
    }
}

\foreach \x [count=\xi] in {13,25}{
    \foreach \y [count=\yi] in {15,3}{
        \filldraw [blue] (0.5*\x,0.5*\y) circle (5pt);
        \pgfmathtruncatemacro{\label}{2*\yi+\xi-2}
        \node[] (ulb2\label) at (0.5*\x-0.5,0.5*\y+0.5) {\Huge $m-b_{\label}$};
        \pgfmathtruncatemacro{\label}{2*\yi+\xi-2}
        \node[] (urb2\label) at (0.5*\x+0.5,0.5*\y+0.5) {};
        \pgfmathtruncatemacro{\label}{2*\yi+\xi-2}
        \node[] (blb2\label) at (0.5*\x-0.5,0.5*\y-0.5) {};
        \pgfmathtruncatemacro{\label}{2*\yi+\xi-2}
        \node[] (brb2\label) at (0.5*\x+0.5,0.5*\y-0.5) {};
    }
}

\foreach \x [count=\xi] in {-20,-8}{
    \foreach \y [count=\yi] in {-8,-20}{
        \filldraw [blue] (0.5*\x,0.5*\y) circle (5pt);
        \pgfmathtruncatemacro{\label}{2*\xi+\yi-2}
        \node[] (ulb3\label) at (0.5*\x-0.5,0.5*\y+0.5) {\Huge $b_{\label}$};
        \pgfmathtruncatemacro{\label}{2*\xi+\yi-2}
        \node[] (urb3\label) at (0.5*\x+0.5,0.5*\y+0.5) {};
        \pgfmathtruncatemacro{\label}{2*\xi+\yi-2}
        \node[] (blb3\label) at (0.5*\x-0.5,0.5*\y-0.5) {};
        \pgfmathtruncatemacro{\label}{2*\xi+\yi-2}
        \node[] (brb3\label) at (0.5*\x+0.5,0.5*\y-0.5) {};
    }
}

\foreach \x [count=\xi] in {-15,-3}{
    \foreach \y [count=\yi] in {-13,-25}{
        \filldraw [blue] (0.5*\x,0.5*\y) circle (5pt);
        \pgfmathtruncatemacro{\label}{2*\xi+\yi-2}
        \node[] (ulb4\label) at (0.5*\x-0.5,0.5*\y+0.5) {\Huge $b_{\label}$};
        \pgfmathtruncatemacro{\label}{2*\xi+\yi-2}
        \node[] (urb4\label) at (0.5*\x+0.5,0.5*\y+0.5) {};
        \pgfmathtruncatemacro{\label}{2*\xi+\yi-2}
        \node[] (blb4\label) at (0.5*\x-0.5,0.5*\y-0.5) {};
        \pgfmathtruncatemacro{\label}{2*\xi+\yi-2}
        \node[] (brb4\label) at (0.5*\x+0.5,0.5*\y-0.5) {};
    }
}

\foreach \x [count=\xi] in {-20,-15}{
    \foreach \y [count=\yi] in {20,15}{
        \filldraw [red] (0.5*\x,0.5*\y) circle (5pt);
        \pgfmathtruncatemacro{\label}{2*\xi+\yi-2}
        \node[] (ula1\label) at (0.5*\x-0.5,0.5*\y+0.5) {\Huge $a_{\label}$};
        \pgfmathtruncatemacro{\label}{2*\xi+\yi-2}
        \node[] (ura1\label) at (0.5*\x+0.5,0.5*\y+0.5) {};
        \pgfmathtruncatemacro{\label}{2*\xi+\yi-2}
        \node[] (bla1\label) at (0.5*\x-0.5,0.5*\y-0.5) {};
        \pgfmathtruncatemacro{\label}{2*\xi+\yi-2}
        \node[] (bra1\label) at (0.5*\x+0.5,0.5*\y-0.5) {};
    }
}

\foreach \x [count=\xi] in {-8,-3}{
    \foreach \y [count=\yi] in {8,3}{
        \filldraw [red] (0.5*\x,0.5*\y) circle (5pt);
        \pgfmathtruncatemacro{\label}{2*\xi+\yi-2}
        \node[] (ula2\label) at (0.5*\x-0.5,0.5*\y+0.5) {\Huge $a_{\label}$};
        \pgfmathtruncatemacro{\label}{2*\xi+\yi-2}
        \node[] (ura2\label) at (0.5*\x+0.5,0.5*\y+0.5) {};
        \pgfmathtruncatemacro{\label}{2*\xi+\yi-2}
        \node[] (bla2\label) at (0.5*\x-0.5,0.5*\y-0.5) {};
        \pgfmathtruncatemacro{\label}{2*\xi+\yi-2}
        \node[] (bra2\label) at (0.5*\x+0.5,0.5*\y-0.5) {};
    }
}

\foreach \x [count=\xi] in {8,13}{
    \foreach \y [count=\yi] in {-8,-13}{
        \filldraw [red] (0.5*\x,0.5*\y) circle (5pt);
        \pgfmathtruncatemacro{\label}{2*\yi+\xi-2}
        \node[] at (0.5*\x-0.5+\xi-1,0.5*\y+0.5) {\Huge $m-a_{\label}$};
        \pgfmathtruncatemacro{\label}{2*\yi+\xi-2}
        \node[] (ula3\label) at (0.5*\x-0.5,0.5*\y+0.5) {};
        \pgfmathtruncatemacro{\label}{2*\yi+\xi-2}
        \node[] (ura3\label) at (0.5*\x+0.5,0.5*\y+0.5) {};
        \pgfmathtruncatemacro{\label}{2*\yi+\xi-2}
        \node[] (bla3\label) at (0.5*\x-0.5,0.5*\y-0.5) {};
        \pgfmathtruncatemacro{\label}{2*\yi+\xi-2}
        \node[] (bra3\label) at (0.5*\x+0.5,0.5*\y-0.5) {};
    }
}

\foreach \x [count=\xi] in {20,25}{
    \foreach \y [count=\yi] in {-20,-25}{
        \filldraw [red] (0.5*\x,0.5*\y) circle (5pt);
        \pgfmathtruncatemacro{\label}{2*\yi+\xi-2}
        \node[] at (0.5*\x-0.5+\xi-1,0.5*\y+0.5) {\Huge $m-a_{\label}$};
        \pgfmathtruncatemacro{\label}{2*\yi+\xi-2}
        \node[] (ula4\label) at (0.5*\x-0.5,0.5*\y+0.5) {};
        \pgfmathtruncatemacro{\label}{2*\yi+\xi-2}
        \node[] (ura4\label) at (0.5*\x+0.5,0.5*\y+0.5) {};
        \pgfmathtruncatemacro{\label}{2*\yi+\xi-2}
        \node[] (bla4\label) at (0.5*\x-0.5,0.5*\y-0.5) {};
        \pgfmathtruncatemacro{\label}{2*\yi+\xi-2}
        \node[] (bra4\label) at (0.5*\x+0.5,0.5*\y-0.5) {};
    }
}

\draw[color=red,very thick] (ula11)--(ulb11)--(ula33)--(ulb43)--(ula24)--(ulb24)--(ula42)--(ulb32)--(ula11);
\draw[color=blue,very thick] (bra11)--(brb12)--(bra43)--(brb44)--(bra24)--(brb23)--(bra32)--(brb31)--(bra11);

\end{tikzpicture}}}\hspace{5pt}
\subfigure[$\{\{3,0\},\{1,0\}\}$.]{\resizebox{0.32\textwidth}{!}{\begin{tikzpicture}[darkstyle/.style={circle,draw,fill=gray!10,very thick,minimum size=20}]
\draw[color=black, very thick] (-25*0.5,0)--(30*0.5,0);
\draw[color=black, very thick] (-25*0.5,25*0.5)--(30*0.5,25*0.5);
\draw[color=black, very thick] (-25*0.5,-30*0.5)--(30*0.5,-30*0.5);
\draw[color=black, very thick] (0,-30*0.5)--(0,25*0.5);
\draw[color=black, very thick] (-25*0.5,-30*0.5)--(-25*0.5,25*0.5);
\draw[color=black, very thick] (30*0.5,-30*0.5)--(30*0.5,25*0.5);

\foreach \x in {-20,-15,-8,-3,8,13,20,25}{
    \draw[color=gray!50, very thick] (-25*0.5,-0.5*\x)--(30*0.5,-0.5*\x);
    \draw[color=gray!50, very thick] (0.5*\x,-30*0.5)--(0.5*\x,25*0.5);
}

\foreach \x [count=\xi] in {8,20}{
    \foreach \y [count=\yi] in {20,8}{
        \filldraw [blue] (0.5*\x,0.5*\y) circle (5pt);
        \pgfmathtruncatemacro{\label}{2*\yi+\xi-2}
        \node[] (ulb1\label) at (0.5*\x-0.5,0.5*\y+0.5) {\Huge $m-b_{\label}$};
        \pgfmathtruncatemacro{\label}{2*\yi+\xi-2}
        \node[] (urb1\label) at (0.5*\x+0.5,0.5*\y+0.5) {};
        \pgfmathtruncatemacro{\label}{2*\yi+\xi-2}
        \node[] (blb1\label) at (0.5*\x-0.5,0.5*\y-0.5) {};
        \pgfmathtruncatemacro{\label}{2*\yi+\xi-2}
        \node[] (brb1\label) at (0.5*\x+0.5,0.5*\y-0.5) {};
    }
}

\foreach \x [count=\xi] in {13,25}{
    \foreach \y [count=\yi] in {15,3}{
        \filldraw [blue] (0.5*\x,0.5*\y) circle (5pt);
        \pgfmathtruncatemacro{\label}{2*\yi+\xi-2}
        \node[] (ulb2\label) at (0.5*\x-0.5,0.5*\y+0.5) {\Huge $m-b_{\label}$};
        \pgfmathtruncatemacro{\label}{2*\yi+\xi-2}
        \node[] (urb2\label) at (0.5*\x+0.5,0.5*\y+0.5) {};
        \pgfmathtruncatemacro{\label}{2*\yi+\xi-2}
        \node[] (blb2\label) at (0.5*\x-0.5,0.5*\y-0.5) {};
        \pgfmathtruncatemacro{\label}{2*\yi+\xi-2}
        \node[] (brb2\label) at (0.5*\x+0.5,0.5*\y-0.5) {};
    }
}

\foreach \x [count=\xi] in {-20,-8}{
    \foreach \y [count=\yi] in {-8,-20}{
        \filldraw [blue] (0.5*\x,0.5*\y) circle (5pt);
        \pgfmathtruncatemacro{\label}{2*\xi+\yi-2}
        \node[] (ulb3\label) at (0.5*\x-0.5,0.5*\y+0.5) {\Huge $b_{\label}$};
        \pgfmathtruncatemacro{\label}{2*\xi+\yi-2}
        \node[] (urb3\label) at (0.5*\x+0.5,0.5*\y+0.5) {};
        \pgfmathtruncatemacro{\label}{2*\xi+\yi-2}
        \node[] (blb3\label) at (0.5*\x-0.5,0.5*\y-0.5) {};
        \pgfmathtruncatemacro{\label}{2*\xi+\yi-2}
        \node[] (brb3\label) at (0.5*\x+0.5,0.5*\y-0.5) {};
    }
}

\foreach \x [count=\xi] in {-15,-3}{
    \foreach \y [count=\yi] in {-13,-25}{
        \filldraw [blue] (0.5*\x,0.5*\y) circle (5pt);
        \pgfmathtruncatemacro{\label}{2*\xi+\yi-2}
        \node[] (ulb4\label) at (0.5*\x-0.5,0.5*\y+0.5) {\Huge $b_{\label}$};
        \pgfmathtruncatemacro{\label}{2*\xi+\yi-2}
        \node[] (urb4\label) at (0.5*\x+0.5,0.5*\y+0.5) {};
        \pgfmathtruncatemacro{\label}{2*\xi+\yi-2}
        \node[] (blb4\label) at (0.5*\x-0.5,0.5*\y-0.5) {};
        \pgfmathtruncatemacro{\label}{2*\xi+\yi-2}
        \node[] (brb4\label) at (0.5*\x+0.5,0.5*\y-0.5) {};
    }
}

\foreach \x [count=\xi] in {-20,-15}{
    \foreach \y [count=\yi] in {20,15}{
        \filldraw [red] (0.5*\x,0.5*\y) circle (5pt);
        \pgfmathtruncatemacro{\label}{2*\xi+\yi-2}
        \node[] (ula1\label) at (0.5*\x-0.5,0.5*\y+0.5) {\Huge $a_{\label}$};
        \pgfmathtruncatemacro{\label}{2*\xi+\yi-2}
        \node[] (ura1\label) at (0.5*\x+0.5,0.5*\y+0.5) {};
        \pgfmathtruncatemacro{\label}{2*\xi+\yi-2}
        \node[] (bla1\label) at (0.5*\x-0.5,0.5*\y-0.5) {};
        \pgfmathtruncatemacro{\label}{2*\xi+\yi-2}
        \node[] (bra1\label) at (0.5*\x+0.5,0.5*\y-0.5) {};
    }
}

\foreach \x [count=\xi] in {-8,-3}{
    \foreach \y [count=\yi] in {8,3}{
        \filldraw [red] (0.5*\x,0.5*\y) circle (5pt);
        \pgfmathtruncatemacro{\label}{2*\xi+\yi-2}
        \node[] (ula2\label) at (0.5*\x-0.5,0.5*\y+0.5) {\Huge $a_{\label}$};
        \pgfmathtruncatemacro{\label}{2*\xi+\yi-2}
        \node[] (ura2\label) at (0.5*\x+0.5,0.5*\y+0.5) {};
        \pgfmathtruncatemacro{\label}{2*\xi+\yi-2}
        \node[] (bla2\label) at (0.5*\x-0.5,0.5*\y-0.5) {};
        \pgfmathtruncatemacro{\label}{2*\xi+\yi-2}
        \node[] (bra2\label) at (0.5*\x+0.5,0.5*\y-0.5) {};
    }
}

\foreach \x [count=\xi] in {8,13}{
    \foreach \y [count=\yi] in {-8,-13}{
        \filldraw [red] (0.5*\x,0.5*\y) circle (5pt);
        \pgfmathtruncatemacro{\label}{2*\yi+\xi-2}
        \node[] at (0.5*\x-0.5+\xi-1,0.5*\y+0.5) {\Huge $m-a_{\label}$};
        \pgfmathtruncatemacro{\label}{2*\yi+\xi-2}
        \node[] (ula3\label) at (0.5*\x-0.5,0.5*\y+0.5) {};
        \pgfmathtruncatemacro{\label}{2*\yi+\xi-2}
        \node[] (ura3\label) at (0.5*\x+0.5,0.5*\y+0.5) {};
        \pgfmathtruncatemacro{\label}{2*\yi+\xi-2}
        \node[] (bla3\label) at (0.5*\x-0.5,0.5*\y-0.5) {};
        \pgfmathtruncatemacro{\label}{2*\yi+\xi-2}
        \node[] (bra3\label) at (0.5*\x+0.5,0.5*\y-0.5) {};
    }
}

\foreach \x [count=\xi] in {20,25}{
    \foreach \y [count=\yi] in {-20,-25}{
        \filldraw [red] (0.5*\x,0.5*\y) circle (5pt);
        \pgfmathtruncatemacro{\label}{2*\yi+\xi-2}
        \node[] at (0.5*\x-0.5+\xi-1,0.5*\y+0.5) {\Huge $m-a_{\label}$};
        \pgfmathtruncatemacro{\label}{2*\yi+\xi-2}
        \node[] (ula4\label) at (0.5*\x-0.5,0.5*\y+0.5) {};
        \pgfmathtruncatemacro{\label}{2*\yi+\xi-2}
        \node[] (ura4\label) at (0.5*\x+0.5,0.5*\y+0.5) {};
        \pgfmathtruncatemacro{\label}{2*\yi+\xi-2}
        \node[] (bla4\label) at (0.5*\x-0.5,0.5*\y-0.5) {};
        \pgfmathtruncatemacro{\label}{2*\yi+\xi-2}
        \node[] (bra4\label) at (0.5*\x+0.5,0.5*\y-0.5) {};
    }
}

\draw[color=blue,very thick] (bra12)--(bra14)--(bra13)--(brb11)--(bra31)--(brb33)--(brb34)--(brb32)--(bra12);
\draw[color=red,very thick] (ula12)--(ula14)--(ula13)--(ulb12)--(ulb14)--(ulb13)--(ula31)--(ulb31)--(ula12);

\end{tikzpicture}}}\hspace{5pt}
\subfigure[$\{\{3,1\},\{0,0\}\}$.]{\resizebox{0.32\textwidth}{!}{\begin{tikzpicture}[darkstyle/.style={circle,draw,fill=gray!10,very thick,minimum size=20}]
\draw[color=black, very thick] (-25*0.5,0)--(30*0.5,0);
\draw[color=black, very thick] (-25*0.5,25*0.5)--(30*0.5,25*0.5);
\draw[color=black, very thick] (-25*0.5,-30*0.5)--(30*0.5,-30*0.5);
\draw[color=black, very thick] (0,-30*0.5)--(0,25*0.5);
\draw[color=black, very thick] (-25*0.5,-30*0.5)--(-25*0.5,25*0.5);
\draw[color=black, very thick] (30*0.5,-30*0.5)--(30*0.5,25*0.5);

\foreach \x in {-20,-15,-8,-3,8,13,20,25}{
    \draw[color=gray!50, very thick] (-25*0.5,-0.5*\x)--(30*0.5,-0.5*\x);
    \draw[color=gray!50, very thick] (0.5*\x,-30*0.5)--(0.5*\x,25*0.5);
}

\foreach \x [count=\xi] in {8,20}{
    \foreach \y [count=\yi] in {20,8}{
        \filldraw [blue] (0.5*\x,0.5*\y) circle (5pt);
        \pgfmathtruncatemacro{\label}{2*\yi+\xi-2}
        \node[] (ulb1\label) at (0.5*\x-0.5,0.5*\y+0.5) {\Huge $m-b_{\label}$};
        \pgfmathtruncatemacro{\label}{2*\yi+\xi-2}
        \node[] (urb1\label) at (0.5*\x+0.5,0.5*\y+0.5) {};
        \pgfmathtruncatemacro{\label}{2*\yi+\xi-2}
        \node[] (blb1\label) at (0.5*\x-0.5,0.5*\y-0.5) {};
        \pgfmathtruncatemacro{\label}{2*\yi+\xi-2}
        \node[] (brb1\label) at (0.5*\x+0.5,0.5*\y-0.5) {};
    }
}

\foreach \x [count=\xi] in {13,25}{
    \foreach \y [count=\yi] in {15,3}{
        \filldraw [blue] (0.5*\x,0.5*\y) circle (5pt);
        \pgfmathtruncatemacro{\label}{2*\yi+\xi-2}
        \node[] (ulb2\label) at (0.5*\x-0.5,0.5*\y+0.5) {\Huge $m-b_{\label}$};
        \pgfmathtruncatemacro{\label}{2*\yi+\xi-2}
        \node[] (urb2\label) at (0.5*\x+0.5,0.5*\y+0.5) {};
        \pgfmathtruncatemacro{\label}{2*\yi+\xi-2}
        \node[] (blb2\label) at (0.5*\x-0.5,0.5*\y-0.5) {};
        \pgfmathtruncatemacro{\label}{2*\yi+\xi-2}
        \node[] (brb2\label) at (0.5*\x+0.5,0.5*\y-0.5) {};
    }
}

\foreach \x [count=\xi] in {-20,-8}{
    \foreach \y [count=\yi] in {-8,-20}{
        \filldraw [blue] (0.5*\x,0.5*\y) circle (5pt);
        \pgfmathtruncatemacro{\label}{2*\xi+\yi-2}
        \node[] (ulb3\label) at (0.5*\x-0.5,0.5*\y+0.5) {\Huge $b_{\label}$};
        \pgfmathtruncatemacro{\label}{2*\xi+\yi-2}
        \node[] (urb3\label) at (0.5*\x+0.5,0.5*\y+0.5) {};
        \pgfmathtruncatemacro{\label}{2*\xi+\yi-2}
        \node[] (blb3\label) at (0.5*\x-0.5,0.5*\y-0.5) {};
        \pgfmathtruncatemacro{\label}{2*\xi+\yi-2}
        \node[] (brb3\label) at (0.5*\x+0.5,0.5*\y-0.5) {};
    }
}

\foreach \x [count=\xi] in {-15,-3}{
    \foreach \y [count=\yi] in {-13,-25}{
        \filldraw [blue] (0.5*\x,0.5*\y) circle (5pt);
        \pgfmathtruncatemacro{\label}{2*\xi+\yi-2}
        \node[] (ulb4\label) at (0.5*\x-0.5,0.5*\y+0.5) {\Huge $b_{\label}$};
        \pgfmathtruncatemacro{\label}{2*\xi+\yi-2}
        \node[] (urb4\label) at (0.5*\x+0.5,0.5*\y+0.5) {};
        \pgfmathtruncatemacro{\label}{2*\xi+\yi-2}
        \node[] (blb4\label) at (0.5*\x-0.5,0.5*\y-0.5) {};
        \pgfmathtruncatemacro{\label}{2*\xi+\yi-2}
        \node[] (brb4\label) at (0.5*\x+0.5,0.5*\y-0.5) {};
    }
}

\foreach \x [count=\xi] in {-20,-15}{
    \foreach \y [count=\yi] in {20,15}{
        \filldraw [red] (0.5*\x,0.5*\y) circle (5pt);
        \pgfmathtruncatemacro{\label}{2*\xi+\yi-2}
        \node[] (ula1\label) at (0.5*\x-0.5,0.5*\y+0.5) {\Huge $a_{\label}$};
        \pgfmathtruncatemacro{\label}{2*\xi+\yi-2}
        \node[] (ura1\label) at (0.5*\x+0.5,0.5*\y+0.5) {};
        \pgfmathtruncatemacro{\label}{2*\xi+\yi-2}
        \node[] (bla1\label) at (0.5*\x-0.5,0.5*\y-0.5) {};
        \pgfmathtruncatemacro{\label}{2*\xi+\yi-2}
        \node[] (bra1\label) at (0.5*\x+0.5,0.5*\y-0.5) {};
    }
}

\foreach \x [count=\xi] in {-8,-3}{
    \foreach \y [count=\yi] in {8,3}{
        \filldraw [red] (0.5*\x,0.5*\y) circle (5pt);
        \pgfmathtruncatemacro{\label}{2*\xi+\yi-2}
        \node[] (ula2\label) at (0.5*\x-0.5,0.5*\y+0.5) {\Huge $a_{\label}$};
        \pgfmathtruncatemacro{\label}{2*\xi+\yi-2}
        \node[] (ura2\label) at (0.5*\x+0.5,0.5*\y+0.5) {};
        \pgfmathtruncatemacro{\label}{2*\xi+\yi-2}
        \node[] (bla2\label) at (0.5*\x-0.5,0.5*\y-0.5) {};
        \pgfmathtruncatemacro{\label}{2*\xi+\yi-2}
        \node[] (bra2\label) at (0.5*\x+0.5,0.5*\y-0.5) {};
    }
}

\foreach \x [count=\xi] in {8,13}{
    \foreach \y [count=\yi] in {-8,-13}{
        \filldraw [red] (0.5*\x,0.5*\y) circle (5pt);
        \pgfmathtruncatemacro{\label}{2*\yi+\xi-2}
        \node[] at (0.5*\x-0.5+\xi-1,0.5*\y+0.5) {\Huge $m-a_{\label}$};
        \pgfmathtruncatemacro{\label}{2*\yi+\xi-2}
        \node[] (ula3\label) at (0.5*\x-0.5,0.5*\y+0.5) {};
        \pgfmathtruncatemacro{\label}{2*\yi+\xi-2}
        \node[] (ura3\label) at (0.5*\x+0.5,0.5*\y+0.5) {};
        \pgfmathtruncatemacro{\label}{2*\yi+\xi-2}
        \node[] (bla3\label) at (0.5*\x-0.5,0.5*\y-0.5) {};
        \pgfmathtruncatemacro{\label}{2*\yi+\xi-2}
        \node[] (bra3\label) at (0.5*\x+0.5,0.5*\y-0.5) {};
    }
}

\foreach \x [count=\xi] in {20,25}{
    \foreach \y [count=\yi] in {-20,-25}{
        \filldraw [red] (0.5*\x,0.5*\y) circle (5pt);
        \pgfmathtruncatemacro{\label}{2*\yi+\xi-2}
        \node[] at (0.5*\x-0.5+\xi-1,0.5*\y+0.5) {\Huge $m-a_{\label}$};
        \pgfmathtruncatemacro{\label}{2*\yi+\xi-2}
        \node[] (ula4\label) at (0.5*\x-0.5,0.5*\y+0.5) {};
        \pgfmathtruncatemacro{\label}{2*\yi+\xi-2}
        \node[] (ura4\label) at (0.5*\x+0.5,0.5*\y+0.5) {};
        \pgfmathtruncatemacro{\label}{2*\yi+\xi-2}
        \node[] (bla4\label) at (0.5*\x-0.5,0.5*\y-0.5) {};
        \pgfmathtruncatemacro{\label}{2*\yi+\xi-2}
        \node[] (bra4\label) at (0.5*\x+0.5,0.5*\y-0.5) {};
    }
}

\draw[color=blue,very thick] (bra12)--(bra14)--(bra13)--(brb11)--(brb13)--(bra21)--(brb34)--(brb32)--(bra12);
\draw[color=red,very thick] (ula12)--(ula14)--(ula13)--(ulb12)--(ulb14)--(ula21)--(ulb33)--(ulb31)--(ula12);

\end{tikzpicture}}}
\caption{Different partitioning patterns of cycle-$8$ candidates resulting from a pair of cycle-$4$ candidates in  $\mathbf{P}_a$ and $\mathbf{P}_b$. The \textbf{partitioning pattern} is represented by $\{\{a,b\},\{c,d\}\}$, where each one of $a,b,c,d$ corresponds to the number of nodes in each group of $\{a_1,a_2,a_3,a_4\}$ that form a cycle-$4$ candidate. The outer bracket separates those correspond to nodes in the top left panel (i.e., $\mathbf{I}_{n_2\times n_2}\otimes\mathbf{A}$) and those correspond to  nodes in the bottom right panel (i.e., within $\mathbf{I}_{r_2\times r_2}\otimes\mathbf{A}^{\textrm{T}}$). For example, in (a), all the nodes associated with $a_i$, $i=1,\dots,4$, are in the top left panel, while each group of $\{a_1,a_2,a_3,a_4\}$ in this panel contains $2$ nodes of the cycle-$8$ candidate. Therefore, the partitioning pattern is $\{\{2,2\},\{0,0\}\}$. Similarly, the partitioning pattern of cycles in Fig.~\ref{fig: cycle 8}(a) is $\{\{2,0\},\{2,0\}\}$. The green cycle in (a) is also a cycle-$8$ candidate, but the alternating sum is not dependent on $b_i$, $i=1,\dots,4$, thus these cycles can be easily removed by removing cycle-$4$ candidates in $\mathbf{P}_a$ with an zero alternating sum.}
\label{fig: enum cycle 4-4}
\end{figure*}

Let $n_{2g,i,j}^A$ and $n_{2g,i,j}^B$ denote the number of cycle-$2g$ candidates, $g=2,3,4$, with alternating sums $i$, $j$ (modulo $L$) respectively in the partitioning matrices $\mathbf{P}_a$, $\mathbf{P}_b$. Let $n_{2g}^A$ and $n_{2g}^B$ denote the total number of cycle-$2g$ candidates in $\mathbf{P}_a$ and $\mathbf{P}_b$, $g=2,3,4$, respectively. Let $L_1$, $L_2$ be the coupling lengths of the SC-HGP codes. Each cycle-$2g$ candidate that satisfies the alternating sum condition gives rise to $L=L_1L_2$ flexible cycles-$2g$ in the resulting SC-HGP code. We define $N_{2g}$ to be the total number of flexible cycles-$2g$ divided by $L$. 

\begin{lemma}\label{lemma: count cycle} Let $e_A$, $e_B$ denote the number of non-zero entries in $\mathbf{A}$, $\mathbf{B}$, respectively. Suppose that for any cycle-$4$ candidate in $\mathbf{P}_a$ and $\mathbf{P}_b$, the alternating sum is non-zero. Then,
\begin{equation}\label{eqn: count cycles sc HGP}
    \begin{split}
        &N_6=n_{6,0,0}^A(n_2+r_2)+n_{6,0,0}^B(n_1+r_1),\\
        &N_8\\
        =&n_{8,0,0}^A(n_2+r_2)+n_{8,0,0}^B(n_1+r_1)+30(n_{6,0,0}^Ae_B\\
        &+n_{6,0,0}^Be_A)+124\Bigl(2n_{4,0,0}^An_{4,0,0}^B+\\
        &\sum_{(i,j)\neq(0,0)} (n_{4,i,j}^A+n_{4,-i,-j}^A)(n_{4,i,j}^B+n_{4,-i,-j}^B)\Bigr).
    \end{split}
\end{equation}
\end{lemma}

\begin{proof}
    The first term of each equation corresponds to the number of cycles associated with a cycle candidate with zero alternating sum in $\mathbf{P}_a$ or $\mathbf{P}_b$. A cycle candidate in $\mathbf{A}$ appears $n_2$ times in $\mathbf{I}_{n_2\times n_2}\otimes \mathbf{A}$ and $r_2$ times in $\mathbf{I}_{r_2\times r_2}\otimes \mathbf{A}^{\textrm{T}}$, hence $n_2+r_2$ times in total. Similarly, a cycle candidate in $\mathbf{B}$ appears $n_1+r_1$ times in total. 
    
    Our initial assumption excludes cycles-$6$ determined by a cycle-$4$ candidate in $\mathbf{P}_a$ ($\mathbf{P}_b$) with zero alternating sum and a nonzero entry in $\mathbf{P}_b$ ($\mathbf{P}_a$). This assumption also excludes all cycles-$8$ determined by a cycle-$4$ candidate in $\mathbf{P}_a$ ($\mathbf{P}_b$) with zero alternating sum and a node or an edge in $\mathbf{P}_b$ ($\mathbf{P}_a$), since the alternating sum of such a cycle candidate reduces to the alternating sum of the associated cycle-$4$ candidate, which is assumed to be nonzero. The green cycle in Fig.~\ref{fig: enum cycle 4-4}(a) is an example of such a cycle-$8$ candidate, its alternating sum is $b_1-(m-a_1)+(m-a_3)-b_1+b_2-(m-a_4)+(m-a_2)-b_2=a_1+a_4-a_2-a_3$. Notice that parameters associated with the edge $b_1$--$b_2$ are all cancelled out and only the alternating sum of the cycle-$4$ candidate $a_1$--$a_2$--$a_3$--$a_4$ remains in the alternating sum, which is nonzero under the assumption.

    We now count cycles-$8$ determined by a cycle-$6$ candidate in $\mathbf{P}_a$ ($\mathbf{P}_b$) and a node in $\mathbf{P}_b$ ($\mathbf{P}_a$). The alternating sum of the cycle-$6$ candidate is zero. The case where the cycle-$6$ candidate lies in $\mathbf{P}_a$ is illustrated in Fig.~\ref{fig: cycle 8}(b). The different colors of the three cycle-$8$ candidates correspond to different distributions $(i,j)$ of the six nodes associated with $a_i$, $i=1,\dots,6$ among the top left panel and the bottom right panel. Blue corresponds to the distribution $(3,3)$, red to $(4,2)$, and green to $(4,1)$. The distributions $(2,4)$ and $(1,5)$ are not illustrated. Nodes in each panel should be consecutive in the original cycle-$6$ candidate, so there $6$ possibilities for each of the $5$ distributions.
    
    Next we count the number of cycles-$8$ determined by a pair of cycle-$4$ candidates with identical or opposite alternating sums, one from each partitioning matrix. There are seven different types (shown in Fig.~\ref{fig: cycle 8}(a) and Fig.~\ref{fig: enum cycle 4-4}) distinguished by the distribution of nodes associated with $a_1,a_2,a_3,a_4$, among the four copies of the cycle-$4$ candidate $a_1$--$a_2$--$a_4$--$a_3$ in the partitioning matrix. Types are labelled by a \textbf{partitioning pattern} $\{\{a,b\},\{c,d\}\}$ where $a$,$b$ indicate the number of nodes in the two copies of $a_1,a_2,a_3,a_4$ in the top left panel, and $c$, $d$ indicate the number of nodes in the two copies of $a_1,a_2,a_3,a_4$ in the bottom right panel. 

    \begin{enumerate}
        \item Fig.~\ref{fig: cycle 8}(a), $\{\{2,0\},\{2,0\}\}$. There are $2\times 2$ ways to choose the copies, after which there are $4$ ways to separate $a_1,a_2,a_3,a_4$. Once the red nodes are chosen, the blue nodes are determined. Therefore, there are $2\times 2\times 4=16$ such cycles.
        \item Fig.~\ref{fig: enum cycle 4-4}(a), $\{\{2,2\},\{0,0\}\}$. There are $2$ ways to choose the panel, after which there are $4$ ways to separate $a_1,a_2,a_3,a_4$. Therefore, there are $2\times 4=8$ such cycles. 
        \item Fig.~\ref{fig: enum cycle 4-4}(b), $\{\{2,0\},\{1,1\}\}$. Two nodes are clustered in a group from one panel, while each group in the other panel contains one node. There are $4$ ways to choose the group containing $2$ red nodes and $4$ ways to choose the $2$ nodes in this group, after which the other two red nodes and all blue nodes are determined (note that only one of the red cycle and the blue cycle can lead to a cycle-$8$). Therefore, there are $4\times 4=16$ such cycles.
        \item Fig.~\ref{fig: enum cycle 4-4}(c), $\{\{2,1\},\{1,0\}\}$. Two nodes are clustered in a group from one panel, while two groups from different panels contain one node. There are $4$ ways to choose the group containing two red nodes, and $4$ ways to choose the $2$ nodes in this group, after which there are $2$ ways to choose the red node in the other group from the current panel. After these choices, the remaining red node and all blue nodes are determined (again, only one of the red cycle and the blue cycle can lead to a cycle-$8$). Therefore, there are $4\times 4\times 2=32$ such cycles.
        \item Fig.~\ref{fig: enum cycle 4-4}(d), $\{\{1,1\},\{1,1\}\}$. Each group has exactly one red node. There are $4$ ways to choose the red node in the top left group, after which all other nodes are determined. Therefore, there are $4$ such cycles.
        \item Fig.~\ref{fig: enum cycle 4-4}(e), $\{\{3,0\},\{1,0\}\}$. Two groups from different panels have $3$ red nodes and $1$ red node, respectively. There are $4$ ways to choose the group containing $3$ red nodes and $4$ ways to choose the three nodes, after which there are $2$ ways to choose the remaining red node. Therefore, there are $4\times 4\times 3=32$ such cycles.
        \item Fig.~\ref{fig: enum cycle 4-4}(f), $\{\{3,1\},\{0,0\}\}$. Two groups from the same panel have $3$ and $1$ red nodes, respectively. There are $4$ ways to choose the group containing the $3$ red nodes, and $4$ ways to choose the three nodes, after which the remaining nodes are determined. Therefore, there are $4\times 4=16$ such cycles.
    \end{enumerate}
    
Each pair of cycle-$4$ candidates from different partitioning matrices ($\mathbf{P}_a$ and $\mathbf{P}_b$) with identical or opposite alternating sums contributes $16+8+16+32+4+32+16=124$ ($248$ if both alternating sums are zero since each red/blue cycles pair contributes to $2$ cycles-$8$ in this case) cycles-$8$, thus completes the derivation of $N_8$. 
\end{proof}

We now optimize the partitioning matrix by extending the probabilistic framework developed in \cite{Yang2023breaking} to SC-HGP codes. We suppose that each entry in the partitioning matrix $\mathbf{P}_a$ ($\mathbf{P}_b$) is a random variable with probability distribution $\mathbb{P}\left[X^A=(i,j)\right]=p^A_{i,j}$ ($\mathbb{P}\left[X^B=(i,j)\right]=p^B_{i,j}$). We associate the probability distributions for entries of $\mathbf{P}_a$ and $\mathbf{P}_b$ with the \textbf{characteristic polynomials}
\begin{equation}
\begin{split}
    &f^A(S,T)=\sum_{i,j=0}^{m_1,m_2} p^A_{i,j} S^iT^j\textrm{ and }\\
    &f^B(S,T)=\sum_{i,j=0}^{m_1,m_2} p^B_{i,j} S^iT^j.
\end{split}
\end{equation}

Starting from \Cref{theo: expected cycle number}, we now derive \Cref{theo: expected cycle number} which expresses the expected number of $N_6$ ($N_8$) of flexible cycles-$6$ (cycles-$8$) in terms of the characteristic polynomials $f^A(S,T)$ and $f^B(S,T)$. This makes it possible to apply gradient descent (\textbf{the gradient descent (GRADE)-algorithmic optimization (AO) method}) to jointly optimize the probability distributions $p^A$ and $p^B$. The initial distribution is obtained by randomly assigning $\lfloor e_A p^A_{i,j} \rfloor$ or $\lceil e_A p^A_{i,j} \rceil$ ( $\lfloor e_B p^B_{i,j} \rfloor$ or $\lceil e_B p^B_{i,j} \rceil$) edges in $\mathbf{A}$ ($\mathbf{B}$) to the $(i,j)$-th component matrix, where $e_A$, $e_B$ denote the number of non-zero entries in $\mathbf{A}$, $\mathbf{B}$, respectively. 

\begin{theo}\label{theo: expected cycle number} Given a polynomial $f(S,T)$, let $\left[\cdot\right]_{i,j}$ denote the coefficient of $S^iT^j$. Then,
    \begin{widetext}
    \begin{equation}\label{eqn: grade obj}
        \begin{split}
        \mathbb{E}_{p^A,p^B}\left[N_6\right]=&n_{6}^A (n_2+r_2)\Bigl[\bigl(f^A(S,T)f^A(S^{-1},T^{-1})\bigr)^{3}\Bigr]_{0,0}+n_{6}^B (n_1+r_1)\Bigl[\bigl(f^B(S,T)f^B(S^{-1},T^{-1})\bigr)^{3}\Bigr]_{0,0},\\
        \mathbb{E}_{p^A,p^B}\left[N_8\right]=&n_{8}^A (n_2+r_2)\Bigl[\bigl(f^A(S,T)f^A(S^{-1},T^{-1})\bigr)^{4}\Bigr]_{0,0}+n_{8}^B (n_1+r_1)\Bigl[\bigl(f^B(S,T)f^B(S^{-1},T^{-})\bigr)^{4}\Bigr]_{0,0}\\
        &+30\Biggl(n_{6}^A e_B\Bigl[\bigl(f^A(S,T)f^A(S^{-1},T^{-1})\bigr)^{3}\Bigr]_{0,0}+n_{6}^B e_A\Bigl[\bigl(f^B(S,T)f^B(S^{-1},T^{-1})\bigr)^{3}\Bigr]_{0,0}\Biggr)\\
        &+248n_{4}^A n_{4}^B\Bigl[\bigl(f^A(S,T)f^A(S^{-1},T^{-1})\bigr)^{2}\bigl(f^B(S,T)f^B(S^{-1},T^{-1})\bigr)^{2}\Bigr]_{0,0}.
    \end{split}
    \end{equation}
    \end{widetext}
\end{theo}

\begin{proof} Let $p^A_{2g,i,j}$ ($p^B_{2g,i,j}$) represent the probability that a cycle-$2g$ candidate in $\mathbf{P}_a$ ($\mathbf{P}_b$) has alternating sum $(i,j)$. It follows from \cite{Yang2023breaking} that
   \begin{widetext}
   \begin{equation}\label{eqn: proof probability SC-QLDPC}
   p_{2g,i,j}^A=p_{2g,-i,-j}^A=\Bigl[\bigl(f^A(S,T)f^A(S^{-1},T^{-1})\bigr)^{g}\Bigr]_{i,j},\ 
   p_{2g,i,j}^B=p_{2g,-i,-j}^B=\Bigl[\bigl(f^B(S,T)f^B(S^{-1},T^{-1})\bigr)^{g}\Bigr]_{i,j}.
    \end{equation}
    \end{widetext}

Since entries in the partitioning matrix are independent, it follows from linearity of expectation that $\mathbb{E}_{p^A}\bigl[n_{2g,i,j}^A\bigr]=n_{2g}^A p_{2g,i,j}^A$ and $\mathbb{E}_{p^B}\bigl[n_{2g,i,j}^B\bigr]=n_{2g}^B p_{2g,i,j}^B$. Therefore,
\begin{widetext}
    \begin{equation} \label{eqn: proof probability SC-QLDPC 2}
        \begin{split}
        \mathbb{E}_{p^A,p^B}\bigl[N_6\bigr]=&n_{6}^A p_{6,0,0}^A(n_2+r_2)+n_{6}^B p_{6,0,0}^B(n_1+r_1),\\
        \mathbb{E}_{p^A,p^B}\bigl[N_8\bigr]=&n_{8}^A p_{8,0,0}^A(n_2+r_2)+n_{8}^B p_{8,0,0}^B(n_1+r_1)+30(n_{6}^A p_{6,0}^A e_B+n_{6}^B p_{6,0,0}^B e_A)\\
        &+124n_{4}^A n_{4}^B\Bigl(2p_{4,0,0}^Ap_{4,0,0}^B+\sum_{(i,j)\neq(0,0)}(p_{4,i,j}^A+p_{4,-i,-j}^A)(p_{4,i,j}^B+p_{4,-i,-j}^B)\Bigr)\\
        =&n_{8}^A p_{8,0}^A(n_2+r_2)+n_{8}^B p_{8,0,0}^B(n_1+r_1)+30(n_{6}^A p_{6,0}^A e_B+n_{6}^B p_{6,0,0}^B e_A)\\
        &+248n_{4}^A n_{4}^B\sum_{(i,j)\neq(0,0)}p_{4,i,j}^Ap_{4,-i,-j}^B.\\
    \end{split}
    \end{equation}
    \end{widetext}
    
The result follows from substituting (\ref{eqn: proof probability SC-QLDPC}) and (\ref{eqn: proof probability SC-QLDPC 2}) in \Cref{lemma: count cycle}.
\end{proof}

\begin{exam} Let $\mathbf{1}_{r\times n}$ be the $r\times n$ matrix with every entry equal to $1$. We now describe the GRADE algorithm for the special case where $\mathbf{A}=\mathbf{1}_{r_1\times n_1}$ and $\mathbf{B}=\mathbf{1}_{r_2\times n_2}$. It follows from \cite{Yang2023breaking} that for $g=2,3,4$, the number $n_{2g}$ of closed paths of length $2g$ in $\mathbf{1}_{r\times n}$ is given by
\begin{equation*}
\begin{split}
        n_4=&\binom{r}{2}\binom{n}{2},\ n_6=6\binom{r}{3}\binom{n}{3},\\
        n_8=&\frac{1}{2}\binom{r}{2}\binom{n}{2}+18\binom{r}{3}\binom{n}{3}+72\binom{r}{4}\binom{n}{4}\\
\end{split}
\end{equation*}
\begin{equation}
    \begin{split}
        &+3\Biggl(\binom{r}{2}\binom{n}{3}+\binom{r}{3}\binom{n}{2}\Biggr)+6\Biggl(\binom{r}{2}\binom{n}{4}\\
        &+\binom{r}{4}\binom{n}{2}\Biggr)+36\Biggl(\binom{r}{3}\binom{n}{4}+\binom{r}{4}\binom{n}{3}\Biggr).\\
    \end{split}
\end{equation}
Hence
\begin{equation}\label{eqn: obj coeff}
    \begin{split}
        w_1&=\frac{n_6}{n_4}=\frac{2}{3}(r-2)(n-2),\\
        w_2&=\frac{n_8}{n_4}=\frac{1}{2}(r^2-3r+3)(n^2-3n+3),\\
        w_3&=\frac{n+r}{n_4}=\frac{4(n+r)}{r(r-1)n(n-1)},\\
        w_4&=\frac{e}{n_4}=\frac{4}{(r-1)(n-1)}.
    \end{split}
\end{equation}

We recall the expressions for $N_6$, $N_8$ provided in \Cref{lemma: count cycle}, and we minimize the objective function $N=wN_6+N_8$.
\begin{widetext}
    \begin{equation}\label{eqn: obj function}
        \begin{split}
        \frac{\mathbb{E}_{p^A,p^B}\left[N\right]}{n_4^An_4^B}=&w\Biggl(w_1^Aw_3^B\Bigl[\bigl(f^A(S,T)f^A(S^{-1},T^{-1})\bigr)^{3}\Bigr]_{0,0}+w_1^Bw_3^A\Bigl[\bigl(f^B(S,T)f^B(S^{-1},T^{-1})\bigr)^{3}\Bigr]_{0,0}\Biggr)\\
        &+\Biggl(w_2^Aw_3^B\Bigl[\bigl(f^A(S,T)f^A(S^{-1},T^{-1})\bigr)^{4}\Bigr]_{0,0}+w_2^Bw_3^A\Bigl[\bigl(f^B(S,T)f^B(S^{-1},T^{-})\bigr)^{4}\Bigr]_{0,0}\Biggr)\\
        &+30\Biggl(w_1^Aw_4^B\Bigl[\bigl(f^A(S,T)f^A(S^{-1},T^{-1})\bigr)^{3}\Bigr]_{0,0}+w_1^Bw_4^A\Bigl[\bigl(f^B(S,T)f^B(S^{-1},T^{-1})\bigr)^{3}\Bigr]_{0,0}\Biggr)\\
        &+248\Bigl[\bigl(f^A(S,T)f^A(S^{-1},T^{-1})\bigr)^{2}\bigl(f^B(S,T)f^B(S^{-1},T^{-1})\bigr)^{2}\Bigr]_{0,0},
        \end{split}
    \end{equation}
    \end{widetext}
where $w_i^A$ and $w_i^B$ are obtained by replacing $r,n$ in (\ref{eqn: obj coeff}) with $r_1,n_1$ and $r_2,n_2$, respectively.
\end{exam}

We consider the objective function defined below, and apply gradient descent algorithm to find a locally optimal pair  $p^A$, $p^B$.
 \begin{equation}\label{eqn: obj grade}
    \begin{split}
    F(p^A,p^B)=&c_3^A\Bigl[\bigl(f^A\bar{f}^A\bigr)^{3}\Bigr]_{0,0}+c_3^B\Bigl[\bigl(f^B\bar{f}^B\bigr)^{3}\Bigr]_{0,0}\\
    &+c_4^A\Bigl[\bigl(f^A\bar{f}^A\bigr)^{4}\Bigr]_{0,0}+c_4^B\Bigl[\bigl(f^B\bar{f}^B\bigr)^{4}\Bigr]_{0,0}\\
    &+c\Bigl[\bigl(f^A\bar{f}^A\bigr)^{2}\bigl(f^B\bar{f}^B\bigr)^{2}\Bigr]_{0,0},
    \end{split}
 \end{equation}
where $c_3^A=n_6^A(wn_2+wr_2+30e_B)/(n_4^An_4^B)$, $c_3^B=n_6^B(wn_1+wr_1+30e_A)/(n_4^An_4^B)$, $c_4^A=n_8^A(n_2+r_2)/(n_4^An_4^B)$, $c_4^B=n_8^B(n_1+r_1)/(n_4^An_4^B)$, and $c=248$\footnote{In the special case where every entry of $\mathbf{A}$ and $\mathbf{B}$ is equal to $1$, $c_3^A=ww_1^Aw_3^B+30w_1^Aw_4^B$, $c_3^B=ww_1^Bw_3^A+30w_1^Bw_4^A$, $c_4^A=w_2^Aw_3^B$, $c_4^B=w_2^Bw_3^A$, and $c=248$.}. We abbreviate polynomials $f^A(S,T)$, $f^A(S^{-1},T^{-1})$ by $f^A$, $\bar{f}^A$, respectively, and $f^B(S,T)$, $f^B(S^{-1},T^{-1})$ by $f^B$, $\bar{f}^B$, respectively.

Given that densities sum up to $0$, we define the Lagrangian by
 \begin{equation}\label{eqn: lagrangian}
    \begin{split}
        &L(p^A,p^B)\\
        =&F(p^A,p^B)+c_A\Bigl(1-\sum\nolimits_{i=0}^{m_1}\sum\nolimits_{j=0}^{m_2}p^A_{i,j}\Bigr)\\
        &+c_B\Bigl(1-\sum\nolimits_{i=0}^{m_1}\sum\nolimits_{j=0}^{m_2}p^B_{i,j}\Bigr).
    \end{split}
 \end{equation}

 Then, the gradient of (\ref{eqn: lagrangian}) is obtained by
\begin{equation}\label{eqn: gradient L}
\begin{split}
&\Bigl(\nabla_{p^A}L(p^A,p^B)\Bigr)_{i,j}=\Bigl(\nabla_{p^A}F(p^A,p^B)\Bigr)_{i,j}-c_A,\\
&\Bigl(\nabla_{p^B}L(p^A,p^B)\Bigr)_{i,j}=\Bigl(\nabla_{p^B}F(p^A,p^B)\Bigr)_{i,j}-c_B,
\end{split}
\end{equation}
where
 \begin{equation}\label{eqn: gradient F}
    \begin{split}
    &\Bigl(\nabla_{p^A}F(p^A,p^B)\Bigr)_{i,j}\\
    =&6c_3^A\Bigl[\bigl(f^A\bar{f}^A\bigr)^{2}f^A\Bigr]_{i,j}+8c_4^A\Bigl[\bigl(f^A\bar{f}^A\bigr)^{3}f^A\Bigr]_{i,j}\\
    &+4c\Bigl[\bigl(f^A\bar{f}^A\bigr)f^A\bigl(f^B\bar{f}^B\bigr)^{2}\Bigr]_{i,j},\textrm{ and }\\
    &\Bigl(\nabla_{p^B}F(p^A,p^B)\Bigr)_{i,j}\\
    =&6c_3^B\Bigl[\bigl(f^B\bar{f}^B\bigr)^{2}f^B\Bigr]_{i,j}+8c_4^B\Bigl[\bigl(f^B\bar{f}^B\bigr)^{3}f^B\Bigr]_{i,j}\\
    &+4c\Bigl[\bigl(f^A\bar{f}^A\bigr)^2\bigl(f^B\bar{f}^B\bigr)f^B\Bigr]_{i,j}.
    \end{split}
 \end{equation}

 The updating law for $p^A$ is $(p^A)^{(\ell+1)}_{i,j}=(p^A)^{(\ell)}_{i,j}+\alpha \Bigl(\Bigl(\nabla_{p^A}F(p^A,p^B)\Bigr)_{i,j}-c_A\Bigr)$, for some step size $\alpha\in\mathbb{R}$. To preserve the constraints on probability densities, we let $c_A$ be the average of all elements in $\nabla_{p^A}F(p^A,p^B)$ at each step. We refer to this process as the \textbf{GRADE} step and details are shown in Algorithm~\ref{algo: GRADE}. 

\begin{algorithm}
\caption{Gradient-Descent Distributor (GRADE) for Cycle Optimization} 
\label{algo: GRADE}
\KwIn{\\
    $\mathbf{A}$, $\mathbf{B}$: component matrices of the HGP base matrix\;
    $w$: weight of each cycle-$6$ candidate (assuming that of a cycle-$8$ is $1$)\;
    $\epsilon,\alpha$: accuracy and step size of gradient descent\;
}
\KwOut{\\
    $p^A$, $p^B$: locally optimal edge distributions over $\{0,1,\dots,m_1\}\times\{0,1,\dots,m_2\}$\;
    $v_{prev}$, $v_{cur}$: the value of the objective function in (\ref{eqn: obj grade}) at the previous and the current iterations\;
    $g_A$, $g_B$: the gradient of the objective function with respect to $p^A$ and $p^B$, respectively\;
    $r_1$, $n_1$, $r_2$, $n_2$: the sizes of $\mathbf{A}$ and $\mathbf{B}$\;
    $e_A$, $e_B$: the number of nonzero elements in $\mathbf{A}$ and $\mathbf{B}$ \;
    $n_6^A$, $n_6^B$ ($n_8^A$, $n_8^B$): the number of length $6$ ($8$) closed paths in $\mathbf{A}$ and $\mathbf{B}$\;
}
\SetKwProg{Fn}{Function}{}{}
\SetKwFunction{F}{initialize}
\Fn{\F{}}{
    $c_3^A\gets n_6^A(wn_2+wr_2+30e_B)/(n_4^An_4^B)$\;
    $c_3^B\gets n_6^B(wn_1+wr_1+30e_A)/(n_4^An_4^B)$\;
    $c_4^A\gets n_8^A(n_2+r_2)/(n_4^An_4^B)$\;
    $c_4^B\gets n_8^B(n_1+r_1)/(n_4^An_4^B)$\;
    $c\gets 248$\;
    $v_{prev},v_{cur}\gets 1$ \;
    $p^A,p^B\gets \frac{1}{(m_1+1)(m_2+1)} \mathbf{1}_{(m_1+1)\times (m_2+1)}$\;
    $g^A,g^B,f^A,\bar{f}^A,f^B,\bar{f}^B\gets \mathbf{0}_{(m_1+1)\times (m_2+1)}$\;
}

\SetKwProg{While}{While}{do}{end}
\While{$|v_{prev}-v_{cur}|>\epsilon$}{
    obtain $F(p^A,p^B)$ by (\ref{eqn: obj grade})\;
    $\nabla_{p^A}F(p^A,p^B)$, $\nabla_{p^B}F(p^A,p^B)$ by (\ref{eqn: gradient F})\;
    $v_{prev}\gets v_{cur}$\;
    $v_{cur}\gets F(p^A,p^B)$\;
    $g^A\gets \nabla_{p^A}F(p^A,p^B)$\;
    $g^B\gets \nabla_{p^B}F(p^A,p^B)$\;
    $g^A\gets g^A-\textrm{mean}(g^A)$\;
    $g^B\gets g^B-\textrm{mean}(g^B)$\;
    
    \If{$v_{prev}>v_{cur}$}{ 
        $\alpha\gets \alpha/2$\;
    }
    $p^A\gets p^A-\alpha g^A/||g^A||$\;
    $p^B\gets p^B-\alpha g^B/||g^B||$\;
}
\KwRet{$p^A$, $p^B$}
\end{algorithm}

\begin{algorithm}
\caption{Cycle-Based GRADE-A Optimizer (AO)}\label{algo: AO_cycle}
\KwIn{\\
    $\mathbf{A}$, $\mathbf{B}$: Component matrices of the HGP base matrix\; \\
    $r_1$, $n_1$, $r_2$, $n_2$, $e_A$, $e_B$: Parameters of $\mathbf{A}$ and $\mathbf{B}$ in Algorithm~\ref{algo: GRADE}\; \\
    $m_1$, $m_2$, $L_1$, $L_2$: Memories and coupling lengths of the 2D-SC ensemble\; \\
    $\mathbf{E}^A$, $\mathbf{E}^B$: Edge lists of $\mathbf{A}$ and $\mathbf{B}$\; \\
    $p^A$, $p^B$: Edge distributions from Algorithm~\ref{algo: GRADE}\; \\
    $w$: Weight of each cycle-$6$ assuming cycle-$8$ has weight $1$\; \\
    $d_1$, $d_2$: Bounds on search space size\;
}
\KwOut{\\
    $\mathbf{P}_a$, $\mathbf{P}_b$: Locally optimal partitioning matrices for $\mathbf{A}$ and $\mathbf{B}$\; \\
    $\mathbf{d}$: Counting vector indicating distance to the initial matrix\;
}
\BlankLine
Obtain the lists $\mathcal{L}^A_{2g}$, $\mathcal{L}^B_{2g}$, $g=2,3,4$, of cycle-$2g$ candidates in the base matrix\;
Find $(\mathbf{u}^A, \mathbf{u}^B) = \arg \min_{\mathbf{x}, \mathbf{y}} ||\frac{1}{e_A}\mathbf{x} - p^A||_2 + ||\frac{1}{e_B}\mathbf{y} - p^B||_2$ where $\|\mathbf{x}\|_1 = e_A$, $\|\mathbf{y}\|_1 = e_B$\;
\ForEach{$0\leq i\leq m_1$, $0\leq j\leq  m_2$}{
    Randomly assign $\mathbf{u}^A[i][j]$ ($\mathbf{u}^B[i][j]$) $(i,j)$'s to $\mathbf{P}_a$ ($\mathbf{P}_b$)\;
}
Initialize $\mathbf{d} \gets \mathbf{0}$ and $\texttt{noptimal} \gets \texttt{False}$\;
\ForEach{$1\leq e_a\leq e_A$, $1\leq e_b\leq e_B$}{
    Extract $(i_a, j_a)$ and $(i_b, j_b)$ from $\mathbf{E}^A[e_a]$ and $\mathbf{E}^B[e_b]$\;
    Compute alternating sums of cycles in $\mathcal{L}_{2g}$ and count cycle statistics\;
    Compute $n_4$, $n_6$, and $n_8$ using \eqref{eqn: count cycles sc HGP}\;
    \ForEach{$(v_a, v_b) \in \{0, 1, \dots, m_1\} \times \{0, 1, \dots, m_2\}$}{
        $\mathbf{d}' \gets \mathbf{d}$; $\mathbf{d}'[v_a] \gets \mathbf{d}'[v_a] + 1$\;
        \If{$\|\mathbf{d}'\|_1 \leq d_1$ and $\|\mathbf{d}'\|_\infty \leq d_2$}{
            Temporarily update $\mathbf{P}_a$ and $\mathbf{P}_b$\;
            Recompute cycle statistics $n'_4$, $n'_6$, $n'_8$, and $n'=wn'_6 + n'_8 $\;
            \If{$n'_4 < n_4$ or $n'< n$}{
                $\texttt{noptimal} \gets \texttt{True}$\;
                Update $n_4$, $n$, $\mathbf{d}$, $\mathbf{P}_a$, $\mathbf{P}_b$\;
            }
        }
    }
}
\If{$\texttt{noptimal}$}{
    \textbf{goto} step 6\;
}
\Return $\mathbf{P}_a, \mathbf{P}_b$\;
\end{algorithm}
We initialize partitioning matrices $\mathbf{P}_a$ and $\mathbf{P}_b$ with empirical distributions close to the optimal pair $p^A$ and $p^B$. We then apply a greedy algorithm (the AO step described in Algorithm~\ref{algo: AO_cycle}) that removes short cycles through replacing some entry with a value that reduces the weighted number of short cycles, and terminating when no such improvement is possible. Note that to reduce complexity, we fix $\mathbf{P}_a$ ($\mathbf{P}_b$) and locally optimize $\mathbf{P}_b$ ($\mathbf{P}_a$) alternatively.

\section{Simulation Results}
\label{sec: simulation results}

In \Cref{subsec: construction}, we construct SC-HGP codes using the methods described in \Cref{subsec: optimization of SC-QLDPC}, and in \Cref{subsec: BP decoders} we describe a BP decoder for these QLDPC codes. In \Cref{subsec: simulation} we present simulation results which demonstrate the value of optimizing the number of short cycles. 
 
\subsection{Constructions}
\label{subsec: construction}

\Cref{subsec: general characteristic } described how an SC-HGP code is specified by a base matrix $\mathbf{A}$ of size $r_1\times n_1$, a base matrix $\mathbf{B}$ of size $r_2\times n_2$, and by the coupling parameters $m_1$, $m_2$, $L_1$, and $L_2$. The parity check matrix of the HGP base code has $r=r_1n_2 +r_2n_1$ rows and $n=r_1r_2 +n_1n_2$ columns, thus it encodes at least $k = n-r = (n_1-r_1)(n_2-r_2)$ logical qubits. Thus, we obtain a SC-HGP code that encodes at least $K= kL_1L_2 = (n_1-r_1)(n_2-r_2) L_1L_2$ logical qubits as $N= (r_1r_2+n_1n_2) L_1L_2$ physical qubits. 

We use Algorithm~\ref{algo: GRADE} and Algorithm~\ref{algo: AO_cycle} to generate two families of SC-HGP codes. The first family contains seven $\left[\left[7300, 2500\right]\right]$ codes with $L_1 = L_2 = 10$ and  $(r_1,r_2,n_1,n_2)=(3,3,8,8)$, and the second family contains seven $\left[\left[5800, 1600\right]\right]$ codes with $L_1=L_2=10$ and $(r_1,r_2,n_1,n_2)=(3,3,7,7)$. \Cref{table: count number 3 8} and \Cref{table: count number 3 7} list the number of short cycles (divided by $L_1L_2$).

Within each family, Code $2$ results from drawing partitioning matrices from the uniform distribution, and Code $4$ results from drawing partitioning matrices from the locally optimal distribution generated by the GRADE algorithm (Algorithm~\ref{algo: GRADE}), Codes $1$ and $3$, respectively, result from applying the AO algorithm (Algorithm~\ref{algo: AO_cycle}) to the partitioning matrices that define Codes $2$ and $4$. Codes $1$-$4$ all have memory $m_1=m_2=2$, Code $5$ has memory $m_1=m_2=1$, Code $6$ has memory $(m_1,m_2)=(1,2)$, and Code $7$ has memory $m_1=m_2=3$. The partitioning matrices that define Codes $5$-$7$ are generated by GRADE-AO (Algorithms $1$ and $2$ in combination). Partitioning matrices are specified by equations in the text, and it is these equation numbers that appear in \Cref{table: count number 3 8} and \Cref{table: count number 3 7}\footnote{Note that in a SC-QLDPC code with memories $(m_1,m_2)$, any $d$ in $\mathbf{P}_a$ and $\mathbf{P}_b$ is a value from $\{0,1,\dots,(m_1+1)(m_2+1)-1\}$, where $d$ refers to $(i,j)=(\lfloor d/(m_2+1) \rfloor, d\mod (m_2+1))$.}. Note that $>$ exists in the table since the counting in (\ref{eqn: count cycles sc HGP}) assumes no existence of cycles-$4$ after partitioning, while in Codes $2$ and $4$, there are cycles-$4$. Therefore, the actual numbers of cycles-$8$ in Codes $2$ and $4$ are larger than numbers counted by (\ref{eqn: count cycles sc HGP}).

\begin{widetext}
\begin{table}[h]
\centering
\caption{Information of $\left[\left[7300,\geq 2500\right]\right]$ codes}
\begin{tabular}{c|c|c|c|c|c|c|c}
\toprule
Common parameters & \multicolumn{7}{c}{$(r_1,n_1,r_2,n_2,L_1,L_2)=(3,8,3,8,10,10)$} \\ \midrule
Codes & Code 1 & Code 2 & Code 3 & Code 4 & Code 5 & Code 6 & Code 7 \\ \midrule
Minimum Distances ($\leq$) & $10$ & $6$ & $12$ & $8$ & $6$ & $8$ & $20$ \\ \midrule
Partitioning Matrices & (\ref{eqn: 3_8_uni_opt}) & (\ref{eqn: 3_8_uni_ini}) & (\ref{eqn: 3_8_gd_opt}) & (\ref{eqn: 3_8_gd_ini}) & (\ref{eqn: 3_8_gd_m1}) & (\ref{eqn: 3_8_gd_m12}) & (\ref{eqn: 3_8_gd_m3})\\ \midrule
Distributions & \multicolumn{2}{c|}{Uniform} & \multicolumn{5}{c}{GRADE}\\ \midrule
Optimized by AO & $\checkmark$ & $\times$ & $\checkmark$ & $\times$ & \multicolumn{3}{c}{$\checkmark$}\\ \midrule
$(m_1,m_2)$ & \multicolumn{4}{c|}{$(2,2)$} & $(1,1)$ & $(1,2)$ & $(3,3)$ \\ \midrule
Number of cycles 4 & 0 & 110 & 0 & 66 & 0 & 0 & 0 \\
Number of cycles 6 & 11 & 264 & 11 & 143 & 583 & 198 & 0 \\
Number of cycles 8 & 5113 & $>$48142 & 5131 & $>$23120 & 64585 & 23266 & 1144 \\ 
\bottomrule
\end{tabular}%
\label{table: count number 3 8}
\end{table}
\end{widetext}

\begin{widetext}
\begin{table}[h]
\centering
\caption{Information of $\left[\left[5800,\geq 1600\right]\right]$ codes}
\begin{tabular}{c|c|c|c|c|c|c|c}
\toprule
Common parameters & \multicolumn{7}{c}{$(r_1,n_1,r_2,n_2,L_1,L_2)=(3,7,3,7,10,10)$}  \\ \midrule
Codes & Code 1 & Code 2 & Code 3 & Code 4 & Code 5 & Code 6 & Code 7 \\ \midrule
Minimum Distances ($\leq$) & $10$ & $10$ & $14$ & $10$ & $8$ & $6$ & $16$ \\ \midrule
Partitioning Matrices & (\ref{eqn: 3_7_uni_opt}) & (\ref{eqn: 3_7_uni_ini}) & (\ref{eqn: 3_7_gd_opt}) & (\ref{eqn: 3_7_gd_ini})  & (\ref{eqn: 3_7_gd_m1}) & (\ref{eqn: 3_7_gd_m12}) & (\ref{eqn: 3_7_gd_m3}) \\ \midrule
Distributions & \multicolumn{2}{c|}{Uniform} & \multicolumn{5}{c}{GRADE}\\ \midrule
Optimized by AO & $\checkmark$ & $\times$ & $\checkmark$ & $\times$ & \multicolumn{3}{c}{$\checkmark$}\\ \midrule
$(m_1,m_2)$ & \multicolumn{4}{c|}{$(2,2)$} & $(1,1)$ & $(1,2)$ & $(3,3)$ \\ \midrule
Number of cycles 4 & 0 & 30 & 0 & 40 & 0 & 0 & 0 \\
Number of cycles 6 & 0 & 150 & 0 & 140 & 320 & 70 & 0 \\
Number of cycles 8 & 2316 & $>$17804 & 2466 & $>$18710 & 30424 & 8594 & 380 \\ 
\bottomrule
\end{tabular}%
\label{table: count number 3 7}
\end{table}
\end{widetext}

\subsection{BP Decoders}
\label{subsec: BP decoders}

The belief propagation (BP) decoder is a probabilistic decoder widely used for decoding classical LDPC codes. In algebraic decoders, the decoder fails if the number of errors is at least the minimum distance of the code. However, if different error patterns lead to identical syndromes and their probabilities are different, then it is possible to do better. A maximum likelihood (ML) decoder chooses the error pattern with the highest probability given the obtained syndromes. ML decoder performance is optimal, but can be of high complexity in large codes. The BP decoder is a low complexity approximation of the ML decoder that efficiently estimates the marginal distribution of each symbol through iterations of local updates of messages between VNs and CNs. BP decoders are known to be optimal on acyclic graphs. In a cyclic graph with large girth (the minimum number of nodes on a cycle), the local computing graph of the BP algorithm is close to a tree and can thus still perform quite well. 

The BP decoder used for QLDPC codes is shown in \Cref{fig: BP update law}. Likelihoods of binary symbols in the classical LDPC decoder are replaced by likelihoods of Pauli matrices, so that CN to VN messages become vectors
\begin{widetext}
\begin{equation}
    \label{eqn: CN-2-VN mes}
    L_{c_j\to v_i}=\Biggl(\log\frac{\mathbb{P}\left[v_i=X|c_j\right]}{\mathbb{P}\left[v_i=I|c_j\right]},\log\frac{\mathbb{P}\left[v_i=Y|c_j\right]}{\mathbb{P}\left[v_i=I|c_j\right]},\log\frac{\mathbb{P}\left[v_i=Z|c_j\right]}{\mathbb{P}\left[v_i=I|c_j\right]}\Biggr).
\end{equation}
\end{widetext}
The VN-to-CN messages are values denoted by
\begin{equation}
    \label{eqn: VN-2-CN mes}
    L_{v_i\to c_j}=\log\frac{\mathbb{P}\left[\langle v_i,S_{i,j}\rangle_s=0\right]}{\mathbb{P}\left[\langle v_i,S_{i,j}\rangle_s=1\right]}.
\end{equation}
The channel log-likelihood-ratios (LLRs) are defined by
\begin{widetext}
\begin{equation*}
    L_i^{ch}=\Biggl(\log\frac{\mathbb{P}\left[v_i=X\right]}{\mathbb{P}\left[v_i=I\right]},\log\frac{\mathbb{P}\left[v_i=Y\right]}{\mathbb{P}\left[v_i=I\right]},\log\frac{\mathbb{P}\left[v_i=Z\right]}{\mathbb{P}\left[v_i=I\right]}\Biggr).
\end{equation*}
\end{widetext}

\begin{figure}
\centering
\subfigure[CN-to-VN update.]{\resizebox{0.23\textwidth}{!}{\begin{tikzpicture}[node distance=1.5cm,
    vn/.style={draw, circle, minimum size=0.5cm, thick},
    cn/.style={draw, rectangle, minimum size=0.5cm, thick}]

\node[cn,label=right:$c_j$] (cn) at (0,0) {};

\node[vn, below of=cn, label=below:$v_i$,font=\small] (vn1) {};
\node[vn, above of=cn, xshift=-1.5cm, yshift=0.5cm, label=above:$v_k$,font=\small] (vn2) {};
\node[vn, above of=cn, xshift=-0.5cm, yshift=0.5cm, label=above:$v_l$,font=\small] (vn3) {};
\node[vn, above of=cn, xshift=0.5cm, yshift=0.5cm,label=above:$v_m$,font=\small] (vn4) {};
\node[vn, above of=cn, xshift=1.5cm, yshift=0.5cm,label=above:$v_n$,font=\small] (vn5) {};

\draw[->] (cn) -- node[below,yshift=-0.1cm,label=left:$L_{c_j\to v_i}$,font=\small] {} (vn1);
\draw[->] (vn2) -- node[above,xshift=0.2cm,yshift=-0.4cm,label=left:$L_{v_k\to c_j}$,font=\small] {} (cn);
\draw[->] (vn3) -- node[above,xshift=0.3cm,yshift=0.1cm] {} (cn);
\draw[->] (vn4) -- node[above,xshift=-0.3cm,yshift=0.1cm] {} (cn);
\draw[->] (vn5) -- node[above,xshift=-0.2cm,yshift=-0.4cm,label=right:$L_{v_t\to c_j}$,font=\small] {} (cn);
\end{tikzpicture}}}
\subfigure[VN-to-CN update.]{\resizebox{0.23\textwidth}{!}{\begin{tikzpicture}[node distance=1.5cm,    vn/.style={draw, circle, minimum size=0.5cm, thick},    cn/.style={draw, rectangle, minimum size=0.5cm, thick}]

\node[vn,label=left:$v_i$] (vn) at (0,0) {};

\node[cn, below of=vn, label=below:$c_j$,font=\small] (cn1) {};
\node[cn, above of=vn, xshift=-1.5cm, yshift=0.5cm, label=above:$c_k$,font=\small] (cn2) {};
\node[cn, above of=vn, xshift=-0.5cm, yshift=0.5cm, label=above:$c_l$,font=\small] (cn3) {};
\node[cn, above of=vn, xshift=0.5cm, yshift=0.5cm,label=above:$c_m$,font=\small] (cn4) {};
\node[cn, above of=vn, xshift=1.5cm, yshift=0.5cm,label=above:$c_t$,font=\small] (cn5) {};
\node[right of=vn, xshift=0.25cm,font=\small] (ch) {$L_{ch}$};

\draw[->] (vn) -- node[below,yshift=-0.1cm,label=right:$L_{v_i\to c_j}$,font=\small] {} (cn1);
\draw[->] (cn2) -- node[above,xshift=0.2cm,yshift=-0.4cm,label=left:$L_{c_k\to v_i}$,font=\small] {} (vn);
\draw[->] (cn3) -- node[above,xshift=0.3cm,yshift=0.1cm] {} (vn);
\draw[->] (cn4) -- node[above,xshift=-0.3cm,yshift=0.1cm] {} (vn);
\draw[->] (cn5) -- node[above,xshift=-0.2cm,yshift=-0.4cm,label=right:$L_{c_t\to v_i}$,font=\small] {} (vn);
\draw[->] (ch) -- (vn);
\end{tikzpicture}}}
\caption{Update laws in the QLDPC decoder.}
\label{fig: BP update law}
\end{figure}

We introduce functions $f:\mathbb{R}^3\times\mathcal{P}\to\mathbb{R}$ and $g:\mathbb{R}\times\mathcal{P}\to\mathbb{R}^3$ that convert between likelihoods on binary symbols and likelihoods on Pauli matrices:
\begin{equation*}
    f(l;S)=\Biggl\{\begin{array}{lr} 
    \log (1+e^{l_1})-\log(e^{l_2}+e^{l_3}), & S=X,\\ 
    \log (1+e^{l_2})-\log(e^{l_3}+e^{l_1}), & S=Y,\\ 
    \log (1+e^{l_3})-\log(e^{l_1}+e^{l_2}), & S=Z,
    \end{array}
\end{equation*}
and 
\begin{equation*}
    g(l;S)=\Biggl\{\begin{array}{lr} 
    (0,-l,-l), & S=X,\\ 
    (0,-l,0), & S=Y,\\ 
    (-l,-l,0), & S=Z.
    \end{array}
\end{equation*}

The update laws are specified by
\begin{equation*}
L_{v_i\to c_j}=f(L'_{v_i\to c_j};S_{i,j}),
\end{equation*}
where $L'_{v_i\to c_j}=L_i^{ch}+\sum\limits_{k\in\mathcal{N}_{i}\setminus\{j\}} L_{c_k\to v_i}$;
and
\begin{equation*}
L_{c_j\to v_i}=g(L'_{c_j\to v_i};S_{i,j}),
\end{equation*}
where $L_{c_j\to v_i}=g(L'_{c_j\to v_i};S_{i,j})$, 
and
\begin{equation*}
\begin{split}
&L'_{c_j\to v_i}\\
=&\prod\limits_{k\in\mathcal{N}_{j}\setminus\{i\}} \textrm{sgn}(L_{v_k\to c_j})\phi\Bigl(\sum\limits_{k\in\mathcal{N}_{j}\setminus\{i\}} \phi(|L_{v_k\to c_j}|)\Bigr),
\end{split}
\end{equation*}
and $\phi(x)=-\log(\tanh(x/2))$.

\subsection{Degenerate Errors}\label{subsec: degenerate errors}

Assume the information sent over the channel is $\mathbf{x}$, and the decoder receives a noisy version $\mathbf{y}$, making an estimation of $\mathbf{x}$ as $\tilde{\mathbf{x}}$. One major difference between quantum error correction (QEC) codes and classical error correction codes (ECC) is as follows: in classical ECC, any $\tilde{\mathbf{x}} \neq \mathbf{x}$ leads to a decoding error; whereas in quantum ECC, $\tilde{\mathbf{x}}$ can still be a correct estimate as long as $\mathbf{x}$ and $\tilde{\mathbf{x}}$ lie in the same coset of the stabilizer group of the QEC code. Thus, in quantum ECC, there is always a potential gain from adjusting $\tilde{\mathbf{x}}$ towards a codeword.

When a BP decoder fails, the estimated vector $\tilde{\mathbf{x}}$ is typically a near-codeword, close to the original $\mathbf{x}$. Due to the presence of degenerate errors in QEC codes, it is often advantageous to push $\tilde{\mathbf{x}}$ toward a nearby codeword $\mathbf{x}'$. Decoding remains successful as long as the difference $(\mathbf{x}' - \mathbf{x})$ lies within the stabilizer group. In highly degenerate codes, algorithms that efficiently push the BP decoder’s estimation back into the codespace can significantly reduce the error floor.

Several works have explored post-processing steps to achieve this goal. The Ordered Statistics Decoder (OSD) has been the most widely applied approach \cite{roffe_decoding_2020,panteleev2021degenerate}. OSD is a post-processing algorithm that runs after BP, and its core idea is to identify a subset $S$ of variable nodes (VNs) that contains a sufficient number of linearly independent columns. This subset enables the estimate to be calculated by inverting a matrix using the syndrome. The subset $S$ is typically chosen from the least reliable VNs, with reliability measured by the log-likelihood ratios (LLRs) produced by the BP decoder.

The Stabilizer Inactivation (SI) decoder \cite{du2022stabilizer} addresses cases where an $X$ ($Z$) stabilizer generator $\mathbf{g}$ has exactly half of its variable nodes (VNs) contained in an $X$ ($Z$) error $\tilde{\mathbf{x}}$. In such scenarios, $(\tilde{\mathbf{x}} + \mathbf{g})$ and $\tilde{\mathbf{x}}$ represent essentially the same error, but $(\tilde{\mathbf{x}} + \mathbf{g})$ may be a better representative of the coset than $\tilde{\mathbf{x}}$, even though BP chooses $\tilde{\mathbf{x}}$. This situation is more likely when the probabilities of $(\tilde{\mathbf{x}} + \mathbf{g})$ and $\tilde{\mathbf{x}}$ are very close. The SI decoder measures the “closeness” between $(\tilde{\mathbf{x}} + \mathbf{g})$ and $\tilde{\mathbf{x}}$ by summing the absolute values of the log-likelihood ratios (LLRs) of all VNs in the support of $\mathbf{g}$. SI recursively spans $\mathbf{g}$ and identifies the one with the smallest measure, such that flipping the support of $\mathbf{g}$ and rerunning BP may lead to a valid codeword.

Belief Propagation with Guided Decimation (BPGD) takes a different approach to sampling error patterns \cite{yao2024belief}. Rather than decoding toward the most probable codeword, BPGD aims to sample an error pattern based on the posterior probabilities. This is achieved by sequentially fixing the values of error bits according to their marginals, as estimated by BP. BPGD has demonstrated competitive performance compared to the widely used BP-OSD under both bit-flip noise and depolarizing noise.

In this manuscript, we employ the BP-OSD decoder \cite{Roffe_LDPC_Python_tools_2022}. However, running the BP-OSD decoder directly on QLDPC codes with lengths exceeding $5000$ becomes computationally prohibitive due to its non-linear complexity, making it too slow to gather data points from the error floor region. To address this, instead of simulating with BP-OSD directly, we adopt a two-step approach as follows:

\begin{enumerate}
\item We first run the original BP small-set-flip (SSF) decoder—which does not require additional BP iterations \cite{grospellier2021combining}. The reason for choosing BP-SSF decoder will be explained later. These steps yield performance curves that are only $0.3$-$0.6$ orders of magnitude away from the BP-OSD curves. We then collect a sufficient number of errors (typically $200$-$400$) that result in decoding failures from the BP-SSF decoder, allowing us to obtain an initial FER $p^{\textrm{SSF}}_e$.
\item We then pass only these errors to BP-OSD to further evaluate the number of unresolved decoding failures. From this, we obtain the ratio of unresolved decoding failures out of the total number of errors, denoted as $\alpha^{\textrm{OSD}}_e$. Consequently, the real logical error rate under BP-OSD is given by $p^{\textrm{SSF}}_e \alpha^{\textrm{OSD}}_e$.
\end{enumerate}

To illustrate the reason for applying BP-SSF decoder, we analyzed the error vectors collected from the non-binary BP decoder. We observed that these error vectors can be divided into three major types, as shown in Fig.~\ref{fig: error vectors}:

\begin{description}
\item[\bfseries Type 1:] All $X$ ($Z$) errors where some variable nodes (VNs) are connected to more unsatisfied check nodes (CNs) than satisfied CNs. These errors are likely to occur due to slow convergence, often resulting from the existence of very short cycles. By increasing the number of iterations to a very large number—though this approach is inefficient—the errors are likely to decrease. Bit flip (BF) can easily remove these errors.
\item[\bfseries Type 2:] All $X$ ($Z$) errors where a subset of VNs is completely contained in or has a large overlap with the support of a $Z$ stabilizer, with the size of this subset exceeding half the weight of the stabilizer. Flipping the stabilizer is an efficient method to address these errors.
\item[\bfseries Type 3:] All $X$ ($Z$) errors that are neither Type I nor Type II, which correspond to classical absorbing sets. These errors represent true absorbing sets or trapping sets that are unsolvable by the BP decoder and require more advanced post-processing steps, as proposed in \cite{panteleev2021degenerate,du2022stabilizer,yao2024belief}.
\end{description}

\begin{figure}
\centering 
\subfigure[Type $1$ errors.]{\includegraphics[width=0.15\textwidth]{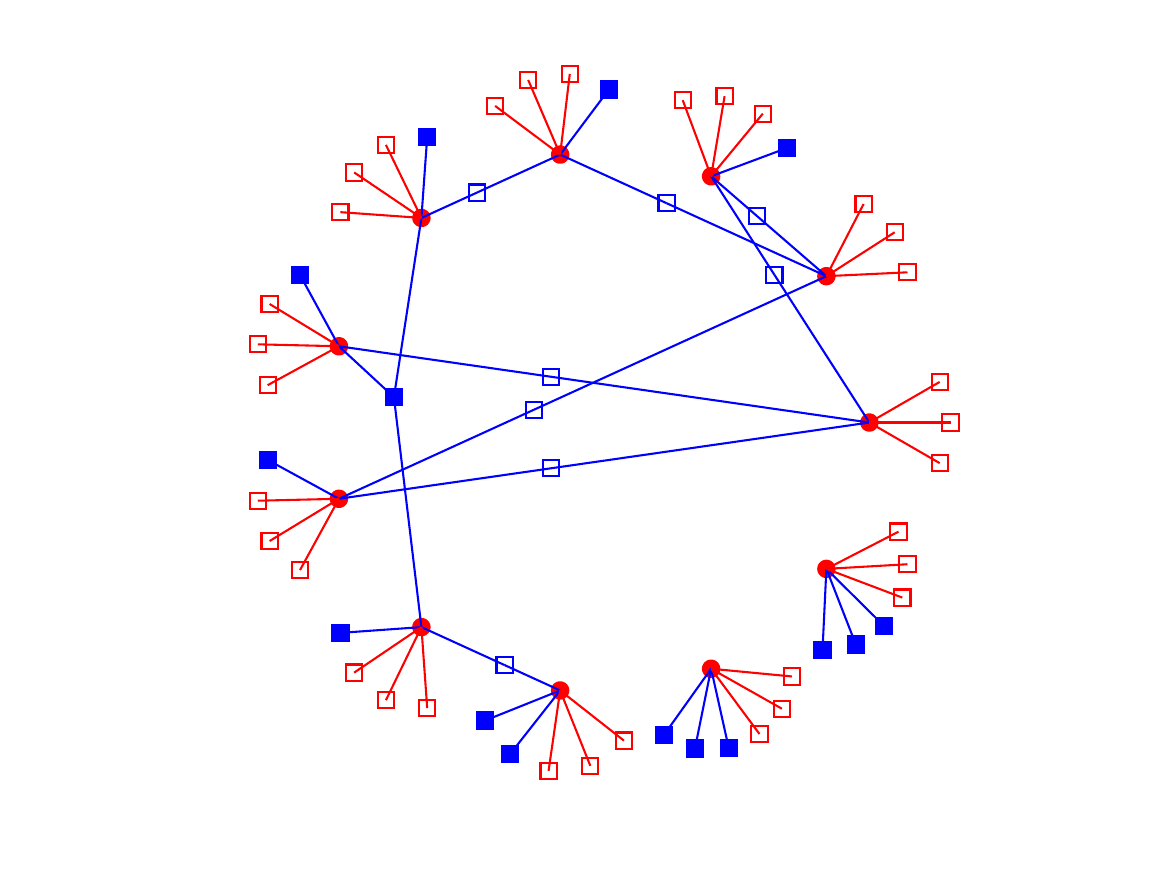}} 
\subfigure[Type $2$ errors.]{\includegraphics[width=0.15\textwidth]{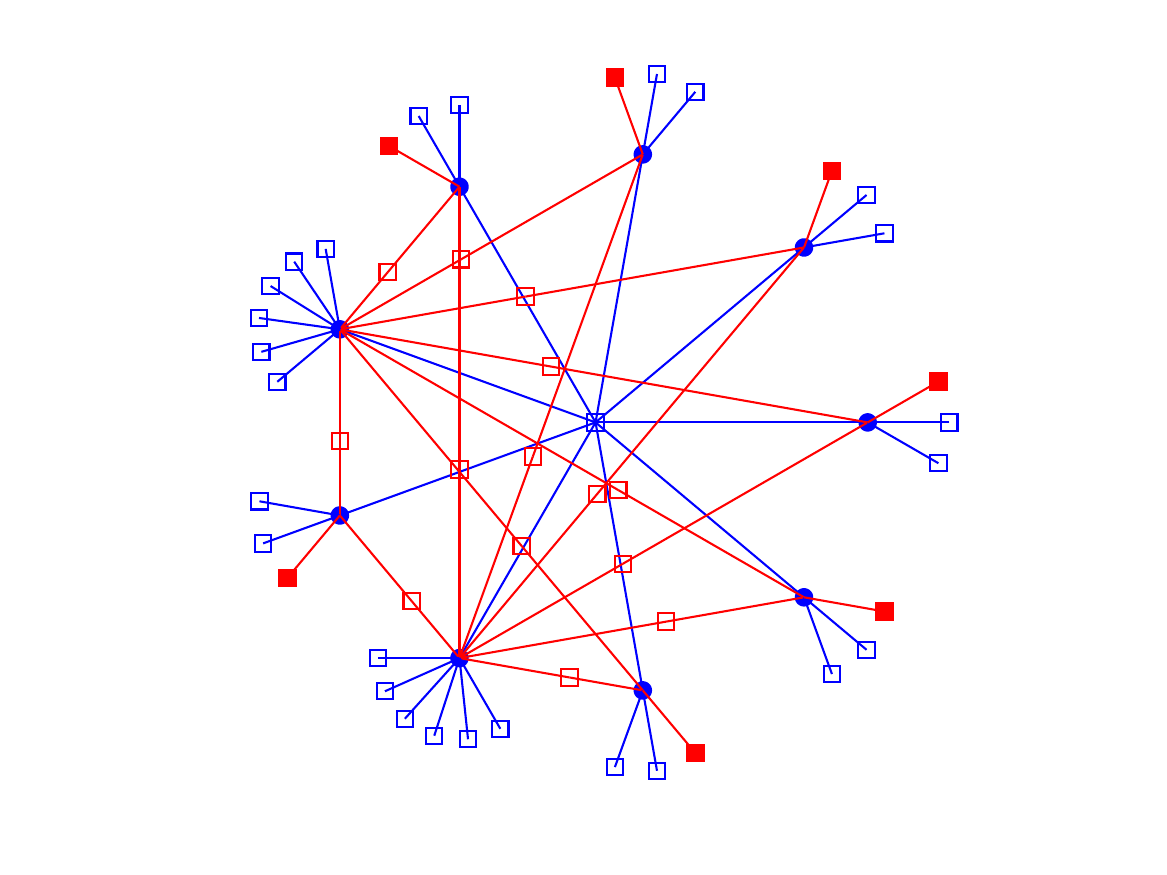}} 
\subfigure[Type $3$ errors.]{\includegraphics[width=0.15\textwidth]{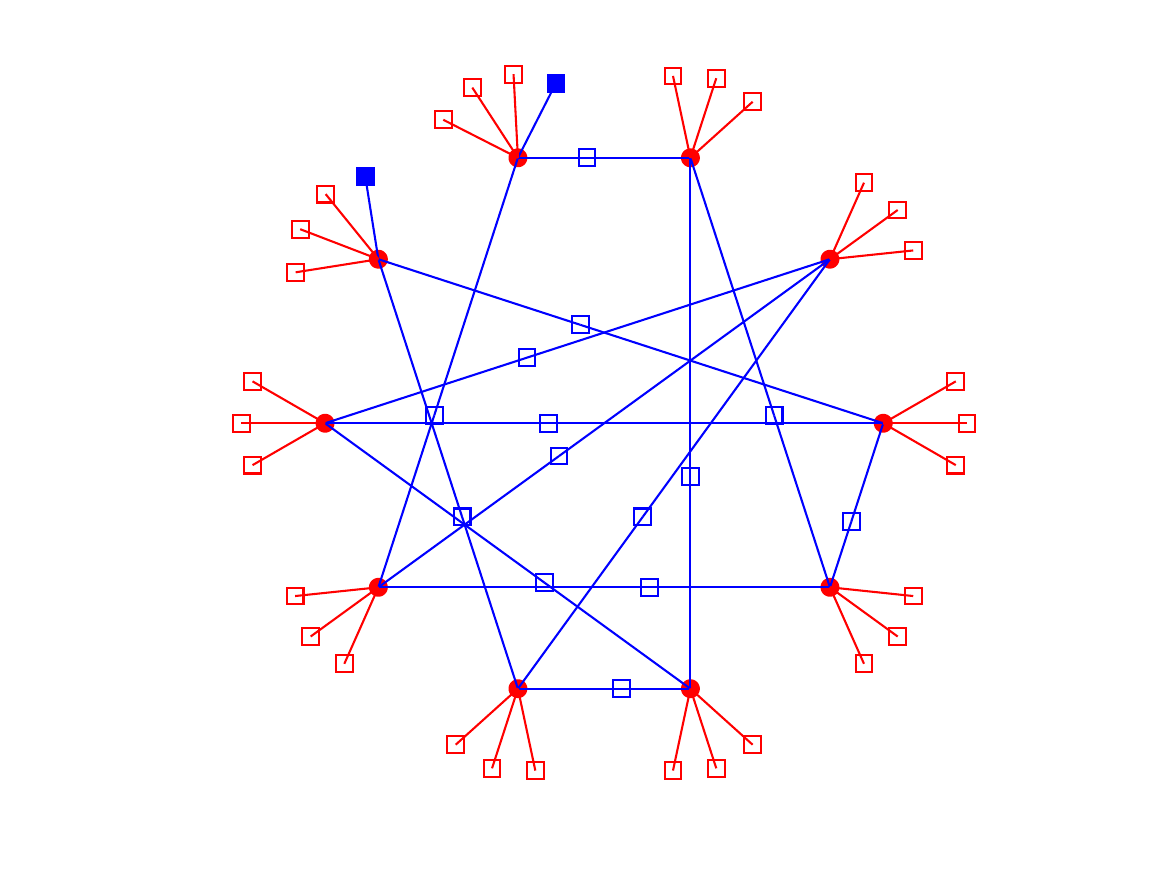}} 
\caption{Typical errors obtained through non-binary BP decoder.}
\label{fig: error vectors}
\end{figure}

Based on the observations above, running the following post-processing steps recursively effectively improves the outcomes of the BP decoder with low complexity, which is exactly the BP-SSF decoder described in \cite{grospellier2021combining}: 

\begin{description}
  \item[\bfseries Step 1:] We perform bit flip (BF) on the estimated error. In particular, we check VNs more than half of their neighbors are unsatisfied CNs, and flip one VN that is adjacent to the largest number of unsatisfied CNs. We perform this step recursively until no VNs could be flipped.
  \item[\bfseries Step 2:] We perform stabilizer flip if no VNs is flipped in Step $1$. Specifically, if there still exist a set $S$ of unsatisfied $X$($Z$) stabilizers (CNs), we check if there also exist any $Z$($X$) stabilizer that more than half of its neighbors are contained in $S$. This is motivated by the stabilizer inactivation decoder proposed in \cite{du2022stabilizer}.
  \item[\bfseries Step 3:] We perform a shortest path search on the errors if no variable nodes (VNs) are flipped in Steps 1 and 2. Specifically, we identify a shortest path that connects at least two unsatisfied check nodes (CNs) and flip all the VNs along this path. To prevent getting trapped in periodic states, we restrict the SP search to include only those VNs that have been visited at most once.
\end{description}

With the aforementioned SSF post-processing steps, we collect all the errors that lead to decoding failures. We then pass those errors through the BP-OSD software \cite{Roffe_LDPC_Python_tools_2022} and record the portion of errors that still remains in uncorrectable under BP-OSD.

\subsection{Numerical Results}
\label{subsec: simulation}
\color{black}

We consider a depolarization channel where the probability of a physical qubit being correct is $1-p$, and the probability of an $X$, $Y$, or $Z$ error is equal to $p/3$. Fig.~\ref{fig: SC-QLDPC_3_8} and Fig.~\ref{fig: SC-QLDPC_3_7} show the result of applying the decoder described in \Cref{subsec: BP decoders} to the Codes $1$ through $7$ specified in \Cref{subsec: construction}. 

\begin{figure}
  \centering
  \includegraphics[width=0.95\linewidth]{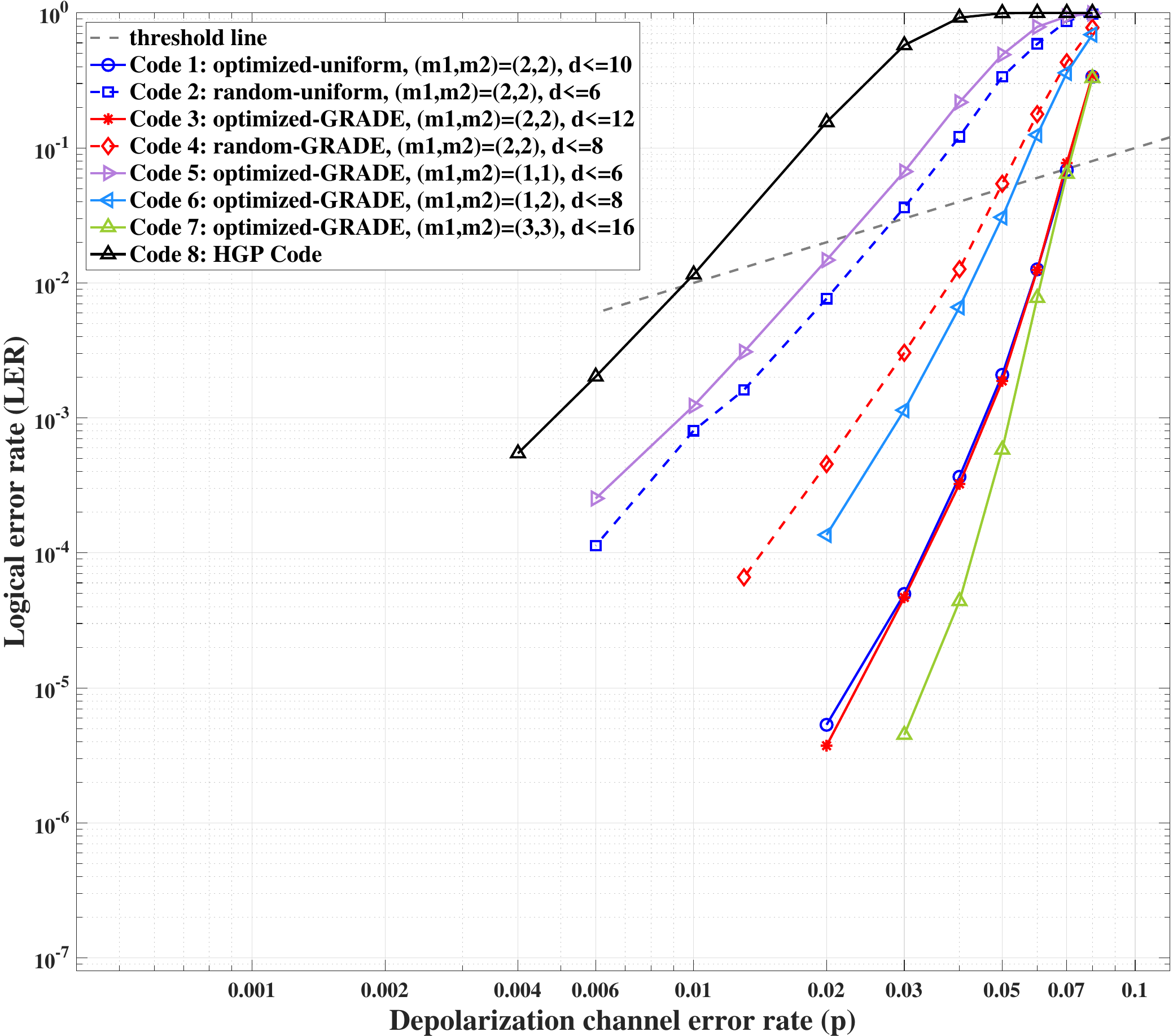}
  \caption{The Logical Error Rate (LER) for $\left[\left[7300, \geq 2500\right]\right]$ quantum low-density parity-check (QLDPC) codes, using the belief propagation with ordered statistics decoding (BP-OSD) decoder, is presented. Codes 1-7 are spatially coupled (SC) QLDPC codes, while Code 8 is a hypergraph product (HGP) code constructed from the hypergraph of a $\left[30, 80\right]$ quasi-cyclic LDPC code with itself.}
  \label{fig: SC-QLDPC_3_8}
\end{figure}

\begin{figure}
  \centering
  \includegraphics[width=0.95\linewidth]{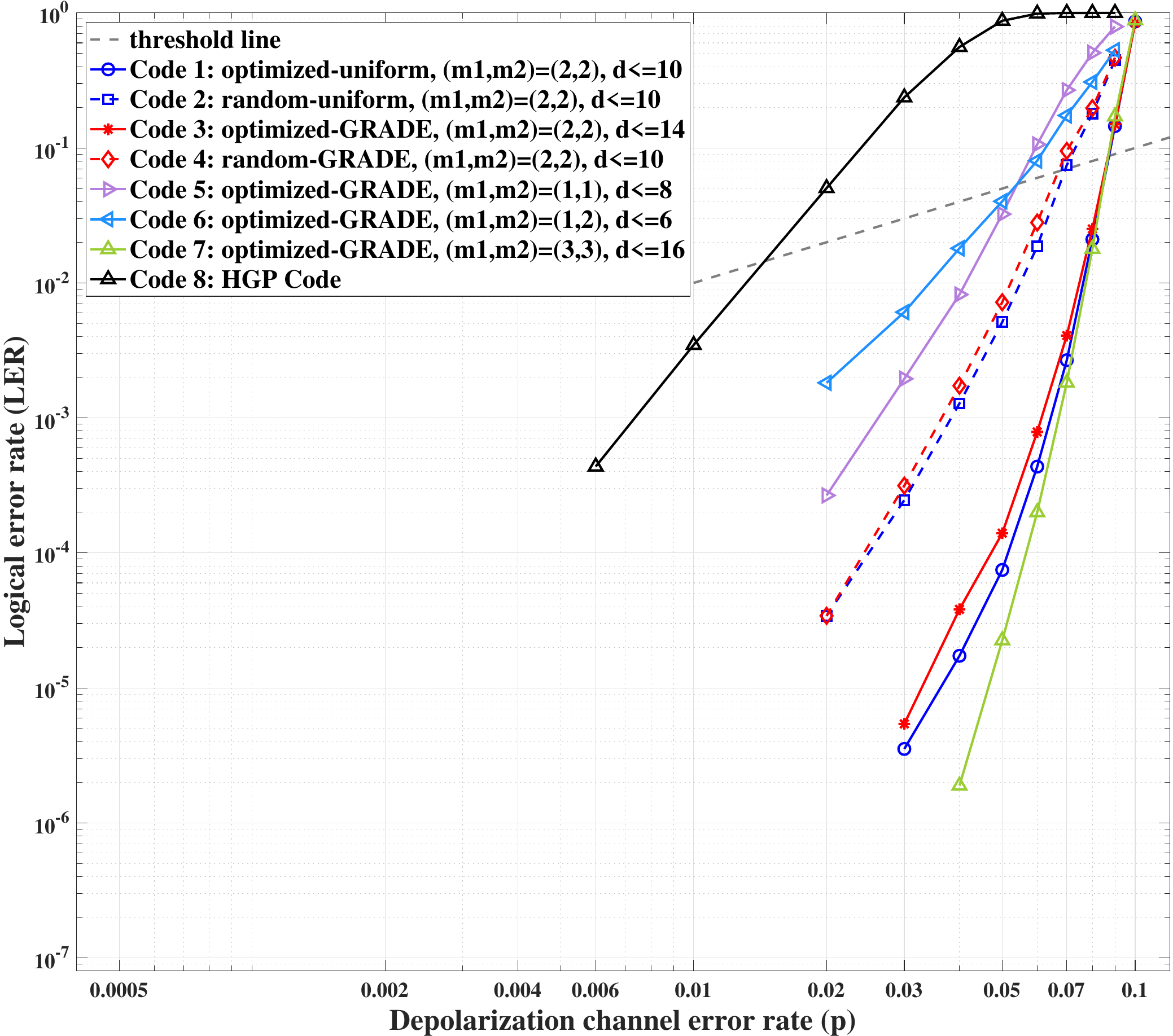}
  \caption{The Logical Error Rate (LER) for $\left[\left[5800, \geq 1600\right]\right]$ quantum low-density parity-check (QLDPC) codes, using the belief propagation with ordered statistics decoding (BP-OSD) decoder, is presented. Codes 1-7 are spatially coupled (SC) QLDPC codes, while Code 8 is a hypergraph product (HGP) code constructed from the hypergraph of a $\left[30, 70\right]$ quasi-cyclic LDPC code with itself.}
  \label{fig: SC-QLDPC_3_7}
\end{figure}

We conducted simulations in parallel on a high-performance computing cluster. The job indices served as seeds for the standard Mersenne Twister random number generator (mt19937 in C++), while Pauli errors were generated using the built-in discrete distribution function. The BP decoder estimates the error vector at the end of each iteration, checking the check node (CN) conditions, and halting when all CN conditions are satisfied or the maximum number of iterations reaches $500$.

If, after $500$ iterations, not all CN conditions are satisfied, the estimated errors undergo the SSF postprocessing step outlined in \Cref{subsec: degenerate errors}. Decoding is considered failed if the CN conditions remain unsatisfied or if they are satisfied but the output does not belong to the stabilizer group. Since SC-HGP codes are CSS codes, checking containment within the stabilizer group is straightforward after applying offline Gaussian elimination to reduce the \(X\)- and \(Z\)-components of the parity-check matrix to systematic form. We collect a sufficient number of errors that result in decoding failures in the BP-SSF decoder, and record the error rate as \(p^{\textrm{SSF}}_e\).

We then pass only the aforementioned errors to the BP-OSD algorithm to further assess the number of unresolved decoding failures \cite{roffe_decoding_2020, panteleev2021degenerate}. The ratio of unresolved decoding failures to the total number of errors is denoted as \(\alpha^{\textrm{OSD}}_e\). Consequently, the actual logical error rate under BP-OSD is given by \(p^{\textrm{SSF}}_e \alpha^{\textrm{OSD}}_e\). The BP-OSD software was obtained from \cite{Roffe_LDPC_Python_tools_2022}, with \(w=2\) and the number of iterations also set to 500.

Fig.~\ref{fig: SC-QLDPC_3_8} shows FER plots for the $\left[\left[7300,2500\right]\right]$ codes specified in \Cref{table: count number 3 8}, and Fig.~\ref{fig: SC-QLDPC_3_7} shows FER plots for the $\left[\left[5800,1600\right]\right]$ codes specified in \Cref{table: count number 3 7}. Dashed lines represent codes ($2$ and $4$) constructed from random partitioning matrices, solid lines represent codes optimized by AO, and the color of the line represents a choice of random/optimized code. In addition, we compare our constructed codes with the $\left[\left[7300,2500\right]\right]$ and $\left[\left[5800,1600\right]\right]$ hypergraph product (HGP) codes. These HGP codes are derived from two identical $\left[30,80\right]$ and $\left[30,70\right]$ binary random LDPC codes, respectively. The parity-check matrices of these random LDPC codes are modified from array-based quasi-cyclic codes, lifted from all-one matrices of sizes $3\times 7$ and $3\times 8$, ensuring the elimination of cycles-$4$ and the minimization of cycles-$6$ in the lifted codes. The lifting matrices are specified in (\ref{eqn: HGP_lifting}). The performance curves for these codes are shown in Fig.~\ref{fig: SC-QLDPC_3_8} and Fig.~\ref{fig: SC-QLDPC_3_7}, where they are labeled as Code 8.

We make the following observations:

\subsubsection{Reducing the number of short cycles leads to significant gains in the performance of SC-HGP codes} We attribute the gain of Code $1$ over Code $2$, and the gain of Code $3$ over Code $4$, to the reduction in the number of short cycles documented in \Cref{table: count number 3 8} and \Cref{table: count number 3 7}.

\begin{rem} (Detrimental objects) There is more to improving finite-length performance than reducing the number of cycles. First, the influence of a cycle depends on its neighborhood, so some cycles are more detrimental than others. Second, objects such as absorbing sets and stopping sets may be more detrimental than cycles. We attribute the performance gains reported in this paper to the fact that short cycles are the basic building blocks of these more general detrimental objects. In future work we might learn error vectors from decoding failures, and expurgate them from the Tanner graph.  
\end{rem}

\subsubsection{High performance is possible with limited memory} Figs.~\ref{fig: SC-QLDPC_3_8} and \ref{fig: SC-QLDPC_3_7} illustrate that increasing memory from $(1,1)$ / $(1, 2)$ to $(2,2)$ / $(3, 3)$ improves the FER by an order of magnitude. However, the gain associated with increasing memory from $(2, 2)$ to $(3, 3)$ is marginal (the FER curves cross in Fig. \ref{fig: SC-QLDPC_3_8} and they overlap in Fig. \ref{fig: SC-QLDPC_3_7}). As the memory increases, the number of rigid cycles dominates the number of flexible cycles, hence reducing the number of flexible cycles has a marginal impact on performance.

\begin{rem} In Fig.~\ref{fig: SC-QLDPC_3_7}, the FER curve of the memory $(1, 2)$ code crosses that of the memory $(1, 1)$ code, despite the fact that neither is close to the rigid cycle limit. It is rare that cycle optimization produces a code with small minimum distance, but that is the case here, since at least $2/3$ of the errors are convergent errors.
\end{rem}

\subsubsection{Rigid cycle limit}
We also observe that while GRADE does consistently lead to better random partitioning matrices than matrices sampled from uniform distributions (i.e., Code $4$ always perform better than Code $2$), the AO-optimized codes (Code $1$ and $3$) almost overlap in both figures, which is also consistent with the fact that their number of cycles are also close. The reason is still due to the rigid cycle limit.
 
\begin{rem} (Rigid Cycle Limit) Consider to what extent we are able to reduce the number of rigid cycles by optimizing the matrices $\mathbf{A}$ and $\mathbf{B}$ in the SC-HGP code construction. Suppose that the $i$-th VN is connected to $d_i=d^X_i+d^Y_i+d^Z_i$ CNs, where $d^X_i$, $d^Y_i$, and $d^Z_i$ respectively denote the number of $X$-type, $Y$-type and $Z$-type connections. Since short cycles have the greatest impact on performance, we focus on cycles-$4$ to show how the distribution of degrees $d_i$ limits the extent to which the number of rigid cycles can be reduced. Commutativity of the stabilizer group implies that the $i'$th VN is contained in at least $(d^X_id^Y_i+d^Y_id^Z_i+d^Z_id^X_i)$ CNs. Hence the total number of cycles-4 is                                              
\begin{equation}\label{eqn: lower bound cycle 4}
    \frac{1}{2}\sum\limits_{i=1}^N (d^X_id^Y_i+d^Y_id^Z_i+d^Z_id^X_i),
\end{equation}
where $N$ is the number of VNs. Observe that for CSS codes, (\ref{eqn: lower bound cycle 4}) reduces to $\frac{1}{2}\sum\nolimits_{i=1}^N d^Z_id^X_i$. It is necessary to optimize degree distributions in order to reduce the number (and influence) of rigid cycles. 
\end{rem} 

\subsubsection{Rigid cycles limit the extent to which GRADE and AO can improve performance} GRADE consistently provides random partitioning matrices that improve upon those sampled from a uniform distribution (for example, Code $4$ is better than Code $2$). Further optimization (AO) is less consequential when the number of cycles is close to the rigid cycle limit (for example, Codes $1$ and $3$ almost overlap). However, note that $\mathbf{A}$ and $\mathbf{B}$ we used here are pretty small and the HGP base matrix is a sparse one. Higher memories might be needed for larger/denser base matrices.

\subsubsection{Performance is not a simple function of minimum distances} Figs.~\ref{fig: SC-QLDPC_3_8} and \ref{fig: SC-QLDPC_3_7} illustrate that SC-HGP codes with smaller minimum distances do not necessarily perform worse than those with larger minimum distances. The fact that Code $1$ in Fig.~\ref{fig: SC-QLDPC_3_7}, with a minimum distance upper bounded by $10$, shows close and even slightly better performance compared to Code $3$, which has a minimum distance upper bounded by $14$, demonstrates that the minimum distance is not the only factor influencing QLDPC code performance. There are two reasons for this:

\begin{enumerate}
\item Even when using a maximum likelihood decoder, the error rate depends on the weight distribution of the codewords. Specifically, the multiplicities of the minimum weight codewords are significant, rather than just the minimum distance alone. 
\item A codeword is also an absorbing set and contains a large number of concentrated cycles. Reducing cycles can strategically lower the number of small-weight codewords.  
\end{enumerate} 

\subsubsection{Comparisons with Prior Work} Consider the memory $(2, 2)$ and $(3, 3)$ codes featured in Figs.~\ref{fig: SC-QLDPC_3_8} and \ref{fig: SC-QLDPC_3_7} with rates as high as $0.3425$ ($2500/7800$) and $0.2759$ ($1600/5800$) respectively. The rate $0.3425$ codes have thresholds around $7.0\%$, and the rate $0.2759$ codes have thresholds around $8.7\%$. The HGP codes appearing in \cite{panteleev2021degenerate} have much lower rates (less than $1/10$) and inferior thresholds. We also compared SC-HGP codes with HGP codes of the same size, as demonstrated by the performance curves of Code $8$ in each figure. It is important to note that we removed all cycles-$4$ and minimized the number of cycles-$6$ in the component matrices using a greedy algorithm to prevent potential performance degradation and ensure a fair comparison. We observe that the HGP codes exhibit inferior performance compared to the SC-HGP codes.

\begin{rem} (Generalized Bicycle (GB) Codes) The $\left[\left[126,28\right]\right]$ GB code appearing in \cite{panteleev2021degenerate} is a short code with high rate and excellent performance (slightly worse than our rate $\left[\left[7300,2500\right]\right]$ codes). It is possible that increasing the length will improve performance. Nevertheless, we have noted that GB codes can be categorized as special classes of 1D-SC-QLDPC codes with high memory. This example illustrates the importance of a well-chosen base matrix. 
\end{rem}

\section{Conclusion and Future Work}
\label{sec:conclusion}

We have described toric codes as quantum counterparts of classical two-dimensional spatially coupled (2D-SC) codes, and introduced spatially-coupled (SC) QLDPC codes as a generalization. We have used the framework of two-dimensional convolution to represent the parity check matrix of a 2D-SC code as a polynomial in two indeterminates. Our representation leads to a simple algebraic condition that is necessary and sufficient for a 2D-SC code to be a stabilizer code. It also leads to formulae relating the number of short cycles of different lengths that arise from short cycles in either component code. We have described how to use these formulae to optimize SC-hypergraph product (HGP) codes by extending methods used to optimize classical 2D-SC codes. Our simulation results show that reducing the number of short cycles leads to significant gains in the performance of SC-HGP codes. We have shown that high performance is possible with limited memory.

Future work includes implementation of windowed decoders, modification of the decoder to efficiently address the bottleneck that results from rigid cycles and development of new SC-QLDPC codes that employ more general base.

\section*{Acknowledgement}
We thank Dr. Victor Albert and Dr. Anthony Leverrier for making us aware of related literature. We thank Dr. Pavel Panteleev for specifying the relation between quasi-abelian lifted product codes and SC-HGP codes, and for bringing the more general lifted product codes defined over non-abelian groups to our attention. We thank Dr. Nithin Raveendran for providing us with estimations of minimum distances of the codes. We also thank Xinyu (Norah) Tan and Vahid Nourozi for their suggestions on improving the presentation of the paper. This work was supported in part by the National Science Foundation through QLCI-CI: Institute for Robust Quantum Simulation, and through Grant CCF-2106213.

\section*{Data Availability}
The optimization software can be found at \cite{SpatiallyCoupledQLDPCCodes}. 

\bibliography{ref}

\clearpage
\onecolumn
\appendix

\section{Proof of \Cref{exam toric code theo: 1}}
\onecolumn
\begin{proof}
\begin{equation*}
\begin{split}
&\langle \mathbf{A},\mathbf{D}\rangle=\Bigl\langle \left[\begin{array}{cc}
X & I\\
I & I
\end{array}\right], \left[\begin{array}{cc}
I & I\\
I & Z
\end{array}\right]\Bigl\rangle=\left[\begin{array}{cc}
0 & 0\\
0 & 0
\end{array}\right];\\ 
&\langle \mathbf{B},\mathbf{C}\rangle=\Bigl\langle \left[\begin{array}{cc}
X & X\\
Z & I
\end{array}\right], \left[\begin{array}{cc}
I & X\\
Z & Z
\end{array}\right]\Bigl\rangle=\left[\begin{array}{cc}
0 & 0\\
0 & 0
\end{array}\right];
\end{split}
\end{equation*}

\begin{equation*}
\begin{split}
\langle \mathbf{A},\mathbf{C}\rangle+\langle \mathbf{B},\mathbf{D}\rangle
&=\Bigl\langle \left[\begin{array}{cc}
X & I\\
I & I
\end{array}\right], \left[\begin{array}{cc}
I & X\\
Z & Z
\end{array}\right]\Bigl\rangle+\Bigl\langle \left[\begin{array}{cc}
X & X\\
Z & I
\end{array}\right], \left[\begin{array}{cc}
I & I\\
I & Z
\end{array}\right]\Bigl\rangle\\
&=\left[\begin{array}{cc}
0 & 1\\
0 & 0
\end{array}\right]+\left[\begin{array}{cc}
0 & 1\\
0 & 0
\end{array}\right]=\left[\begin{array}{cc}
0 & 0\\
0 & 0
\end{array}\right];
\end{split}
\end{equation*}
\begin{equation*}
\begin{split}
\langle \mathbf{A},\mathbf{B}\rangle+\langle \mathbf{C},\mathbf{D}\rangle
=&\Bigl\langle \left[\begin{array}{cc}
X & I\\
I & I
\end{array}\right], \left[\begin{array}{cc}
X & X\\
Z & I
\end{array}\right]\Bigl\rangle+\Bigl\langle \left[\begin{array}{cc}
I & X\\
Z & Z
\end{array}\right], \left[\begin{array}{cc}
I & I\\
I & Z
\end{array}\right]\Bigl\rangle\\
=&\left[\begin{array}{cc}
0 & 1\\
0 & 0
\end{array}\right]+\left[\begin{array}{cc}
0 & 1\\
0 & 0
\end{array}\right]
=\left[\begin{array}{cc}
0 & 0\\
0 & 0
\end{array}\right];
\end{split}
\end{equation*}
\begin{equation*}
\begin{split}
&\langle \mathbf{A},\mathbf{A}\rangle+\langle \mathbf{B},\mathbf{B}\rangle+\langle \mathbf{C},\mathbf{C}\rangle+\langle \mathbf{D},\mathbf{D}\rangle
=\langle \mathbf{B},\mathbf{B}\rangle+\langle \mathbf{C},\mathbf{C}\rangle\\
=&\Bigl\langle \left[\begin{array}{cc}
X & X\\
Z & I
\end{array}\right], \left[\begin{array}{cc}
X & X\\
Z & I
\end{array}\right]\Bigl\rangle+\Bigl\langle \left[\begin{array}{cc}
I & X\\
Z & Z
\end{array}\right], \left[\begin{array}{cc}
I & X\\
Z & Z
\end{array}\right]\Bigl\rangle\\
=&\left[\begin{array}{cc}
0 & 1\\
1 & 0
\end{array}\right]+\left[\begin{array}{cc}
0 & 1\\
1 & 0
\end{array}\right]=\left[\begin{array}{cc}
0 & 0\\
0 & 0
\end{array}\right].
\end{split}
\end{equation*}

\end{proof}

\section{Constructions}
\onecolumn
\begin{equation}\label{eqn: 3_8_uni_opt}
\mathbf{P}_a = \left[\begin{array}{cccccccc}
  2 & 1 & 3 & 8 & 4 & 8 & 3 & 3 \\
  2 & 0 & 6 & 1 & 6 & 6 & 2 & 5 \\
  6 & 8 & 2 & 0 & 4 & 1 & 5 & 7
\end{array}\right],\ 
\mathbf{P}_b = \left[\begin{array}{cccccccc}
2 & 2 & 6 & 5 & 6 & 3 & 1 & 0 \\
7 & 6 & 2 & 0 & 0 & 4 & 3 & 8 \\
6 & 0 & 0 & 7 & 5 & 8 & 5 & 3 \\
\end{array}\right].
\end{equation}

\begin{equation}\label{eqn: 3_8_uni_ini}
\mathbf{P}_a = \left[\begin{array}{cccccccc}
2 & 3 & 5 & 3 & 4 & 0 & 7 & 0 \\
3 & 6 & 1 & 6 & 6 & 0 & 2 & 8 \\
4 & 8 & 2 & 7 & 8 & 5 & 5 & 1
\end{array}\right],\ 
\mathbf{P}_b = \left[\begin{array}{cccccccc}
3 & 3 & 6 & 6 & 8 & 3 & 2 & 5 \\
1 & 4 & 2 & 0 & 0 & 4 & 7 & 8 \\
5 & 1 & 0 & 7 & 5 & 8 & 6 & 2 \\
\end{array}\right].
\end{equation}

\begin{equation}\label{eqn: 3_8_gd_opt}
\mathbf{P}_a =\left[\begin{array}{cccccccc}
8 & 1 & 2 & 2 & 6 & 8 & 6 & 7 \\
0 & 2 & 4 & 5 & 3 & 7 & 1 & 5 \\
8 & 6 & 8 & 6 & 2 & 2 & 1 & 0
\end{array}\right],\ 
\mathbf{P}_b = \left[\begin{array}{cccccccc}
8 & 8 & 2 & 6 & 6 & 0 & 3 & 0 \\
7 & 6 & 6 & 2 & 1 & 8 & 2 & 5 \\
3 & 1 & 8 & 5 & 2 & 3 & 8 & 7
\end{array}\right].
\end{equation}

\begin{equation}\label{eqn: 3_8_gd_ini}
\mathbf{P}_a = \left[\begin{array}{cccccccc}
0 & 1 & 7 & 2 & 5 & 8 & 6 & 8 \\
0 & 2 & 4 & 8 & 3 & 6 & 0 & 5 \\
3 & 6 & 8 & 6 & 2 & 2 & 1 & 0
\end{array}\right],\ 
\mathbf{P}_b = \left[\begin{array}{cccccccc}
0 & 5 & 1 & 6 & 3 & 3 & 2 & 0 \\
6 & 6 & 8 & 0 & 1 & 8 & 2 & 5 \\
0 & 4 & 8 & 8 & 2 & 6 & 2 & 7
\end{array}\right].
\end{equation}

\begin{equation}\label{eqn: 3_8_gd_m1}
\mathbf{P}_a = \left[\begin{array}{cccccccc}
1 & 3 & 0 & 2 & 0 & 0 & 1 & 2 \\
0 & 0 & 3 & 1 & 1 & 2 & 2 & 2 \\
2 & 0 & 0 & 0 & 2 & 3 & 0 & 1 \\
\end{array}\right],\ 
\mathbf{P}_b = \left[\begin{array}{cccccccc}
1 & 2 & 1 & 3 & 0 & 2 & 1 & 0 \\
1 & 1 & 0 & 0 & 1 & 0 & 2 & 2 \\
2 & 0 & 3 & 0 & 3 & 1 & 0 & 1 \\
\end{array}\right].
\end{equation}

\begin{equation}\label{eqn: 3_8_gd_m12}
\mathbf{P}_a = \left[\begin{array}{cccccccc}
2 & 4 & 3 & 1 & 2 & 3 & 2 & 3 \\
3 & 3 & 0 & 5 & 4 & 5 & 0 & 1 \\
2 & 0 & 2 & 3 & 0 & 0 & 3 & 5
\end{array}\right],\ 
\mathbf{P}_b = \left[\begin{array}{cccccccc}
5 & 1 & 3 & 0 & 2 & 0 & 3 & 5 \\
0 & 0 & 1 & 5 & 3 & 2 & 2 & 3 \\
1 & 5 & 5 & 1 & 3 & 5 & 3 & 2
\end{array}\right].
\end{equation}

\begin{equation}\label{eqn: 3_8_gd_m3}
\mathbf{P}_a = \left[\begin{array}{cccccccc}
    14 & 12 & 3 & 3 & 9 & 11 & 0 & 2 \\
    15 & 6 & 8 & 12 & 1 & 14 & 0 & 7 \\
    4 & 0 & 15 & 10 & 15 & 3 & 11 & 12
\end{array}\right],\ 
\mathbf{P}_b = \left[\begin{array}{cccccccc}
    14 & 9 & 5 & 5 & 12 & 0 & 7 & 3 \\
    12 & 3 & 6 & 2 & 1 & 15 & 0 & 8 \\
    0 & 1 & 15 & 14 & 7 & 2 & 4 & 15
\end{array}\right].
\end{equation}

\begin{equation}\label{eqn: 3_7_uni_opt}
\mathbf{P}_a = \left[\begin{array}{ccccccc}
2 & 5 & 6 & 8 & 0 & 6 & 5 \\
6 & 7 & 2 & 6 & 5 & 0 & 6 \\
7 & 0 & 2 & 1 & 7 & 8 & 5
\end{array}\right],\ 
\mathbf{P}_b = \left[\begin{array}{ccccccc}
2 & 1 & 3 & 2 & 0 & 7 & 6 \\
1 & 8 & 8 & 6 & 8 & 0 & 5 \\
6 & 6 & 2 & 3 & 0 & 5 & 1
\end{array}\right].
\end{equation}

\begin{equation}\label{eqn: 3_7_uni_ini}
\mathbf{P}_a = \left[\begin{array}{ccccccc}
8 & 3 & 1 & 7 & 7 & 2 & 1 \\
2 & 4 & 8 & 1 & 6 & 0 & 4 \\
7 & 5 & 5 & 4 & 3 & 6 & 0 \\
\end{array}\right],\ 
\mathbf{P}_b = \left[\begin{array}{ccccccc}
5 & 0 & 4 & 8 & 1 & 6 & 1 \\
6 & 8 & 1 & 4 & 5 & 4 & 3 \\
3 & 7 & 2 & 7 & 7 & 0 & 2 \\
\end{array}\right].
\end{equation}

\begin{equation}\label{eqn: 3_7_gd_opt}
\mathbf{P}_a = \left[\begin{array}{ccccccc}
6 & 2 & 5 & 1 & 6 & 0 & 8 \\
4 & 3 & 8 & 0 & 1 & 2 & 2 \\
2 & 6 & 0 & 8 & 1 & 4 & 1 
\end{array}\right],\ 
\mathbf{P}_b = \left[\begin{array}{ccccccc}
6 & 1 & 0 & 6 & 7 & 5 & 2 \\
8 & 3 & 0 & 2 & 2 & 4 & 6 \\
1 & 8 & 3 & 6 & 8 & 6 & 0
\end{array}\right]. 
\end{equation}

\begin{equation}\label{eqn: 3_7_gd_ini}
\mathbf{P}_a = \left[\begin{array}{ccccccc}
6 & 2 & 3 & 0 & 6 & 7 & 0 \\
2 & 6 & 8 & 5 & 6 & 8 & 7 \\
2 & 0 & 0 & 8 & 1 & 4 & 1
\end{array}\right],\ 
\mathbf{P}_b = \left[\begin{array}{ccccccc}
2 & 5 & 7 & 6 & 7 & 2 & 2 \\
8 & 0 & 0 & 1 & 6 & 4 & 6 \\
1 & 8 & 3 & 6 & 8 & 0 & 0
\end{array}\right].  
\end{equation}

\begin{equation}\label{eqn: 3_7_gd_m1}
\mathbf{P}_a = \left[\begin{array}{ccccccc}
0 & 3 & 1 & 3 & 2 & 1 & 2 \\
1 & 0 & 2 & 3 & 1 & 3 & 0 \\
2 & 3 & 0 & 0 & 3 & 2 & 1 \\
\end{array}\right],\ 
\mathbf{P}_b = \left[\begin{array}{ccccccc}
1 & 3 & 3 & 0 & 0 & 2 & 1 \\
2 & 1 & 0 & 2 & 3 & 3 & 1 \\
3 & 0 & 2 & 1 & 3 & 1 & 2 \\
\end{array}\right].  
\end{equation}

\begin{equation}\label{eqn: 3_7_gd_m12}
\mathbf{P}_a = \left[\begin{array}{ccccccc}
0 & 2 & 0 & 3 & 2 & 5 & 3 \\
5 & 0 & 3 & 4 & 1 & 2 & 5 \\
0 & 4 & 2 & 2 & 3 & 3 & 1 \\
\end{array}\right],\ 
\mathbf{P}_b = \left[\begin{array}{ccccccc}
2 & 4 & 5 & 2 & 0 & 5 & 0 \\
3 & 2 & 2 & 4 & 2 & 3 & 5 \\
4 & 3 & 0 & 0 & 5 & 2 & 4
\end{array}\right].  
\end{equation}

\begin{equation}\label{eqn: 3_7_gd_m3}
\mathbf{P}_a = \left[\begin{array}{ccccccc}
4 & 10 & 3 & 12 & 13 & 1 & 5 \\
7 & 8 & 8 & 11 & 7 & 11 & 3 \\
1 & 6 & 12 & 3 & 2 & 0 & 15 \\
\end{array}\right],\ 
\mathbf{P}_b = \left[\begin{array}{ccccccc}
4 & 11 & 13 & 12 & 7 & 3 & 12 \\
3 & 13 & 15 & 10 & 6 & 1 & 0 \\
4 & 1 & 5 & 7 & 14 & 15 & 3 \\
\end{array}\right].  
\end{equation}

\begin{equation}\label{eqn: HGP_lifting}
\mathbf{L}_1 = \left[\begin{array}{ccccccc}
0 & 0 & 5 & 0 & 0 & 0 & 0 \\
0 & 1 & 2 & 3 & 4 & 5 & 6 \\
0 & 2 & 4 & 6 & 8 & 1 & 3 \\
\end{array}\right], \mathbf{L}_2 = \left[\begin{array}{cccccccc}
0 & 0 & 0 & 0 & 0 & 0 & 0 & 0 \\
0 & 1 & 2 & 3 & 4 & 5 & 6 & 7 \\
0 & 2 & 4 & 6 & 9 & 1 & 3 & 5 \\
\end{array}\right].
\end{equation}

\end{document}